\documentclass{LMCS}

\def\dOi{10(2:11)2014}
\lmcsheading%
{\dOi}
{1--50}
{}
{}
{Nov.~19, 2012}
{Jun.~18, 2014}
{}

\ACMCCS{[{\bf  Security and privacy}]:  Formal methods and theory of
  security---Formal security models}


\usepackage[latin1]{inputenc}

\usepackage{amssymb}
\usepackage{proof}
\usepackage{url}
\usepackage{hyperref}

\newtheorem{corollary}{Corollary}

\newcommand{\size}{\mathit{size}}

\newcommand{\enc}{\mathsf{encs}}
\newcommand{\enca}{\mathsf{enca}}
\newcommand{\sign}{\mathsf{sign}}
\newcommand{\pub}{\mathsf{pub}}
\newcommand{\priv}{\mathsf{priv}}
\newcommand{\shk}{\mathsf{shk}}
\newcommand{\h}{\mathsf{h}}
\newcommand{\f}{\mathsf{f}}

\newcommand{\Nces}{\mathcal{N}}
\newcommand{\Xvars}{\mathcal{X}}
\newcommand{\Yvars}{\mathcal{Y}}

\newcommand{\vars}{\mathit{vars}}
\newcommand{\names}{\mathit{names}}

\newcommand{\true}{\mathsf{true}}
\newcommand{\false}{\mathsf{false}}
\newcommand{\learn}{\mathsf{learn}}

\newcommand{\Comp}{\mathsf{C}}
\newcommand{\Nocomp}{\mathsf{NC}}
\newcommand{\length}{\mathsf{length}}

\newcommand{\Pred}{{\mathcal P}}
\newcommand{\Prd}{\mathsf{P}}
\newcommand{\Secret}{\mathsf{Secret}}
\newcommand{\Start}{\mathsf{Start}}
\newcommand{\End}{\mathsf{End}}
\newcommand{\Sec}{\mathsf{S}}
\newcommand{\Aliv}{\mathsf{A}}
\newcommand{\Weak}{\mathsf{WA}}

\newcommand{\Q}{\mathsf{Q}}

\newcommand{\Tags}{\mathsf{Tags}}

\newcommand{\agent}{\mathit{agents}}

\newcommand{\lgAtom}{\mathit{lgKeys}}

\newcommand{\st}{\mathsf{St}}
\newcommand{\encst}{\mathsf{CryptSt}}

\newcommand{\dom}{\operatorname{dom}}

\newcommand{\mgu}{\operatorname{mgu}}
\newcommand{\pair}[2]{\langle{#1},{#2}\rangle}

\newcommand{\Roles}{\mathsf{Roles}}    

\def\init{\mathsf{init}}

\newcommand{\sce}{\mathsf{sc}}
\newcommand{\rcv}{\mathsf{rcv}}
\newcommand{\snd}{\mathsf{snd}}

\newcommand{\N}{\mathbb{N}}

\newcommand{\plaintext}{\mathit{plaintext}}

\newcommand{\Kcal}{\mathcal{K}}
\newcommand{\A}{\mathcal{A}}

\newcommand{\comp}{\mathsf{comp}}

\newcommand{\montag}{\mathsf{tag}}

\newcommand{\evt}{\mathsf{e}}
\newcommand{\seqevt}{\mathsf{seq}}
\newcommand{\tr}{\mathsf{tr}}
\newcommand{\exec}{\mathsf{exec}}
\newcommand{\K}{\mathsf{K}}

\newcommand{\W}{\mathsf{Ws}}
\newcommand{\Witness}{\mathsf{Ws}}

\newcommand{\sid}{\mathit{sid}}
\newcommand{\sess}{\mathit{sess}}

\newcommand{\Logic}{\mathcal{L}}

\newcommand{\tagt}{\mathsf{ExpectedTag}}
\newcommand{\tagst}{\mathsf{ExpectedTags}}
\newcommand{\tagh}{\mathsf{HeadTag}}
\newcommand{\sameTagas}{{\mathsf{sameTagAs}}}
\newcommand{\stal}{\mathsf{St_{alien}}}
\newcommand{\Annotate}{\mathsf{annotate}}

\title{Dynamic Tags for Security Protocols}

\author[M.~Arapinis]{Myrto Arapinis\rsuper a}
\address{{\lsuper a}School of Informatics, University of Edinburgh, UK}
\email{marapini@inf.ed.ac.uk}

\author[S.~Delaune]{Stéphanie Delaune\rsuper b}
\address{{\lsuper b}LSV, CNRS \& ENS Cachan, France}
\email{delaune@lsv.ens-cachan.fr}

\author[S.~Kremer]{Steve Kremer\rsuper c}
\address{{\lsuper c}Inria Nancy - Grand Est \& LORIA, France}
\email{Steve.Kremer@inria.fr}

\begin{document}
\maketitle

\begin{abstract}
  The design and verification of cryptographic protocols is a
  notoriously difficult task, even in symbolic models which take an
  abstract view of cryptography. This is mainly due to the fact that
  protocols may interact with an arbitrary attacker which yields a
  verification problem that has several sources of unboundedness (size
  of messages, number of sessions, {etc.}).

  In this paper, we characterize a class of protocols for which
  deciding security for an unbounded number of sessions is
  decidable. More precisely, we present a simple transformation which
  maps a protocol that is secure for a bounded number of protocol
  sessions (a decidable problem) to a protocol that is secure for an
  unbounded number of sessions. The precise number of sessions that
  need to be considered is a function of the security property and we
  show that for several classical security properties a single session
  is sufficient. {Therefore, in many cases our results yields a design
    strategy for security protocols:} (i) design a protocol intended
  to be secure for a {single session}; and (ii) apply our
  transformation to obtain a protocol which is secure for an unbounded
  number of sessions.

\end{abstract}

\section{Introduction}
\label{sec:intro}

Security protocols are distributed programs which aim at guaranteeing
properties such as confidentiality of data, authentication of
participants, etc. The security of these protocols relies on the one
hand on the security of cryptographic primitives, e.g. encryption and
digital signatures, and on the other hand on the concurrency-related
aspects of the protocols themselves. History has shown that even if
cryptography is supposed to be perfect, such as in the classical
Dolev-Yao model~\cite{dolev81security}, the correct design of security
protocols is notoriously error-prone.  See for
instance~\cite{ClarkJacob97} for an early survey on attacks.  These
difficulties come mainly from two sources of unboundedness: a protocol
may be executed several times (we need to consider several protocol
\emph{sessions}) and the attacker is allowed to build messages of
unbounded size.  Indeed, secrecy is known to be undecidable when an
unbounded number of sessions is allowed, even if the message size is
bounded~\cite{durgin99undecidability}. However, when the number of
sessions is bounded, and even without assuming a bounded message size,
the problem becomes co-NP-complete~\cite{RT03TCS}.  Moreover, special
purpose verification tools (e.g.~\cite{avispa2005}) exist which are
highly efficient when the number of sessions is small.

In this paper we propose a protocol transformation which maps a
protocol that is secure for a bounded number of sessions to a protocol
that is secure for an unbounded number of sessions. The exact number
of sessions that need to be considered depends on the security
property under study. We express security properties in a temporal
logic with past similar to the logics
of~\cite{TheseCorin06,CD-fmsd09}. This logic is expressive enough to
model security properties such as secrecy and {several
  flavors of non-injective authentication properties}. As we will see
for these classical security properties verifying a single session
will be sufficient {and our result provides a strategy to
  design secure protocols:} (i)
  design a protocol intended to be secure for a
    {single session}; and
  (ii) apply our transformation and obtain a protocol which is secure
    for an unbounded number of sessions.

\paragraph*{\bf Our transformation.} 
Suppose that~$\Pi$ is a protocol between~$k$ participants $A_1,
\ldots, A_k$. Our transformation adds to~$\Pi$ a preamble in which
each participant sends a freshly generated nonce~$N_i$ together with
his identity to all other participants. This allows each participant
to compute a dynamic, {session-dependent} \emph{tag}
$\langle A_1, N_1\rangle , \ldots, \langle A_k, N_k \rangle$ that will
be used to tag each encryption and signature in~$\Pi$. Our
transformation is surprisingly simple and does not require any
cryptographic protection of the preamble, i.e., an active attacker is
allowed to interfere with this preliminary phase.  Intuitively, the
security relies on the fact that the participant~$A_i$ decides on a
given tag for a given session which is ensured to be fresh as it
contains his own freshly generated nonce~$N_i$.  The transformation is
computationally light as it does not add any cryptographic
application; it may merely increase the size of messages to be
encrypted or signed.
The transformation applies to a large class of protocols, which may
use symmetric and asymmetric encryption, digital signature and hash
functions.

We may note that, \emph{en passant}, we identify a class of tagged protocols for
which security is decidable for an unbounded number of sessions. This directly
follows from our main result as it stipulates that verifying security for a
bounded number of protocol sessions is sufficient to conclude security for an
unbounded number of sessions.

\paragraph*{\bf Related Work.}

The kind of compiler we propose here has also been investigated in the
area of cryptographic design in computational models, especially for
the design of group key exchange protocols. For example, Katz and
Yung~\cite{KatzY03} proposed a compiler which transforms a key
exchange protocol secure against a passive eavesdropper into an
authenticated protocol which is secure against an active
attacker. Earlier work includes compilers for $2$-party protocols
(e.g.~\cite{BellareCK98}). In the symbolic model, recent
works~\cite{CortierWZ07,beauquier07} allow one to transform a protocol
which is secure in a weak sense (roughly no
attacker~\cite{CortierWZ07} or just a passive one~\cite{beauquier07}
and a single session) into a protocol secure in the presence of an
active attacker and for an unbounded number of sessions.  All of these
prior works share however a common drawback: the proposed
transformations make heavy use of cryptography. This is mainly due to
the fact that the security assumptions made on the input protocol are
rather weak. As already mentioned in~\cite{CortierWZ07}, it is
important, from an efficiency perspective to lighten the use of
cryptographic primitives. In this work, we succeed in doing so at the
price of requiring stronger security guarantees on the input
protocol. However, we argue that this is acceptable since efficient
automatic tools exist to decide this security criterion on the input
protocols.  Recently, our transformation has also been adapted to the
case of offline guessing attacks in {password-based}
protocols~\cite{CDK-fsttcs11}. On the one hand the result presented
in~\cite{CDK-fsttcs11} is more complicated as it considers a more
complex security property but, on the other hand, the proof is
simplified by the fact that the password is the only secret shared
between different sessions.

\smallskip{}

We can also compare our work with existing decidable protocol classes
for an unbounded number of sessions. An early result is the PTIME
complexity result by Dolev \emph{et al.}~\cite{dolev82crypto} for a
restricted class, called \emph{ping-pong} protocols.  Other classes
have been proposed by Ramanujam and Suresh~\cite{RS03,RSjournal05},
and Lowe~\cite{Lowejournaltowards}. However, in both cases, temporary
secrets, composed keys and ciphertext forwarding are not allowed which
discards protocols {(even their tagged version)}, such as the Yahalom protocol~\cite{ClarkJacob97}.

\smallskip{}

Different kinds of tags have also been considered
in~\cite{CC-csf10,AD-fsttcs07,CD-fmsd09,BlanchetPodelskiFOSSACS03,RS03}. However
these tags are \emph{static} and have a different aim. While our
dynamic tagging scheme avoids confusing messages from different
sessions, these static tags avoid confusing different messages inside
{the same} session and do not prevent that {the
  same} message is reused in two different sessions.  Under some
additional assumptions (e.g. no temporary secret, no ciphertext
forwarding), several decidability
results~\cite{RSjournal05,Lowejournaltowards} have been obtained by
showing that it is sufficient to consider one session per role. But
those results cannot deal with protocols {which rely on
  ciphertext forwarding and/or temporary secrets}. In the framework we
consider here, the question whether such static tags would be
sufficient to obtain decidability is still an open question
(see~\cite{AD-fsttcs07}). In a similar way, static tags have also been
used by Heather~et~al.~\cite{HeatherLoweScchneider00} to avoid type
confusion attacks.

Finally, we may note that our tags are reminiscent of session tags in
the UC framework~\cite{CanettiFOCS01} and in particular the method
proposed by Barak~et~al.~\cite{cryptoeprint:2004:006} for computing
them. However, in addition to the important differences in the models,
these works do not propose a general, systematic transformation which
guarantees (joint state) composition between sessions.

\smallskip

This paper can be seen as an extended and enriched version
of~\cite{ADK-lpar08}. In~\cite{ADK-lpar08}, our reduction result was
only established for the secrecy property whereas we consider here a
larger class of security properties that includes several levels of
authentication. Moreover, the proof of our main result is now
self-contained and does not rely anymore on the constraint solving
procedure presented in~\cite{CCZ10}.
 
\paragraph{\bf Outline of the paper.}
Our paper is organized in two parts: Part~I presents our result and
all the necessary background for the result to be formally stated and
Part~II is devoted to giving an overview of the proof of the result
(for readability some of the more technical proofs are only given in
{an appendix}).

In Part I we first introduce our abstract representation of protocol
messages (Section~\ref{sec:model}) and our formal models for security protocols
(Section~\ref{sec:protocol}) and properties (Section~\ref{sec:properties}). Next, in
Section~\ref{sec:transfo}, we formally define our protocol transformation and
state our main result which guarantees that attacks only require a
bounded number of sessions.

In Part II we give an overview of our proof. In
Section~\ref{sec:1-step} we define a transformation on protocol
executions and show that a transformed execution 
\begin{enumerate}[label=(\roman*)]
\item has several good properties (it is both valid and well-formed), and
\item preserves the satisfaction of attack formulas.
\end{enumerate}
In Section~\ref{sec:2-step} we show that we can restrict the sessions
that are involved in a valid, well-formed execution while preserving
\begin{enumerate}[label=(\roman*)]
\item validity and well-formedness, and
\item satisfaction of attack formulas.
\end{enumerate}
Finally, in Section~\ref{sec:3-step}, we use the results from the previous two
sections to prove our main result.

%\newpage

\null \hfill {\large\bf  --- PART I: \textsc{Presentation of our reduction result} ---} \hfill\null
\medskip{}

\section{Messages and intruder capabilities}
\label{sec:model}

%%%
%%% Messages
%%%

\subsection{Messages}
\label{subsec:mesages}

We use an abstract term algebra to model the messages of a
protocol. For this we fix several disjoint sets.  We consider an
infinite set of \emph{agents} $\mathcal{A} = \{\epsilon, a,b\ldots\}$
with the special agent~$\epsilon$ standing for the attacker and an
infinite set of \emph{agent variables} $\Xvars = \{x_A, x_B,\ldots\}$.
We also need to consider an infinite set of \emph{names} $\mathcal{N}
= \{n,m\ldots\}$ and an infinite set of \emph{variables} $\Yvars =
\{y, z, \ldots\}$.  Among this set of names, we consider the infinite
set of names $\Nces_\epsilon =\{n^\epsilon, \ldots\}$ that corresponds
to names known initially by the attacker.  We consider the following
\emph{signature} $\mathcal{F} = \{\enc/2, \enca/2, \sign/2,
\langle\,\rangle/2, \h/1, \pub/1, \priv/1, \shk/2\}$.  These function
symbols model cryptographic primitives. The symbol $\langle\,\rangle$
represents pairing.  The term $\enc(m,k)$ (resp. $\enca(m,k)$)
represents the message $m$ encrypted with the symmetric
(resp. asymmetric) key $k$ whereas the term $\sign(m,k)$ represents
the message~$m$ signed by the key~$k$.  The function $\h$ models a
hash function whereas $\pub(a)$ and $\priv(a)$ are used to model the
public and the private key respectively of an agent~$a$, and $\shk(a,
b)$ (= $\shk(b,a)$) is used to model the long-term symmetric key
shared by agents $a$ and~$b$.  Names are used to model atomic data
such as nonces.  The set of \emph{terms} is defined inductively by the
following grammar:
\[
\begin{array}{lcl@{\hspace{0.4cm}}l}
  t, t_1, t_2, \ldots & ::= & & \hspace{-0.4cm} \text{term} \\
  & | & x & \text{agent variable } x \in \mathcal{X} \\
  & | & a & \text{agent } a \in \mathcal{A} \\
  & | & y & \text{variable } y \in \mathcal{Y} \\
  & | & n & \text{name } n \in \mathcal{N} \\
  & | & \pub(u) & \text{application of the symbol $\pub$  on $u \in
    \mathcal{A} \cup \mathcal{X}$} \\
  & | & \priv(u) & \text{application of the symbol $\priv$  on $u \in
    \mathcal{A} \cup \mathcal{X}$} \\
  & | & \shk(u_1,u_2) & \text{application of the symbol $\shk$ on
    $u_1, u_2 \in \mathcal{A} \cup \mathcal{X}$} \\
  & | & \h(t) & \text{application of $\h$} \\
  & | & \f(t_1,t_2) & \text{application of symbol }\f \in \{\enc,
  \enca, \sign, \langle\, \rangle\}
\end{array}
\]

We sometimes write $\langle t_1, \ldots, t_n\rangle$ instead of
writing $\langle t_1, \langle \ldots, \langle t_{n-1}, t_n\rangle
\ldots \rangle\rangle$.  We say that a term is \emph{ground} if it has
no variable.  We consider the usual notations for manipulating terms.
A position~$p$ in a term~$t$ is a sequence of integers.  The empty
sequence $\varepsilon$ denotes the top-most position. The subterm
of~$t$ at position~$p$ is written $t|_p$.
%The term obtained by replacing t|p with u is denoted t[u]p.
We write $\vars(t)$ (resp. $\names(t)$, $\agent(t)$) for the set of
variables (resp. names, agents) occurring in~$t$.  We write $\st(t)$
for the set of {\emph{syntactic subterms}} of a term~$t$ and
define the set of \emph{cryptographic subterms} of a term~$t$ as
$\encst(t) = \{\f(t_1,\dots,t_n)\in\st(t)\ |\
\f\in\{\enc,\enca,\sign,\h\}\}.$ Moreover, we define the set of
\emph{long-term keys} as $\lgAtom = \{\priv(a)~|~ a\in \A\} \cup
\{\shk(a,b)~|~a,b \in \A\}$ and the set of \emph{long-term keys of a
  term $t$} as
\[
\lgAtom(t) = \{\priv(u)~|\pub(u)\in\st(t) \text{ or } \priv(u)\in\st(t)\}\cup\{\shk(u_1,u_2)\in\st(t)\}.
\]

\noindent and we define $\Kcal_\epsilon = \{\priv(\epsilon)\} \cup
\{\shk(a,\epsilon)~|~a \in \A\}$. Intuitively $\Kcal_\epsilon$ represents  the
set of long-term keys of the attacker.
An \emph{atom} is a long-term key, a name or a variable.

We define the set of \emph{plaintexts} of a term~$t$ as the set of atoms that occur in
plaintext position, i.e.
\begin{itemize}
\item   $\plaintext(\h(u))=\plaintext(\f(u, v)) =
\plaintext(u)$ for $\f \in \{\enc, \enca, \sign\}$,
\item  $\plaintext(\langle u, v \rangle) = \plaintext(u)
\cup \plaintext(v)$, and
\item $\plaintext(u)
= \{u\}$ \; otherwise.
\end{itemize}

All these notions are extended to sets of terms and to
other kinds of term containers
as expected. 
We denote by $\#S$ the cardinality
  of a set~$S$.
Substitutions are written ${\sigma = \{x_1 \mapsto t_1 ,\ldots, x_n
  \mapsto {t_n} \}}$ where its $\emph{domain}$ is $\dom(\sigma) = \{
x_1,\ldots,x_n\}$.  The substitution~$\sigma$ is \emph{ground} if all
the $t_i$ are ground.  The application of a substitution $\sigma$ to a
term~$t$ is written $\sigma(t)$ or $t\sigma$. Two terms~$t_1$ and~$t_2$ 
are \emph{unifiable} if $t_1\sigma = t_2\sigma$ for some
substitution~$\sigma$, that is called a \emph{unifier}.  
We denote by $\mgu(t_1,t_2)$ the \emph{most general unifier} of~$t_1$
and~$t_2$.

\begin{exa}
  Let $t=\enc(\langle n, a \rangle, \shk(a,b))$. We have that
  $\vars(t)=\emptyset$, i.e.  $t$ is ground, $\names(t) = \{n\}$,
  $\agent(t) = \{a,b\}$, 
$\lgAtom(t) = \{\shk(a,b)\}$, $\plaintext(t)=\{n, a\}$, {and $\st(t) =
\{t,\,\langle n, a\rangle, \, \shk(a,b),\, n, \, a\}$}. The terms $\shk(a,b)$, $a$, $n$  and~$\priv(a)$ are
  atoms.
\end{exa}

%%%
%%% Intruder capabilitites
%%%

\subsection{Intruder capabilities}
\label{subsec:intruder}

We model the intruder's abilities to construct new messages by the
deduction system given in Figure~\ref{fig:deduction}. The first line
describes the \emph{composition rules}. The second line describes the
\emph{decomposition rules}.
The intuitive meaning of these rules is that an intruder can compose new
messages by pairing, encrypting, signing and hashing previously known
messages provided he has the corresponding keys. Conversely, he can
decompose messages by projecting or decrypting provided he has the
decryption keys.
Our optional rule 
expresses that an intruder can retrieve the whole message from
its signature. Whether this property holds depends on the actual
signature scheme. Therefore we consider this rule to be optional. Our
results hold in both cases.

%\medskip{}

\begin{figure}[t]
\[
\begin{array}{c}
\prooftree
u\qquad v
\justifies
\pair{u}{v}
\endprooftree
\quad\quad %\quad
 \prooftree
u\qquad v \justifies
\enc(u,v)
\endprooftree
\quad\quad %\quad
 \prooftree u\qquad v \justifies
 \enca(u,v)
\endprooftree
\quad\quad %\quad
\prooftree u\qquad v \justifies
\sign(u,v)
\endprooftree
\quad\quad %\quad
\prooftree u \justifies
 \h(u)
\endprooftree
\\ \\
\prooftree
\pair{u}{v}
\justifies
u
\endprooftree
\quad\quad
\prooftree 
\pair{u}{v} 
\justifies 
v
\endprooftree
\quad\quad
\prooftree \enc(u,v)\quad v
\justifies u
\endprooftree
\quad\quad
\prooftree
\enca(u,\pub(v))\quad\priv(v)
\justifies
u
\endprooftree
\quad\quad
\prooftree 
\sign(u,v) 
\justifies
u 
\using 
\mbox{\it (optional)}
\endprooftree

\end{array}
\]
\caption{Intruder deduction system.}
\label{fig:deduction}
\end{figure}

\begin{defi}[deducible]~\label{def:deducibility}
  We say that a term $u$ is \emph{deducible} from a set of terms~$T$,
  denoted $T\vdash u$, if there exists a tree such that its root is
  labeled by $u$, its leaves are labeled with $v \in T \cup \A \cup
  {\Nces_\epsilon} \cup
  \Kcal_\epsilon \cup \{\pub(a) ~|~ a \in \A\}$ and for every  node labeled by $v$ having $n$
  sons labeled by $v_{1}, \dots , v_{n}$ we have that 
$\;\;\begin{prooftree}
v_{1} \; \ldots  \; v_{n}
\justifies
v
\end{prooftree}\;\;$
is an
  instance of one of the inference rules given in Figure~\ref{fig:deduction}.
\end{defi}

\begin{exa}
  The term $\pair{n}{\shk(a,b)}$ is deducible from $\{\enc(n,\shk(a,b)),\shk(a,b)\}$.
\end{exa}

We are now able to state the following lemma that can be easily proved by
induction on the proof tree witnessing $T \vdash t$.

\begin{lem}\label{lem:plaintext.deducibility}
 Let $T$ be a set of terms and $t$ be a term such that $T\vdash t$. We
 have that:
\[
\plaintext(t)\ \subseteq\ \plaintext(T)\cup \A \cup
{\Nces_\epsilon} \cup \Kcal_\epsilon \cup \{\pub(a) ~|~
a \in \A\}.
\]
\end{lem}
%\noindent The proof is a straightforward induction on the proof tree of $T\vdash t$.

% 10/10/2012 - S.D.: la preuve est ok mais tellement simple que je ne
% la trouve pas trÃ¨s utile compte tenu de la longueur du papier.

% \begin{proof}
%   Let $u\in\plaintext(t). $We prove by induction on the depth $d$ of the proof $T \vdash t$, that $u\in\plaintext(T)\cup \A \cup \Nces_\epsilon \cup \Kcal_\epsilon \cup \{\pub(a) ~|~ a \in \A\}$.
%   \medskip{}

%   \noindent\emph{\underline{Case $d=1$.}} In that case  $t\in T\cup \A \cup \Nces_\epsilon \cup \Kcal_\epsilon \cup \{\pub(a) ~|~ a \in \A\}$. But then by definition $u \in \plaintext(t) \subseteq \plaintext(T) \cup \A \cup \Nces_\epsilon \cup \Kcal_\epsilon \cup \{\pub(a) ~|~ a \in \A\}$.
%   \medskip{}

%   \noindent\emph{\underline{Case $d>1$.}} In that case the proof of $T\vdash t$ is of the form
%   \[
%   \begin{prooftree}
%     \begin{prooftree}
%       \dots
%       \justifies
%       T\vdash t_1
%     \end{prooftree}
%     \quad
%     \dots
%     \quad
%     \begin{prooftree}
%       \dots
%       \justifies
%       T\vdash t_n
%     \end{prooftree}
%     \justifies
%     T\vdash t
%   \end{prooftree}
%   \]
%   and $\plaintext(t) \subseteq \underset{i\in [n]}{\bigcup} \plaintext(t_i)$. Hence, there exists $i\in [n]$ such that $u\in \plaintext(t_i)$. But then we can conclude by inductive hypothesis that $u \in \plaintext(T)\cup \A \cup \Nces_\epsilon \cup \Kcal_\epsilon \cup \{\pub(a) ~|~ a \in \A\}$.\qedhere

% \end{proof}\input{messages}

\section{Model for security protocols}
\label{sec:protocol}

In this section, we give a language for specifying protocols and define their
execution in the presence of an active attacker.  Our model is similar to
existing ones (see e.g.~\cite{RT03TCS,CD-fmsd09}).

%%%
%%% Syntax
%%%

\subsection{Syntax}
\label{subsec:syntax-protocol}

We consider protocols specified in a language allowing parties to exchange messages built from identities and randomly generated nonces using pairing, public key, symmetric encryption, hashing and digital signatures.  The individual behavior of each protocol participant is defined by a \textit{role} describing a sequence of \emph{events}. The main events we consider are \emph{communication events} (i.e. message receptions and message transmissions) and \emph{status events} to mark different stages reached by the protocol. These status events will help us specify a large class of security properties (a logic of properties is given in Section~\ref{sec:properties}). These are issued by participants to denote their current state in the execution of a protocol role.

\begin{defi}[event]
  An \emph{event} is either
  \begin{itemize}
  \item a \emph{communication event}, i.e. a message reception,
    denoted by $\rcv(m)$ or a message transmission, denoted by
    $\snd(m)$, where $m$ is a term; or
    % \item a \emph{new event}, denoted by $\new \; X$ where $X$ is a
    %   variable; or
  \item a \emph{status event} of the form $\Prd(t_1, \ldots, t_n)$ where
    each~$t_i$ is a term (not necessarily ground) and~$\Prd\in\Pred$ is a
    predicate symbol of arity~$n$.
  \end{itemize}
\end{defi}

\noindent Typically, status events give information about the state of the principal. For
instance, we will consider a status event that indicates that the principal has
started or finished a session. The set of variables of an event is defined as
expected, considering all the terms occurring in the event's specification.

\begin{defi}[roles]
  \label{def:role}
  \label{def:protocol}
  A role is of the form $\lambda x_1.\dots\lambda x_k.\nu y_1.\dots\nu
  y_p.\; \seqevt$, where:
  \begin{itemize}
  \item $X = \{x_1,\ldots,x_k\}$ is a set of agent variables, i.e. the
    parameters of the role corresponding to the~$k$ participants of
    the protocol,
  \item $Y = \{y_1,\ldots,y_p\}$ is a set of variables: the nonces
    generated by the role,
  \item $\seqevt=\evt_1;\evt_2; \ldots ;\evt_\ell$ is a sequence of
    events such that ${(\vars(\seqevt) \smallsetminus {X}) \subseteq
      \mathcal{Y}}$, i.e. {all agent variables} are parameters.
  \end{itemize}
  Moreover, we have that:
  \begin{enumerate}
  \item\label{def:role:or} $\seqevt$ satisfies the \emph{origination
      property}, that is for any send or status event $\evt_i$, for
    any variable {$x \in \vars(\evt_i) \smallsetminus (X \cup Y)$, we
      have that $x \in \vars(\evt_j)$ for some receive event~$\evt_j$
      where $j < i$}; and
  \item\label{def:role:plaintextor} $\seqevt$ satisfies the
    \emph{plaintext origination property}, that is for any send or
    status event~$\evt_i$, for any variable {$x \in \plaintext(\evt_i)
      \smallsetminus (X \cup Y)$}, we have that
    $x\in\plaintext(\evt_j)$ for some receive event~$\evt_j$ where $j
    < i$.
  \end{enumerate}
  
  \noindent The set of roles is denoted by~$\Roles$.  The
  \emph{length} of a role is the number of elements in its sequence of
  events.  A \emph{$k$-party protocol} is a mapping
  ${\Pi:[k]\to\Roles}$, where $[k] = \{1,2,\ldots,k\}$.
\end{defi}

The condition (\ref{def:role:or}) above ensures that each variable
which appears in a send or status event is a nonce, a parameter, or a
variable that has been introduced in a previously received
message. Condition  (\ref{def:role:plaintextor}) ensures that a key
used for encrypting or signing cannot be extracted and used as
plaintext, e.g. forbidding a sequence $\rcv(\enc(y,z));\snd(z)$.

% The set of variables, names or extended names of a protocol is
% defined as expected, considering all the terms occurring in the
% role's specification.

\begin{exa}    
  \label{ex:syntax}
  We illustrate our protocol syntax on the familiar Needham-Schroeder public-key
  protocol~\cite{NS78}.  In our syntax this protocol is modeled as follows.
  \[
  \begin{array}{rl}
    \Pi(1)=& \lambda x_A.\lambda x_B.\nu y.\\
    & \snd(\enca(\pair{y}{x_A},\pub(x_B)));\\
    & \rcv(\enca(\pair{y}{z},\pub(x_A)));\\
    & \snd(\enca(z,\pub(x_B)))
  \end{array}\quad
  \begin{array}{rl}
    \Pi(2)=& \lambda x_A.\lambda x_B. \nu y'.\\
    & \rcv(\enca(\pair{z'}{x_A},\pub(x_B)));\\
    & \snd(\enca(\pair{z'}{y'},\pub(x_A)));\\
    & \rcv(\enca(y',\pub(x_B)))
  \end{array}
\]

The initiator, role $\Pi(1)$ played by~$x_A$, sends to the responder, role
$\Pi(2)$ played by~$x_B$, his identity together with a freshly generated
nonce~$y$, encrypted with the responder's public key. The {responder} replies by
copying the initiator's nonce and adds a fresh nonce~$y'$, encrypted by the
initiator's public key. The initiator acknowledges by forwarding the responder's
nonce encrypted by {his} public key.
\end{exa}

Clearly, not all protocols written using the syntax above are
meaningful. In particular, some of them might not be
\emph{executable}.  For instance, a $k$-party protocol where $\Pi(1)
:= \rcv(\h(x)); \snd(x)$ is not executable since an agent is not able
to extract the content of a hash.  A precise definition of
executability is not relevant for our result. We only need to consider
the weaker plaintext origination hypothesis
(Condition~\ref{def:role:plaintextor} stated in Definition~\ref{def:role}).
In particular, our result also holds for {non-executable} protocols such
as the one given above.

%%%
%%% Semantics
%%%

\subsection{Semantics}
\label{subsec:semantics-protocol}

In our model, a session corresponds to the instantiation of one
role. This means in particular that one ``normal execution'' of a
$k$-party protocol requires~$k$ sessions, one per role\footnote{In the
  literature, the word session is often used in an abusive way to
  represent an execution of the \emph{protocol}, i.e. one session per
  role, whereas we use it for the execution of a \emph{role}.}.  We
may want to consider several sessions corresponding to different
instantiations of a same role.  Since the adversary may block,
redirect and send new messages, all the sessions might be interleaved
in many ways. Such an interleaving is captured by the notion of a
\emph{scenario}.

\begin{defi}[scenario]
  A \emph{scenario for a protocol $\Pi:[k]\to\Roles$} is a sequence
  $\sce = (r_1,\sid_1)\cdots(r_n,\sid_n)$ where $r_i$ is a role and
  $\sid_i$ a session identifier such that $1\leq r_i\leq k$,
  $\sid_i\in\N\smallsetminus\{0\}$,
%\myrto{We didn't use anywhere the fact that terms in $T_0$ are
%annotated with the session id 0.}
% S.D.: si dans la preuve, 
  the number of identical occurrences of a pair $(r,\sid)$ is smaller
  than the length of the role~$r$, and $\sid_i = \sid_j$ implies
  $r_i=r_j$.
\end{defi}

The condition on identical occurrences ensures that a role cannot
execute more events than it contains. The last condition ensures that
a session number is not reused by other roles.  We say that
$(r,s)\in\sce$ if $(r,s)$ is an element of the sequence $\sce$.

Given a scenario and an instantiation for the parameters, we define a \emph{symbolic trace}, that is a sequence of events that corresponds to the interleaving of the scenario, for which the parameters have been instantiated, fresh nonces are generated and variables are renamed to avoid name collisions between different sessions.

\begin{defi}[symbolic trace]\label{def:scenario}
  Let $\Pi$ be a $k$-party protocol with
  \begin{center}
    $\Pi(j) = \lambda x^j_1.\dots\lambda x^j_k.\nu y^j_1.\dots\nu
    y^j_{p_j}. \evt^j_1;\dots ;\evt^j_{\ell_j}$ \;\;\; for $1\le j\le
    k$.
  \end{center}
  Given a scenario $\sce = (r_1,\sid_1)\cdots(r_n,\sid_n)$ and a
  function $\alpha: \mathbb{N}\rightarrow\mathcal{A}^k$, the
  \emph{symbolic trace $\tr = \evt_1^{\sid_1};\ldots; \evt_n^{\sid_n}$
    associated to $\sce$ and~$\alpha$} is defined as follows.
 
% Let $q_i = \# \{(r_j,\sid_j)\in\sce\mid j\leq i, \sid_j =
 % \sid_i\}$, i.e. the number of previous occurrences in $\sce$ of the
 % session $\sid_i$.  

 Let {$q_i = \# \{j~|~j  \leq i,  (r_j,\sid_j)\in\sce,
   \mbox{ and }  \sid_j =
  \sid_i\}$}, i.e. {the number of occurrences up to this point} in $\sce$ of the
  session $\sid_i$.  
We have that $q_i\leq \ell_{r_i}$ and
  ${\evt_i=(\evt^{r_i}_{q_i})\sigma_{r_i,\sid_i}}$, where
  $\dom(\sigma_{r,\sid}) = \vars(\Pi(r))$ and
  \begin{itemize}
 \item $\sigma_{r,\sid}(y) = n_y^\sid$ if
    $y\in\{y^r_1,\dots,y^r_{p_r}\}$, where $n_y^\sid$ is a fresh name
    {from $\mathcal{N}$};
  \item $\sigma_{r,\sid}(x^r_i) = a_i$ when $\alpha(\sid) = (a_1,
    \ldots, a_k)$;
  \item $\sigma_{r,\sid}(z) = z^\sid$ otherwise, where $z^\sid$ is a
    fresh variable.
  \end{itemize} {A session~$\sid$ is said to be \emph{dishonest}
    w.r.t.~$\alpha$ and a set of ground atoms~$T_0$ when $\alpha(\sid)
    = (a_1,\ldots,a_k)$ and $T_0 \vdash \priv(a_i)$ or $T_0 \vdash
    \shk(a_i,v)$ for some $v \neq \epsilon$ and $1\leq i \leq k$.}
\end{defi}

\noindent {Intuitively, a session $\sid$ is honest if all of its
  participants, from the point of view of the agent playing the
  session~$\sid$, are honest (i.e. they are neither the attacker
  $\epsilon$ nor did they disclose their long-term keys).}  Note that
since all agent variables occurring in a role, occur as parameters of
this role (see Definition~\ref{def:protocol}), a symbolic trace does
not contain agent variables.

{The notational conventions we use for names and variables
  occurring in a symbolic trace (\emph{e.g.} $n_y^\sid$ and $z^\sid$)
  are not really relevant to state our main result. However, we will
  rely on this notation in Part II when we prove our reduction
  result.}

\begin{exa}
\label{ex:symbolic-trace}
Consider again the Needham-Schroeder protocol. Let~$\Pi(1)$ and~$\Pi(2)$ be the
two roles introduced in Example~\ref{ex:syntax}. Let~$s_1$ and~$s_2$ be two sessions
numbers ($s_1 \neq s_2$), $\sce =
(1,s_1)(2,s_2)(2,s_2)(1,s_1)(1,s_1)$ and $\alpha$ the function such
that $\dom(\alpha) = \{s_1, s_2\}$, $\alpha(s_1) = (a,c)$, and
$\alpha(s_2) = (a,b)$. This is the scenario
allowing us to retrieve the famous attack due to Lowe~\cite{lowe96breaking}.
The symbolic trace associated to $\Pi$, $\sce$, and $\alpha$ is given below:
\[
\begin{array}{rcl}
\tr &=& 
\snd(\enca(\langle n_{y}^{s_1}, a\rangle, \pub(c)));\\
&&\rcv(\enca(\langle {z'}^{s_2}, a \rangle, \pub(b))); \,
\snd(\enca(\langle {z'}^{s_2}, n_{y'}^{s_2}\rangle, \pub(a)));\\
&&\rcv(\enca(\langle n_{y}^{s_1}, z^{s_1} \rangle, \pub(a))); \,
\snd(\enca(z^{s_1}, \pub(c)))
\end{array}
\]
\end{exa}

An \emph{execution trace} is an instance of such a symbolic trace.
Appending an event~$\evt$ to an execution trace~$\exec$ is
written~$\exec; \, \evt$. The function $\length$ has the usual
meaning: $\length([\,]) = 0$ and $\length(\exec; \,\evt) = 1 +
\length(\exec)$. The prefix of an execution trace consisting of the
first~$i$ events is denoted as $\exec_i$, with $\exec_0 =[\,]$ and
$\exec_n = \exec$ when $n \geq \length(\exec)$.

\begin{defi}[knowledge of an execution trace
  $\exec$]\label{dfn:intruder.knowledge} Let $\exec$ be an execution trace.
  The knowledge of~$\exec$ is the set of terms given by $\K(\exec) = \{ u ~|~
  \snd(u) \in \exec\}$.
\end{defi}

As usual, we are only interested in \emph{valid} execution traces - those traces
where the attacker only sends messages that he can compute from his
initial knowledge
and the messages he has seen on the network.

\begin{defi}[valid execution trace]\label{def:validexec}
  Let $T_0$ be a set of ground terms (intuitively $T_0$ represents the initial knowledge of the attacker).  A ground execution trace $\exec = \evt^{\sid_1}_1; \ldots; \evt^{\sid_\ell}_\ell$ is \emph{valid} w.r.t.~$T_0$ if for all $1 \leq i \leq \ell$, whenever $\evt_i = \rcv(m)$, we have that ${T_0 \cup \K(\exec_i) \vdash m}$.
\end{defi}

\begin{exa}
\label{ex:valid}
Let $T_0 = \{a, b, c, \priv(c)\}$.  Let~$\tr$ be the symbolic trace described in Example~\ref{ex:symbolic-trace} and $\sigma = \{z^{s_1} \mapsto n_{y'}^{s_2}, \; {z'}^{s_2} \mapsto n_y^{s_1} \}$.  The execution trace~$\tr\sigma$ is valid w.r.t.~$T_0$. Indeed, we have that
\begin{itemize}
\item $T_1 \stackrel{\mathsf{def}}{=} T_0 \cup  \{\enca(\langle n_y^{s_1}, a \rangle, \pub(c))\} \,\vdash\,
  \enca(\langle n_y^{s_1}, a \rangle, \pub(b))$, and 
\item $T_1 \cup \{\enca(\langle n_y^{s_1}, n_{y'}^{s_2} \rangle, \pub(a))\} 
\,\vdash\, \enca(\langle n_y^{s_1}, n_{y'}^{s_2} \rangle, \pub(a))$.
\end{itemize}
\end{exa}

The purpose of the following lemma is to characterize the terms that
occur in plaintext position in a valid execution.
{Intuitively, the lemma states that any plaintext occurring
  in a valid execution either occurs as a plaintext in the underlying
  symbolic trace, or was known by the attacker since the beginning,
  i.e., is part of the attacker's initial knowledge.}

\newsavebox{\lemplaintext}
\sbox{\lemplaintext}{\vbox{%
\begin{lem}\label{lem:plaintext}
  Let $\Pi$ be a $k$-party protocol and $\tr = [\evt\evt_1^{\sid_1}; \ldots; \evt\evt_\ell^{\sid_\ell}]$ be a symbolic trace associated to it. Let $T_0$ be a set of ground atoms, and $\exec = [\evt_1^{\sid_1}; \ldots; \evt_\ell^{\sid_\ell}]$ be a valid execution trace associated to $\tr$ (w.r.t. $T_0$). We have that:
  \[
  \plaintext(\exec)\ \subseteq\ \plaintext(\tr) \cup T_0 \cup \Nces_\epsilon \cup \mathcal{K}_\epsilon \cup \A \cup \{\pub(a) \mid a\in\A\}.
  \]
\end{lem}
}}

\noindent\usebox{\lemplaintext}

This lemma can be shown by induction on the length of the underlying
symbolic trace. We rely on Lemma~\ref{lem:plaintext.deducibility} to
deal with the case of a receive event, and on the plaintext
origination property (Condition (\ref{def:role:plaintextor}) in
Definition~\ref{def:role}) to deal with the case of a status or a send
event.

\section{Security properties}
\label{sec:properties}

In this section, we propose a logic for specifying security
properties. Our logic is similar to existing ones (see
e.g~\cite{TheseCorin06,CD-fmsd09}). {In particular, it is
  expressive enough to specify security properties like secrecy and
  different forms of authentication including aliveness, weak
  agreement and non-injective agreement.}  Its semantics is defined as
usual on execution traces.

%%%%%%%%%%%%%%%%%%%%%%%%%%%%%%%%%%%%%%%%%%%%%%%%%%%%%%%%%%%%%%%%%%%%%%%%
%                                                                      %
%                          La logique                                  %
%                                                                      %
%%%%%%%%%%%%%%%%%%%%%%%%%%%%%%%%%%%%%%%%%%%%%%%%%%%%%%%%%%%%%%%%%%%%%%%%

\subsection{A logic for security properties}
\label{subsec:logic}

As in~\cite{CD-fmsd09}, status events are used to specify security properties
while the other events describe the execution of the protocol. We only consider
one temporal operator and this operator should only concern status events. That
is why we divide the logic into two layers.

\begin{defi}
\label{def:syntax}
  A formula of $\Logic$ is an expression $\phi$ defined by the following grammar:
  \[
  \begin{array}{rcl}
    \phi, \phi_i &:=& \learn(u_0) \mid \neg
    \phi \mid \exists x. \phi \mid \phi_1 \vee \phi_2 \mid \Comp(u)
    \mid \Diamond \, \psi \mid  \psi \\[2mm]
    \psi, \psi_i  &:=& \true \mid \Prd(u_1,\ldots, u_n)  \mid \neg
    \psi \mid \psi_1 \vee \psi_2
  \end{array}
  \]
  \noindent where $u_0, u_1, \ldots, u_n$ are terms and $u \in \mathcal{A} \cup \mathcal{X}$. 
\end{defi}

Standard formulas $\true$, $\neg \phi$, and $\phi_1 \vee \phi_2$ carry
the usual meaning. The formula~$\learn(u_0)$ states that the attacker
knows the term~$u_0$, whereas $\Prd(u_1, \ldots, u_n)$ is a status
event. The formula $\Comp(u)$ states that the agent $u$ is compromised
(his secret keys are known to the attacker). The formula $\Diamond\,
\psi$ means that `$\psi$ held in the past'. When $x$ is a variable, we
write $\exists x. \phi$ to bind~$x$ in~$\phi$, with the quantifier
carrying the usual meaning. Other operators can be represented using
the above defined operators. For instance, the abbreviations
$\Nocomp(u)$, $\false$, $\wedge$, {$\forall$}, and
$\Rightarrow$ are defined by {$\Nocomp(u) \stackrel{\mathsf{def}}{=}
  \neg \Comp(u)$}, $\false \stackrel{\mathsf{def}}{=} \neg \true$,
$\phi_1 \wedge \phi_2 \stackrel{\mathsf{def}}{=} \neg (\neg\phi_1 \vee
\neg\phi_2)$, ${\forall x. \phi
  \stackrel{\mathsf{def}}{=} \neg\exists x.\neg\phi}$, and $\phi_1
\Rightarrow \phi_2 \stackrel{\mathsf{def}}{=} \neg\phi_1 \vee \phi_2$.

In the sequel, we assume that formulas are \emph{closed}, i.e. they
contain no free variables, and that each variable is quantified at
most once (this can be easily ensured by using renaming). We also
assume that the variables occurring in a formula~$\phi$ are disjoint
from the variables occurring in the considered symbolic
 trace.

\medskip{}

Formulas are interpreted at some position along an execution trace as stated in Definition~\ref{def:validity}.

\begin{defi}[concrete validity]
  \label{def:validity}                 
  Let $\phi$ be a closed formula in $\Logic$, $\exec$ be a ground execution trace and~$T_0$ be a set of ground terms. We define ${\langle \exec, T_0 \rangle \models \phi}$ as:
  \[
  \begin{array}{lcl}
    \langle \exec, T_0 \rangle \models \true && \\
    \langle \exec, T_0 \rangle \models \learn(m) & \;\; \mbox{iff}\;\;
    &T_0 \cup  \K(\exec) \vdash m \\
    \langle \exec, T_0 \rangle \models \neg \phi & \mbox{iff} & \langle
    \exec, T_0 \rangle \not\models \phi \\
    % \langle \exec, T_0 \rangle \models \phi_1 \wedge \phi_2 &
    % \mbox{iff} & \langle \exec, T_0 \rangle \models \phi_1
    % ~\mbox{and}~\langle \exec, T_0 \rangle \models \phi_2 \\
    \langle \exec, T_0 \rangle \models \phi_1 \vee \phi_2 & \mbox{iff} &
    \langle \exec, T_0 \rangle \models \phi_1 ~\mbox{or else}~\langle \exec,
    T_0 \rangle \models \phi_2 \\
    \langle \exec, T_0 \rangle \models \exists x. \phi & \mbox{iff} &
    \mbox{there exists a ground term $t$ s.t.}~ \langle \exec, T_0
    \rangle \models \phi\{x \mapsto t\} \\
    % \langle \exec, T_0 \rangle \models \forall x. \phi & \mbox{iff} &
    % \mbox{for all ground term $t$ we have that}~ \langle \exec, T_0
    % \rangle \models \phi[x \mapsto t] \\
    \langle \exec, T_0 \rangle \models \Prd(t_1, \ldots, t_n) & \;\;
    \mbox{iff} \;\; & \exec = \exec'; \,  \Prd(t_1, \ldots, t_n) \\
    \langle \exec, T_0 \rangle \models \Comp(u) & \;\; \mbox{iff} \;\; &
    T_0 \vdash \priv(u) \mbox{ or } T_0 \vdash \shk(u,v) \mbox{ for
      some $v \neq \epsilon$}\\
    % \langle \exec, T_0 \rangle \models \Y \psi & \mbox{iff} &
    % \overline{\exec} = \exec'; \, \evt ~\mbox{and}~\langle \exec', T_0
    % \rangle \models \psi \\
    \langle \exec, T_0 \rangle \models \Diamond \, \psi & \mbox{iff} &
    \exists i \in [0, \length (\exec)] \mbox{ such that } \langle \exec_i,
    T_0\rangle \models \psi
  \end{array}
  \]
Given a protocol $\Pi$ and a set of ground terms~$T_0$, we say that $\Pi \models \phi$ w.r.t.~$T_0$, if $\langle \exec, T_0\rangle \models \phi$ for all valid execution traces~$\exec$ of~$\Pi$ w.r.t. $T_0$.
\end{defi}

We now define the subset of $\Logic$ for which our result holds. We
say a formula in $\Logic$ is \emph{quantifier-free} if it does not
contain any $\exists$. A formula is \emph{modality-free} if it does
not contain any $\Diamond$.  We will only consider \emph{attack
  formulas} of the form $\exists x_1. \ldots.  \exists x_n. \phi'$
where $\phi'$ is quantifier-free, and we consider also some additional
syntactic restrictions.  {Therefore, the security formulas
  we consider are of the form $\forall x_1, \ldots, \forall x_n. \neg
  \phi'$, i.e. the negation of an attack formula.}

%We also need to control the occurrences of $\learn(m)$ and the We say that a formula $\phi$ is \emph{positive} (resp. \emph{negative}) if every occurrence o%f $\learn(m)$ in $\phi$ appears under an even (resp. odd) number of negations. This restriction allows us to avoid negated  deducibility constraints.

\begin{defi}[attack formula]
  \label{def:attack}
  An attack formula is an expression of the form 
\[
\exists x_1. \dots. \exists x_n. \phi
\]

\noindent where all the variables $x_i$ are distinct and $\phi$ is a quantifier-free formula of $\Logic$ satisfying the following conditions:
  \begin{enumerate}
  \item\label{def:attack-cond1} all subterms of $\phi$ are atomic terms with no names, i.e. $\st(\phi) \subseteq \A \cup \Xvars \cup \Yvars$,
    % \myrto{Peut-on trouver des cles long-terme dans $\phi$?} S.D.: Non, mais ce n'est pas une reelle restriction pour l'expressivite.
  \item for any term $t$, $\learn(t)$ can only occur positively in $\phi$, i.e. under an even number of negations, 
  \item any variable occurs at most once in a positive status event,
  \item \label{def:attack-cond4} if $\Diamond\psi$ is a subformula of $\phi$ that occurs negatively in $\phi$, then a status event can only occur positively in $\psi$.
  \end{enumerate}
\end{defi}

\noindent As we will see next this fragment is expressive enough to model
classical security properties.

%%%%%%%%%%%%%%%%%%%%%%%%%%%%%%%%%%%%%%%%%%%%%%%%%%%%%%%%%%%%%%%%%%%%%%%%
%                                                                      %
%                           Examples                                   %
%                                                                      %
%%%%%%%%%%%%%%%%%%%%%%%%%%%%%%%%%%%%%%%%%%%%%%%%%%%%%%%%%%%%%%%%%%%%%%%%

\subsection{Some security properties}
\label{subsec:ex-properties}

We now show how classical security properties like secrecy and
{several flavors of non-injective authentication properties}
can be expressed in our logic.

\subsubsection{Secrecy}
\label{subsubsec:secrecy}

The secrecy property is the inability of the intruder to learn a
message {(e.g. a nonce, a key, or a compound term)} that is specified (using a status event) as confidential. 
We will show how to specify the secrecy property for a nonce for example with a formula in $\Logic$. Let $\Pi$ be a $k$-party protocol with
\[
\Pi(j) = \lambda x_1^j.\dots \lambda x_k^j.\nu y^j_1.\dots \nu y^j_{p_j}.\evt^j_1;\dots; \evt^j_{\ell_j} \quad \text{for } 1\le j\le k.
\]
and let $y^j_h$ ($1\le j\le k$ and $1\le h\le p_j$) be the nonce variable whose instantiations should remain confidential. In order to specify that all the instances of $y^j_h$ must remain secret we build from $\Pi$, a protocol $\Pi_\Sec$ as follows. Let $\Secret$ be a predicate not occurring in $\Pi$, then
\[
\Pi_\Sec(n) =
\left\{
  \begin{array}[c]{l}
    \Pi(n) \hfill \text{for } 1\le n\le k \text{ and } n\not=j \\
    \lambda x_1^j.\dots \lambda x_k^j.\nu y^j_1.\dots \nu y^j_{p_j}. \Secret(x_1^j, \dots, x_k^j, y^j_h); \evt^j_1;\dots; \evt^j_{\ell_j} \quad\quad\quad\quad\quad \text{for } n=j
  \end{array}
\right.
\]
During an execution, the predicate $\Secret$ will link each instance $n^\sid_{y^j_h}$ of $y^j_h$ to the participants of the corresponding session $\sid$. The following property expresses that the non-compromised instances of $y^j_h$ should remain confidential
\[
\phi_\Sec = \forall x_1.\dots.\forall x_k. \forall y.\ [((\Diamond\Secret(x_1, \dots, x_k, y))\ \wedge\ \Nocomp(x_1)\ \wedge\ \dots\ \wedge\ \Nocomp(x_k))\ \Rightarrow\ \neg\learn(y)].
\]
And the following formula is the corresponding attack formula
\[
\overline{\phi_\Sec} = \exists x_1.\dots.\exists x_k. \exists y.\ [(\Diamond\Secret(x_1, \dots, x_k, y))\ \wedge\ \Nocomp(x_1)\ \wedge\ \dots\ \wedge\ \Nocomp(x_k)\ \wedge\ \learn(y)].
\]
which satisfies the 4 conditions of the definition of an attack
formula (Definition~\ref{def:attack}).  {Note that the same
  construction can be used to model the secrecy of a compound term $t$
  as seen by the agent executing the role $\Pi(j)$. For this, we
  simply add a status event $\Secret(x_1^j, \dots, x_k^j, t)$ in
  $\Pi(j)$, and keep the attack formula unchanged. The 4 conditions
  stated in Definition~\ref{def:attack} are still satisfied.}

\begin{exa}
  \label{ex:secrecy}
  Let us come back to the Needham-Schroeder protocol as presented in
  Example~\ref{ex:syntax} to illustrate this property, and let's
  specify that the nonce $y'$ generated by the responder is
  confidential. In order to do so, we build the 2-party protocol
  $\Pi_\Sec$ following the above mentioned construction, i.e.
  such that 
$\Pi_\Sec(1) = \Pi(1)$, and
  \[
  \begin{array}{rl}
    \Pi_\Sec(2) = & \lambda x_A.\lambda x_B. \nu y'.\\
                 & \Secret(x_A, x_B, y') \\
                 & \rcv(\enca(\pair{z'}{x_A},\pub(x_B)));\\
                 & \snd(\enca(\pair{z'}{y'},\pub(x_A)));\\
                 & \rcv(\enca(y',\pub(x_B)))
  \end{array}
  \]
  An attack on the secrecy of $y'$, is any valid execution trace of $\Pi_\Sec$
  that reveals to the intruder an honest instance of $y'$ (i.e. generated
  by an honest session of $\Pi_\Sec(2)$).  Formally, an attack on the secrecy of
  $y'$ is a valid execution trace of $\Pi_\Sec$ that satisfies the following
  attack formula
  \[
  \overline{\phi_\Sec} = \exists y_A.\exists y_B. \exists x.\ [(\Diamond\Secret(y_A, \dots, y_B, x))\ \wedge\ \Nocomp(y_A)\ \wedge\ \Nocomp(y_B)\ \wedge\ \learn(x)].
  \]
  
  Let us consider as initial intruder knowledge $T_0 = \{a, b, c,
  \priv(c)\}$, the scenario $\sce =
  (1,s_1)(2,s_2)(2,s_2)(2,s_2)(1,s_1)(1,s_1)$, and the function $\alpha$
  such that $\dom(\alpha) = \{s_1, s_2\}$,\break $\alpha(s_1) = (a,c)$, and
  $\alpha(s_2) = (a,b)$. We denote by $\tr$  the  symbolic trace
  associated to $\Pi_\Sec$, $\sce$, and $\alpha$.
Let $\sigma$ be the substitution such that  $\sigma = \{z^{s_1}
  \mapsto n_{y'}^{s_2}, \; {z'}^{s_2} \mapsto n_y^{s_1} \}$. The
  execution trace~$\tr\sigma$ is valid w.r.t.~$T_0$.
This execution
  corresponds to the famous attack due to Lowe~\cite{lowe96breaking},
  and formally satisfies $\overline{\phi_\Sec}$, i.e. $\langle
  T_0, \tr\sigma\rangle\ \models\ \overline{\phi_\Sec}$, 
and thus {$ \Pi_\Sec\ \not\models\ \phi_\Sec$ w.r.t. $T_0$}.
\end{exa}

We are now going to look at how to formally express authentication properties. 
\subsubsection{Aliveness}
\label{subsubsec:aliveness}

We start with the weakest notion of authentication in the hierarchy of
Lowe~\cite{Lowe97}, namely aliveness. Informally, a protocol~$\Pi$
satisfies aliveness if and only if each time a participant~$a$
finishes an honest session involving participant~$b$ (of any of the
roles of~$\Pi$), $b$ has at least partially executed one session (of
any of the roles of~$\Pi$), 
and in that sense~$b$ is alive.

In order to express this property, we need to detect in the executions of~$\Pi$, each time a session starts and ends. This can be achieved by adding status events at the beginning and the end of each role. More precisely, if we consider the $k$-party protocol $\Pi$ with
\[
\Pi(j) = \lambda x_1^j.\dots \lambda x_k^j.\nu y^j_1.\dots \nu y^j_{p_j}.\evt^j_1;\dots; \evt^j_{\ell_j} \quad \text{for } 1\le j\le k.
\]
We build the protocol $\Pi_\Aliv$ by inserting new status events as follows:
\[
\Pi_\Aliv(j) = \lambda x_1^j.\dots \lambda x_k^j.\nu y^j_1.\dots \nu y^j_{p_j}. \Start(x^j_j); \evt^j_1;\dots; \evt^j_{\ell_j}; \End(x^j_1, \dots, x^j_k) \quad \text{for } 1\le j\le k.
\]
where the predicates $\Start$ and $\End$ will mark in an execution the beginning and the end of each session, and will link together the effective participants of each session. Aliveness can then be modelled by the following formula
\[
\phi_\Aliv = 
\left\{ 
\begin{array}{c} 
 \forall y_1.\dots \forall y_k. \big[ \End(y_1, \dots, y_k)\ \wedge\ \Nocomp(y_1)\ \wedge\ \dots\ \wedge\
  \Nocomp(y_k)\ \hspace{3cm}\\
\hfill  \Rightarrow\ \Diamond \Start(y_1) \wedge \ldots \wedge
\Diamond \Start(y_k)\big]
\end{array}
\right.
\]
An attack on protocol $\Pi$ w.r.t. aliveness is thus a trace of $\Pi_\Aliv$ satisfying the following attack formula
\[
\overline{\phi_\Aliv} =
\left\{
\begin{array}{c}
\exists y_1.\dots \exists y_k.\big[ \End(y_1, \dots, y_k)\ \wedge\
\Nocomp(y_1)\ \wedge\ \dots\ \wedge\ \Nocomp(y_k)\ \hspace{3cm} \\
\hfill \wedge\ \big(\neg\Diamond \Start(y_1) \vee \ldots \vee
\neg\Diamond \Start(y_k) \big) \big]
\end{array}
\right.
\]

\begin{exa}
  \label{ex:aliveness}
  Let us come back to the Needham-Schroeder protocol as presented in Example~\ref{ex:syntax} to illustrate this property. In order to do so, we build the 2-party protocol $\Pi_\Aliv$ following the above mentioned construction, i.e. such that
  \[
  \begin{array}{rl}
    \Pi_\Aliv(1)=& \lambda x_A.\lambda x_B.\nu y.\\
    & \Start(x_A) \\
    & \snd(\enca(\pair{y}{x_A},\pub(x_B)));\\
    & \rcv(\enca(\pair{y}{z},\pub(x_A)));\\
    & \snd(\enca(z,\pub(x_B))) \\
    & \End(x_A, x_B)
  \end{array}\quad
  \begin{array}{rl}
    \Pi_\Aliv(2)=& \lambda x_A.\lambda x_B. \nu y'.\\
    & \Start(x_B) \\
    & \rcv(\enca(\pair{z'}{x_A},\pub(x_B)));\\
    & \snd(\enca(\pair{z'}{y'},\pub(x_A)));\\
    & \rcv(\enca(y',\pub(x_B))) \\
    & \End(x_A, x_B)
  \end{array}
  \]
  Now this protocol satisfies aliveness, if in every valid execution trace of $\Pi_\Aliv$ during which an agent $a$ executing an honest session of role $\Pi_\Aliv(1)$ (\emph{resp.} $\Pi_\Aliv(2)$) with agent $b$, $b$ has also initiated a session of the protocol. Formally, $\Pi$ satisfies aliveness if every valid execution trace of $\Pi_\Aliv$ satisfies the following formula; 
  \[
\phi_\Aliv = 
  \begin{array}{c}
     \forall x_A. \forall x_B. \;
        \End(x_A, x_B)\ \wedge\ \Nocomp(x_A)\ \wedge\ \Nocomp(x_B)\
        \Rightarrow\ [\Diamond \Start(x_A) \wedge \Diamond \Start(x_B)] \\
  \end{array}
  \]
Consider the symbolic trace $\tr$ associated to the scenario 
\[
\sce_\Aliv =
  (1,s_1)(1,s_1)(2,s_2)(2,s_2)(2,s_2)(1,s_1)(1,s_1)(1,s_1)(2,s_2)(2,s_2)
\]
  and the function $\alpha$ as defined in
  Example~\ref{ex:secrecy}. Actually, we have that $\langle T_0,
  \tr\sigma \rangle \models \phi_\Aliv$ using the set $T_0$ and the substitution
  $\sigma$ as defined in Example~\ref{ex:secrecy}. More generally, 
 using an automatic tool such as ProVerif~\cite{blanchet01}, one can
 prove that the 
Needham-Schroeder protocol satisfies aliveness w.r.t. the
initial intruder knowledge $T_0 = \{a, b, c, \priv(c)\}$, i.e.
$\Pi_\Aliv\models \phi_\Aliv$ w.r.t. $T_0$.
\end{exa}

\subsubsection{Weak agreement}
\label{subsubsec:weakagreement}

Weak agreement is slightly stronger than aliveness. Informally, a
protocol $\Pi$ satisfies weak agreement, if and only if each time a
participant $a$ finishes an honest session involving participant $b$ 
(of any of the roles of $\Pi$), $b$ has at least initiated a session involving $a$ (of any of the roles of $\Pi$).

Again, in order to express this property, we need to detect in the executions of $\Pi$, each time a session starts and ends, but also which participants are involved in each session that is initiated. This can be achieved by adding status events at the beginning and the end of each role. More precisely, if we consider the $k$-party protocol $\Pi$ with
\[
\Pi(j) = \lambda x_1^j.\dots \lambda x_k^j.\nu y^j_1.\dots \nu y^j_{p_j}.\evt^j_1;\dots; \evt^j_{\ell_j} \quad \text{for } 1\le j\le k.
\]
We build the protocol $\Pi_\Weak$ by inserting new status events as follows:
\[
\begin{array}{l}
\Pi_\Weak(j) = \lambda x_1^j.\dots \lambda x_k^j.\nu y^j_1.\dots \nu
y^j_{p_j}. \Start(x^j_j, x^j_1); \dots; \Start_{j,k}(x^j_j, x^j_k); \hspace{2cm}\\
\hfill \evt^j_1;\dots; \evt^j_{\ell_j}; \End_j(x^j_1, \dots, x^j_k)
\quad\quad \text{for } 1\le j\le k.
\end{array}
\]
where the predicates $\Start$ and $\End$ will mark in an execution the beginning and the end of each session, and will link together the effective participants of each session both at the beginning and the end of the session. Weak agreement can then be modelled by the following formula
\[
\begin{array}{c}
  \phi_\Weak = \forall y^1_1.\dots \forall y^1_k. \dots \forall y^k_1.\dots \forall y^k_k. \\[2mm]
  \underset{j\in \{1,\ldots, k\}}{\bigwedge}
  \left[
    \begin{array}[c]{l}
      \End_j(y^j_1, \dots, y^j_k)\ \wedge\ \Nocomp(y^j_1)\ \wedge\
      \dots\ \wedge\ \Nocomp(y^j_k)\ \Rightarrow\ \underset{i\in
        \{1,\ldots, k\}, i\not=j}{\bigwedge}\ \Diamond \Start(y^j_i, y^j_j)
    \end{array}
  \right]
\end{array}
\]
An attack on protocol $\Pi$ w.r.t. aliveness is thus a trace of $\Pi_\Weak$ satisfying the following attack formula
\[
\begin{array}{c}
  \overline{\phi_\Weak} \equiv \exists y_1.\dots \exists y_k.\ \End_j(y_1, \dots, y_k)\ \wedge\ \Nocomp(y_1)\ \wedge\ \dots\ \wedge\ \Nocomp(y_k)\ \wedge\  \neg\Diamond \Start(y_i, y_j)
\end{array}
\]
for some $j, i\in \{1,\ldots, k\}$ with $i\not=j$.

\begin{exa}
  \label{ex:weakagreement}
  Let us come back to the Needham-Schroeder protocol as presented in Example~\ref{ex:syntax} to illustrate this property. In order to do so, we build the 2-party protocol $\Pi_\Weak$ following the above mentioned construction, i.e. such that:
  \[
  \begin{array}{rl}
    \Pi_\Weak(1)=& \lambda x_A.\lambda x_B.\nu y.\\
    & \Start(x_A, x_A) \\
    & \Start(x_A, x_B) \\
    & \snd(\enca(\pair{y}{x_A},\pub(x_B)));\\
    & \rcv(\enca(\pair{y}{z},\pub(x_A)));\\
    & \snd(\enca(z,\pub(x_B))) \\
    & \End_1(x_A, x_B)
  \end{array}\quad
  \begin{array}{rl}
    \Pi_\Weak(2)=& \lambda x_A.\lambda x_B. \nu y'.\\
    & \Start(x_B, x_A) \\
    & \Start(x_B, x_B) \\
    & \rcv(\enca(\pair{z'}{x_A},\pub(x_B)));\\
    & \snd(\enca(\pair{z'}{y'},\pub(x_A)));\\
    & \rcv(\enca(y',\pub(x_B))) \\
    & \End_2(x_A, x_B)
  \end{array}
  \]
  
  Now this protocol satisfies weak agreement, if in every valid
  execution trace of~$\Pi_\Weak$ during which an agent~$a$ executing
  an honest session of role~$\Pi_\Weak(1)$ (\emph{resp.}
  $\Pi_\Weak(2)$) with agent~$b$, $b$ has also initiated a session of
  the protocol involving agent~$a$. In other words, $\Pi$ admits an
  attack w.r.t. weak agreement if there exists a valid
  execution trace of $\Pi_\Weak$ that satisfies the 
following formula:
  \[
  \begin{array}{c}
    \overline{\phi_\Weak} \equiv \exists x^1_A. \exists x^1_B. \exists x^2_A. \exists x^2_B. \\ [2mm]
    \left[
      \begin{array}[c]{c}
        \End_1(x^1_A, x^1_B)\ \wedge\ \Nocomp(x^1_A)\ \wedge\ \Nocomp(x^1_B)\ \wedge\ \neg\Diamond \Start(x^1_B, x^1_A) \\
        \vee \\
        \End_2(x^2_A, x^2_B)\ \wedge\ \Nocomp(x^2_A)\ \wedge\ \Nocomp(x^2_B)\ \wedge\ \neg\Diamond \Start(x^2_A, x^2_B)
  \end{array}
    \right]
  \end{array}
  \]

  \noindent Let's consider as initial intruder knowledge $T_0 = \{a, b, c,
  \priv(c)\}$, the scenario
  \[
  \sce =
  (1,s_1)(1,s_1)(1,s_1)(2,s_2)(2,s_2)(2,s_2)(2,s_2)(1,s_1)(1,s_1)(1,s_1)(2,s_2)(2,s_2)(2,s_2),
  \]
  the function~$\alpha$ such that $\dom(\alpha) = \{s_1, s_2\}$, $\alpha(s_1) =
  (a,c)$, and $\alpha(s_2) = (a,b)$, and the substitution $\sigma = \{z^{s_1}
  \mapsto n_{y'}^{s_2}, \; {z'}^{s_2} \mapsto n_y^{s_1} \}$. The execution
  trace~$\tr\sigma$ is valid w.r.t.~$T_0$ with $\tr$ the symbolic trace
  associated to $\sce$ and $\alpha$. This execution corresponds to the famous
  attack due to Lowe~\cite{lowe96breaking}, and formally satisfies
  $\overline{\phi_\Weak}$, i.e. $\langle T_0, \tr\sigma\rangle\ \models\
  \overline{\phi_\Weak}$, and thus $\Pi_\Weak\ \not\models\ \phi_\Weak$
  w.r.t.~$T_0$.

\end{exa}

\section{Transformation of protocols}
\label{sec:transfo}

In Section~\ref{subsec:transfo} we define our transformation before we
state our main result in Section~\ref{subsec:result} whose proof is
postponed to Part~II.  
%Finally, we discuss the tags which are used in
%our transformation in Section~\ref{subsec:othertags}.

%
% Our transformation
%

\subsection{Our transformation}
\label{subsec:transfo}

Given an input protocol $\Pi$, our transformation will compute a new
protocol~$\widetilde{\Pi}$ which {consists in} two phases. During the first
phase, the protocol participants try to agree on some common,
dynamically generated, session identifier~$\tau$. For this, each
participant sends a freshly generated nonce~$N_i$ together with his
identity~$A_i$ to all other participants. (Note that if broadcast is
not practical or if not all identities are known to each participant,
the message can be sent to some of the participants who forwards the
message.) At the end of this preamble, each participant computes a
session identifier: $\tau = \langle \pair{A_1}{N_1}, \ldots,
\pair{A_k}{N_k} \rangle $.  Note that an active attacker may interfere
with this initialization phase and may intercept and replace some of
the nonces. Hence, the protocol participants do not necessarily agree
on the same session identifier~$\tau$ after this preamble. In fact,
each participant computes his own 
session identifier, say~$\tau_j$.
During the second phase, each participant~$j$ executes the original protocol in
which the dynamically computed identifier is used for tagging each application
of a cryptographic primitive.  In this phase, when a participant opens an
encryption, he checks that the tag is in accordance with the nonces he received
during the initialization phase. In particular, he can test the presence of his
own nonce.

The transformation, using the informal Alice-Bob notation, is described below and relies on the tagging operation that is formally defined in Definition~\ref{def:tagging}. 

\begin{minipage}{3.5cm}
  \[
  \Pi = 
  \left\{
    \begin{array}{ll}
      A_{i_1} \to A_{j_1}:& m_1\\
      \multicolumn{2}{c}\vdots\\
      A_{i_\ell} \to A_{j_\ell}:& m_\ell 
    \end{array}
  \right.
  \]
\end{minipage}
\quad\quad\quad\quad
\begin{minipage}{5cm}
\[
\widetilde{\Pi} = 
\left\{
\begin{array}{llcll}
\multicolumn{2}{l}{\mbox{Phase 1}}
&\;\;\;\;\;&
\multicolumn{2}{l}{\mbox{Phase 2}}\\[2mm]
A_1 \to \mathit{All}:& \langle A_1, N_1\rangle
&& 
A_{i_1} \to A_{j_1}:& \lbrack m_1 \rbrack_{\tau}\\
\multicolumn{2}{c}\vdots & \multicolumn{3}{c}\vdots\\
A_k \to \mathit{All}:& \langle A_k, N_k \rangle
&&
A_{i_\ell} \to A_{j_\ell}:& \lbrack m_\ell \rbrack_\tau\\[3mm]
\multicolumn{5}{l}{\mbox{where}\; \tau = \langle \montag_1, \ldots, \montag_k 
   \rangle \; \mbox{with $\montag_i = \langle A_i, N_i
    \rangle$}}
\end{array}
\right.
\]
\end{minipage}

\medskip{}

Note that, the Alice-Bob notation only represents what happens in a normal execution, i.e. with no intervention of the attacker. Of course, in such a situation, the participants agree on the same session identifier~$\tau$ used in the second phase.

\begin{defi}[$k$-tag, $k$-tagging]
  \label{def:tagging} 
 A \emph{$k$-tag} is a term $\langle \langle a_1, v_1 \rangle, \ldots, \langle
  a_k, v_k\rangle \rangle$ where each $a_i \in \A$ and each $v_i$ is a term.
Let $u$ be a term and $\montag$ be a $k$-tag.  The $k$-tagging of~$u$ with~$\montag$,
  denoted~$[u]_\montag$, is inductively defined as follows:
  \[
  \begin{array}{lclr}
    \lbrack\pair{u_1}{u_2}\rbrack_\montag &\; =\; & \pair{\lbrack u_1\rbrack_\montag}{\lbrack u_2\rbrack_\montag} \\
    \lbrack \f(u_1,u_2)\rbrack_\montag & = & \f(\pair{\montag}{\lbrack u_1\rbrack_\montag},\lbrack u_2\rbrack_\montag) & \;\;\;\;\;\;\;~\mbox{for}~\f\in\{\enc,\enca,\sign\} \\
    \lbrack \h(u_1)\rbrack_\montag & = & \h(\pair{\montag}{\lbrack u_1\rbrack_\montag}) \\
    \lbrack u\rbrack_\montag & = & u & \mbox{otherwise} \\
  \end{array}
  \]
  We say that a term $t$ is \emph{$k$-tagged} if $u|_{1.1}$ is a $k$-tag for any $u \in \encst(t)$.
\end{defi}

These notions are extended to events and sequences of events as expected. We are now able to formally define our transformation.

\begin{defi}[protocol transformation]
  \label{def:transfo}
  Let~$\Pi$ be a $k$-party protocol such that 
  \[
  \Pi(j) = \lambda x^j_1\dots\lambda x^j_k.\nu y^j_1\dots\nu y^j_{p_j}.\seqevt^j \hspace{0.5cm} \mbox{for}\; 1\leq j\leq k.
  \]
  \noindent and the variables $z_i^j$ ($1 \leq i,j \leq k$) do not appear in~$\Pi$ (which can always be ensured by renaming variables in $\Pi$).  The transformed protocol $\widetilde{\Pi}$ is a $k$-party protocol defined as follows:
  \[
  \widetilde{\Pi}(j) =  \lambda x^j_1\dots\lambda x^j_k.\nu y^j_1\dots\nu y^j_{p_j}. \nu z^j_j. \widetilde{\Pi}^{\init}(j) ; \lbrack\seqevt^j\rbrack_{\tau_j}\hspace{0.5cm} \mbox{for}\; 1\leq j \leq k
  \]
  \noindent where   $$\widetilde{\Pi}^{\init}(j) = \rcv(u_1^j); \ldots ; \rcv(u_{j-1}^j); \snd(u_j^j) ; \rcv(u_{j+1}^j) ; \ldots ; \rcv(u_k^j)$$ and
  $\tau_j = \langle u_1^j, \ldots, u_k^j\rangle \mbox{ with }u_i^j = \langle x_i^j, z_i^j \rangle$.
\end{defi}

In the above definition, the protocol $\widetilde{\Pi}^{\init}$ models the
initialization phase and the variables~$z_i^j$ correspond to the
nonces that are generated and exchanged during this phase. In
particular for the role~$j$, the variable~$z_j^j$ is a freshly
generated nonce while the other variables~$z_i^j$,~${i \not = j}$, are
expected to be bound to the other participant's nonces in the receive
events. Remember also that the variables~$x_i^j$ are the role
parameters which correspond to the agents. The tag computed by the
$j\textsuperscript{th}$ role in our transformation consists in the
concatenation of the $k$ names of the agents involved in the protocol,
together with the $k-1$ terms received during the initialization phase
as well as the fresh nonce generated by the role~$j$ itself,
i.e.~$z_j^j$. We illustrate this transformation on the
Needham-Schroeder protocol introduced in Section~\ref{sec:model}.

\begin{exa} 
  \label{ex:ns}
  Consider the Needham-Schroeder protocol described in Example~\ref{ex:syntax}. Applying our transformation we obtain a 2-party protocol $\widetilde{\Pi}$.  The role $\widetilde{\Pi}(2)$ is described below. The role $\widetilde{\Pi}(1)$ can be obtained in a similar way.
  \[
  \begin{array}{rl}
    \widetilde\Pi(2)= & \lambda x_A\lambda x_B. \nu y'. \nu z_B. \rcv(\langle x_A, z_A \rangle);  \snd(\langle x_B, z_B\rangle); \\
    & \rcv(\enca(\pair{\tau}{\pair{z'}{x_A}}, \pub(x_B))); \\
    & \snd(\enca(\pair{\tau}{\pair{z'}{y'}}, \pub(x_A))); \\
    & \rcv(\enca(\pair{\tau}{y'}, \pub(x_B)))
  \end{array}
  \]
  
  \noindent where $\tau = \langle \pair{x_A}{z_A},
  \pair{x_B}{z_B}\rangle$. Note that Lowe's famous man-in-the-middle
  attack~\cite{lowe96breaking} described in
  Example~\ref{ex:weakagreement} does not exist anymore on
  $\widetilde{\Pi}$. 
%Indeed, one can prove this using the ProVerif tool~\cite{}.
\end{exa}

%%%%%%%%%%%%%%%%%%%%%%%%%%%%%%%%%%%%%%%%%%%%%%%%%%%%%%%%%%%%%%%%%%%%%%%%%%
%                                                                        %
%                              MAIN THEOREM                              %
%                                                                        %
%%%%%%%%%%%%%%%%%%%%%%%%%%%%%%%%%%%%%%%%%%%%%%%%%%%%%%%%%%%%%%%%%%%%%%%%%%                    

\subsection{Main theorem}
\label{subsec:result}

Roughly, our result states that if the compiled protocol admits an
attack that may involve several sessions, then there exists an attack
which only requires a bounded number of sessions of each role, and the
bound only depends on the security formula under study. More formally,
we define the size of a formula as follows:

\begin{defi}[size of a formula]\label{def:sizePhi}
  Let $\phi$ be a  formula. The \emph{size} of $\phi$, denoted $\|\phi\|$, is defined as follows:
  % \[
  % \|\phi\| = \mathsf{M}^+(\phi) % + \mathsf{NbLearn}(\phi)
  % \]
  \[
  \begin{array}{rcl c rcl}
    \|\true\| & \stackrel{\mathsf{def}}{=} & 0 &    \ \ \ \ & \|\true\|^- & \stackrel{\mathsf{def}}{=} & 0 \\
    \|\Prd(t_1, \dots, t_n)\| & \stackrel{\mathsf{def}}{=} & 1 &    & \|\Prd(t_1, \dots, t_n)\|^- & \stackrel{\mathsf{def}}{=} & 1 \\
    \|\learn(t)\| & \stackrel{\mathsf{def}}{=} & 0 & &    \|\learn(t) \|^- & \stackrel{\mathsf{def}}{=} & 0 \\
    \|\mathsf{C}(t)\| & \stackrel{\mathsf{def}}{=} & 0 & &    \|\mathsf{C}(t)\|^- & \stackrel{\mathsf{def}}{=} & 0 \\
    \|\neg\phi\| & \stackrel{\mathsf{def}}{=} &    \|\phi\|^- & & \|\neg\phi\|^-    & \stackrel{\mathsf{def}}{=} & \|\phi\| \\
    \|\phi_1 \vee \phi_2\| & \stackrel{\mathsf{def}}{=} &    \max\{\|\phi_1\|, \|\phi_2\|\} & &
    \|\phi_1 \vee \phi_2\|^- & \stackrel{\mathsf{def}}{=} &     \|\phi_1\|^- + \|\phi_2\|^- \\
   \|\exists x.\ \phi\| & \stackrel{\mathsf{def}}{=} &    \|\phi\| & & \|\exists x.\ \phi\|^-    & \stackrel{\mathsf{def}}{=} & \|\phi\|^- \\
    \|\Diamond\phi\| & \stackrel{\mathsf{def}}{=} &    \|\phi\| & & \|\Diamond\phi\|^-  & \stackrel{\mathsf{def}}{=} & 0
  \end{array}
  \]
\end{defi}

{Intuitively, when an attack trace involves several sessions
  of each role, not all the sessions are necessary to mount the
  attack. We only need to keep those sessions that witness the
  satisfiability of the attack formula $\phi$.  By definition of an
  attack formula (see Definition~\ref{def:attack}), we know that each
  variable occurring in $\phi$ also occurs in a positive status
  events. Thus, there is no need to take into account the number of
  occurrences of $\learn(t)$ in the previous definition.}
 
\medskip{}

\begin{exa}
\label{ex:size}
{Note that $\|\phi_1 \wedge \phi_2\| = \|\phi_1\| +
\|\phi_2\|$.} Considering the attack formulas $\overline{\phi_\Sec}$, $\overline{\phi_\Aliv}$,
and $\overline{\phi_\Weak}$ as defined in Section~\ref{subsec:ex-properties}, we
have that $\|\overline{\phi_\Sec}\| = \|\overline{\phi_\Aliv}\| =  \|\overline{\phi_\Weak}\| = 1$.
% Details for secrecy
  % \noindent $
  % \begin{array}{rcl}
  %   \|\overline{\phi_\Sec}\| & = & 
  %   \left\|
  %     \begin{array}[c]{l}
  %       \exists x_1.\dots.\exists x_k. \exists y. \\
  %       \quad \quad \lbrack(\Diamond\Secret(x_1, \dots, x_k, y))\ \wedge\ \Nocomp(x_1)\ \wedge\ \dots\ \wedge\ \Nocomp(x_k)\ \wedge\ \learn(y)\rbrack
  %     \end{array}
  %   \right\| \\
  %   \\
  %   & = &     \left\|
  %     \begin{array}[c]{l}
  %       \quad \quad \lbrack(\Diamond\Secret(x_1, \dots, x_k, y))\ \wedge\ \Nocomp(x_1)\ \wedge\ \dots\ \wedge\ \Nocomp(x_k)\ \wedge\ \learn(y)\rbrack
  %     \end{array}
  %   \right\| \\
  %   \\
  %   & = &     \left\|
  %     \begin{array}[c]{l}
  %       \quad \quad \Diamond\Secret(x_1, \dots, x_k, y))
  %     \end{array}
  %   \right\| \\
  %   & = & 1
  % \end{array}
  % $
\end{exa}

\noindent We are now able to state our main transference result.

\newsavebox{\theotransfoun}
\sbox{\theotransfoun}{\vbox{%
\begin{thm}
  \label{theo:transfo1}
  Let~$\Pi$ be a $k$-party protocol, $\widetilde{\Pi}$ be its
  corresponding transformed protocol and $T_0$ be a set of ground
  atoms such that $\lgAtom(\Pi) \cap \plaintext(\Pi) \subseteq T_0
  \cup \Kcal_\epsilon$. Let $\phi$ be an attack formula such that
  $\widetilde{\Pi} \models \phi$ w.r.t. $T_0$.  There exists a
  valid execution trace~$\exec$ of~$\widetilde{\Pi}$ such that:
  \begin{center}
    $\langle \exec,T_0 \rangle \models \phi$ and $\exec$ involves at
    most $\|\phi\|$ sessions of each role.
  \end{center} 
\end{thm}
}}

\noindent\usebox{\theotransfoun}

Applying our result, we can now establish that if a protocol built
according to our transformation admits an attack on secrecy
(resp. aliveness, weak agreement), then it admits an attack that
involves at most one session of each role.  The situation is however
slightly more complicated than it may seem at first sight. As we have
an infinite number of agent names there is an infinite number of
sessions, which one would need to verify separately. Actually we can
avoid this combinatorial explosion thanks to the following well-known
result~\cite{esop03journal}: when verifying secrecy properties it is
sufficient to consider two agents (an honest agent and a dishonest
one). Hence, using this result, we can instantiate all the parameters
using only two agent names. Similar reduction results also exist for
authentication properties (see~\cite{esop03journal}).

{Note that we only consider protocols whose long-term secret
  keys do not occur in plaintext position. This assumption is required
  to ensure that the ``small scenario'', (\emph{i.e.}, the one that
  involves only $\| \phi \|$ sessions of each role) will violate the
  same security property~$\phi$. We may actually relax this assumption
  if we consider an execution trace that reveals such a long-term key
  as a violation of the security property as well. The result is
  stated in this way in~\cite{arapinis-these2008}.}

Actually, for the security properties presented in the previous
section, we can go even further and only consider one \emph{honest}
session of each role.

% \stef{Est-ce vrai aussi pour les autres proprietes?}
% \stef{J'imagine qu'il faut aussi correctement mettre les events dans
%   les protocols}

% \newsavebox{\corsecrecy}
% \sbox{\corsecrecy}{\vbox{%
%     \begin{corollary}
%       \label{cor:secrecy}
%       Let~$\Pi_\Sec$ be a $k$-party protocol and $\phi_\Sec$ be the secrecy
%       property has defined in Section~\ref{subsubsec:secrecy}. Let $T_0$ be
%       a set of ground atoms such that $\langle T_0,
%       \widetilde{\Pi_\Sec}\rangle \models \overline{\phi_\Sec}$ where
%       $\widetilde{\Pi_\Sec}$ is the transformed protocol associated to
%       $\Pi_\Sec$. Assume that $\lgAtom(\Pi_\Sec) \cap \plaintext(\Pi_\Sec)
%       \subseteq T_0 \cup \Kcal_\epsilon$. There exists a valid execution
%       trace $\exec$ of $\widetilde{\Pi_\Sec}$ such that:
%       \begin{center}
%         $\langle \exec,T_0 \rangle \models \overline{\phi_\Sec}$ and
%         $\exec$ involves at most one \emph{honest} session of each role.
%       \end{center} 
%     \end{corollary}
%   }}

\newsavebox{\corsecrecy}
\sbox{\corsecrecy}{\vbox{%
    \begin{corollary}
      \label{cor:secrecy}
      Let~$\Pi$ be a $k$-party protocol, $\Pi_\Sec$ (respectively,
      $\Pi_\Aliv$, $\Pi_\Weak$) be the annotated protocol for modeling
      secrecy (respectively aliveness and weak agreement) as defined
      in Section~\ref{subsubsec:secrecy}, and $\widetilde\Pi_\Sec$
      (respectively, $\widetilde\Pi_\Aliv$, $\widetilde\Pi_\Weak$) the
      corresponding transformed protocol. Let $T_0$ be a set of ground
      atoms such that $\lgAtom(\Pi) \cap \plaintext(\Pi) \subseteq T_0
      \cup \Kcal_\epsilon$ and $\overline{\phi_\Sec}$ (respectively
      $\overline{\phi_\Aliv}$, $\overline{\phi_\Weak}$) an attack
      formula against secrecy (respectively aliveness and weak
      agreement) as defined in Section~\ref{subsubsec:secrecy}.  For
      ${\sf X} \in \{ \Sec, \Aliv, \Weak \}$ we have that if
      $\widetilde{\Pi}_{\sf X} \models \overline{\phi_{\sf X}}$
      w.r.t. $T_0$ then there exists a valid execution trace
      $\exec$ of $\widetilde\Pi_{\sf X}$ such that:
      \begin{center}
        $\langle \exec,T_0 \rangle \models \overline{\phi_{\sf X}}$
        and $\exec$ involves at most one \emph{honest} session of each
        role.
      \end{center} 
    \end{corollary}
  }}

\noindent\usebox{\corsecrecy}

\subsection{{Alternative} ways of tagging protocols.}

{Our transformation is computationally light as it does not add any cryptographic application. However, it increases significantly the size of messages to be encrypted and signed. As an alternative, we may choose to hash the tags. Our results still hold in this setting.}

We have also considered an alternative, slightly different transformation that does not include the identities in the tag, i.e., the tag is simply the sequence of nonces. Our main result, Theorem~\ref{theo:transfo1}, still holds as the proof does not use the presence of the identities. However, the stronger results presented on particular properties stated in Corollary~\ref{cor:secrecy}, do not hold anymore, as 
the proof crucially relies on the presence of the agent names in the tag. When omitting identities, even for secrecy, we need to additionally check for 
attacks that involve a session engaged with the attacker. Indeed, on the example of the Needham-Schroeder protocol the man-in-the-middle attack is not prevented by this weaker tagging scheme. However, the result requires one to also consider one dishonest session for each role, hence including the attack scenario.  In both cases, it is important for the tags to be \emph{collaborative}, i.e. all participants do contribute by adding a 
fresh nonce.

\newpage

\null \hfill {\large\bf  --- PART II: \textsc{Proof of our reduction result} ---} \hfill\null
\bigskip{}

In this part, we give an overview of the proof of our reduction result stated in
Theorem~\ref{theo:transfo1}. Assume that our protocol $\widetilde{\Pi}$ admits
an attack.
\begin{enumerate}
\item We first show that there is an attack on a \emph{well-formed} execution
  trace (Section~\ref{sec:1-step}). In a well-formed execution trace (see
  Definition~\ref{def:wfexec}), terms are necessarily tagged with the expected
  tag, i.e. the tag computed during the initialization phase. Moreover,
  only names coming from sessions tagged in the same way can be used in the
  events of those sessions.  In order to prove this, we define a transformation
  $\overline{\; \cdot \;}$ that transforms an execution trace {to} a well-formed
  one by abstracting some subterms (those that are not tagged properly using the
  expected tag) by fresh nonces. We show that this transformation preserves the
  validity of the trace (Proposition~\ref{prop:bar.valid}) as well as the
  satisfiability of the attack formula under study
  (Proposition~\ref{prop:bar.attack}).

\item Then, given a set of sessions $S$ and a valid and well-formed execution
  $\exec$ that satisfies the attack formula, we show that $\exec|_S$,
  i.e. the restriction of $\exec$ to the events coming from a session in
  $S$ is still an execution satisfying the attack formula. Since messages coming
  from one session can be used to build a message for another session, this can
  only be achieved by requiring some conditions on $S$. Basically, to ensure the
  validity of the execution $\exec|_S$, we have to ensure that sessions that
  share the same tag are either all in $S$ or none of them is in~$S$ (see
  Proposition~\ref{pro:validity-restriction}). Then, to ensure the
  satisfiability of the attack formula, we have to keep enough sessions but we
  can bound a priori the number of sessions that is needed to mount an attack
  (see Proposition~\ref{pro:attacktildepibis}).
\end{enumerate}

\section{First step: towards a well-formed execution trace}
\label{sec:1-step}

In this section, we formally define our notion of \emph{well-formedness} and we
propose a transformation that allows us to transform a trace $\exec$ into a
well-formed one $\overline{\exec}$ (Section~\ref{subsec:1-step-transfo})
preserving its validity (Section~\ref{subsec:validity}) and the satisfiability
of the attack formula (Section~\ref{subsec:sat}).

\subsection{Well-formed}

The idea behind our notion of well-formedness is to ensure that each
term will be properly tagged. Basically, this means that each term has to be
tagged with its expected tag, i.e. the one computed during the
initialization phase of the protocol (phase~$1$). 
From now on, when we
consider a trace $\exec$ issued from a protocol~$\widetilde{\Pi}$, we assume that
the events occurring in $\exec$ are annotated with their session
identifier, and we write $\exec = [\evt_1^{\sid_1}; \ldots ;
\evt_\ell^{\sid_\ell}]$ when we want to refer to these annotations explicitly.

{The transformation that we consider will abstract some
  subterms by fresh names from the intruder's knowledge (\emph{i.e.}
  names in~$\Nces_\epsilon$). Those names will be denoted by
  $n^{\epsilon,S}_t$ where $S$ is set of session identifiers, and $t$
  is a term. Intuitively, such a name will be used to abstract the
  subterm $t$ when used in an event from a session $\sid \in S$.  We
  assume that those names (which constitute an infinite subset of
  $\Nces_\epsilon$) are not used anywhere else. In particular, they do
  not occur in the execution trace before applying our
  transformation.}

\begin{defi}[$\tagt(\exec, \sid)$]\label{def:tagt}
  Let $\Pi$ be a $k$-party protocol and let $\exec = [\evt_1^{\sid_1};
  \ldots ; \evt_\ell^{\sid_\ell}]$ be an execution trace of
  $\widetilde{\Pi}$. Let $\sid$ be a session identifier and
  $[\evt_{i_1}^{\sid}; \ldots ; \evt_{i_h}^{\sid}]$ (with $1 \leq i_1
  < \ldots < i_h \leq \ell$) be the sequence of communication events in
  $\exec$ that are annotated with $\sid$. We define the \emph{expected
    tag of a session~$\sid$ in $\exec$} as
  \begin{itemize}
  \item $\tagt(\exec, \sid) = \bot$ when $h < k$,
  \item $\tagt(\exec, sid) = \langle m_1, \dots, m_k\rangle$
    otherwise, where for all $j \in \{1, \dots, k\}$, $m_j$ is such
    that $\evt^{\sid}_{i_j} = \rcv(m_j)$ or $\evt^{\sid}_{i_j} =
    \snd(m_j)$.
  \end{itemize}
\end{defi}
% \begin{defi}[$\tagt(\exec, \sid)$]\label{def:tagt}
%   Let $\Pi$ be a $k$-party protocol and $\exec$ be an execution trace
%   of $\widetilde{\Pi}$. Let $\sid$ be a session identifier and $h$ be the
%   number of communication events in $\exec$ that are annotated with
%   $\sid$. Let $m_j$ (with $1 \leq j \leq h$) the term that occurs in
%   the $j$\textsuperscript{th} communication event of $\exec$ among
%   those that are annotated with $\sid$.  We define the expected tag in
%   $\tr$ of a session~$\sid$ as follows:

%   \begin{itemize}
%   \item $\tagt(\exec, \sid) = \bot$ when $h < k$,
%   \item $\tagt(\exec, sid) = \langle m_1, \dots, m_k\rangle$ when for
%     all $j \in \{1, \dots, k\}$, $m_j$ is such that $\evt^{\sid}_{j} =
%     \rcv(m_j)$ or $\evt^{\sid}_{j} = \snd(m_j)$.
%   \end{itemize}
% \end{defi}

Roughly, the expected tag associated to a session $\sid$ is the one
obtained by putting together the messages that occur in the $k$ first
communication events annotated with $\sid$ that occur in $\exec$. When
those events do not exist, the expected tag of $\sid$ is undefined.
We define $\tagst(\exec)$ to denote the set of expected tags that
occur in the trace $\exec$. More formally, we have that:
\[
\tagst(\exec) = \underset{\sid}{\bigcup} \; {\{\tagt(\exec, \sid)\}}.
\]

Since a session is the execution of one role, it is likely that
several sessions will have the same expected tag. However, note that
sessions that correspond to the execution of the same role
(e.g. the $j$\textsuperscript{th} role) cannot have the same
expected tag since the tag will contain a fresh nonce at its
$j$\textsuperscript{th} position.

\begin{defi}[$\sameTagas(\exec, \sid)$]\label{def:sametagas}
  Let $\Pi$ be a $k$-party protocol and let $\exec$ be an execution trace (not
  necessarily valid) of $\widetilde{\Pi}$. We define $\sameTagas(\exec,\sid)$ to
  be the set of sessions sharing the same expected tag with the session~$\sid$,
  i.e.
  \[
  \sameTagas(\exec,\sid) =\!
  \left\{
    \begin{array}[c]{l}
      \!\!\{\sid\} \;\, \mbox{if $\tagt(\exec, \sid) = \bot$} \\
      \!\!\{sid' \mid \tagt(\exec, sid') = \tagt(\exec, sid)\} \;\, \mbox{otherwise}
    \end{array}
  \right.
  \]
\end{defi}

Our notion of well-formedness aims to ensure that each event that
occurs in a trace is tagged properly.  For this, we first define
$\Tags(\exec,\sid)$. This set corresponds to the tags that actually
occur in the events issued from the session $\sid$ in the execution
trace $\exec$.

\begin{defi}[$\Tags(\exec,\sid)$]
  Let $\Pi$ be a $k$-party protocol and $\exec = [\evt^{\sid_1}_{1};
  \dots; \evt^{\sid_\ell}_{\ell}] $ be an execution trace of
  $\widetilde{\Pi}$ which is $k$-tagged.  Let $\sid$ be a session
  identifier.
%, and   $ [\evt^{\sid}_{1}; \dots; \evt^{\sid}_{h}]
%\stackrel{\mathsf{def}}{=} \exec|_{\{\sid\}}$. 
  We define the tags of a session~$\sid$ in $\exec$ as follows:
  \[
  \Tags(\exec,\sid) = \{u|_{1.1} ~|~u\in\encst(\evt^{\sid_j}_j) \mbox{
    for some } j \in \{1,\ldots,\ell\} \mbox{ such that } \sid_j = \sid\}.
  \]
\end{defi}

We define $\Tags(\exec)$ to denote the set of tags that occur in the
trace $\exec$. More formally, we have that
\[\Tags(\exec) = \underset{\sid}{\bigcup}\; \Tags(\exec, \sid).
\]

We are now able to define our notion of well-formed execution trace.

\begin{defi}[well-formed execution trace]\label{def:wfexec}
  Let $\Pi$ be a $k$-party protocol, and $\exec = [\evt_1^{\sid_1}; \ldots; \evt_\ell^{\sid_\ell}]$ be an execution trace associated to $\widetilde{\Pi}$. We say that~$\exec$ is \emph{well-formed} if:
  \begin{enumerate}
  \item\label{def:wfexec1} $\exec$ is $k$-tagged, i.e. for all $t\in \st(\exec)$, $t$ is $k$-tagged;
  \item\label{def:wfexec2} $\Tags(\exec,\sid) \subseteq \{\tagt(\exec,\sid)\}$ for every $\sid$;
  \item\label{def:wfexec3} For every $i$, we have that $\names(\evt_i^{\sid_i}) \subseteq \{n_t^{\epsilon,S} ~|~t \in T\} \cup \{n_y^{\sid} ~|~ y\in \Yvars \text{ and }\sid \in S\}$ where $S = \sameTagas(\exec,\sid_i)$.
  \end{enumerate}
\end{defi}

\noindent Intuitively, in a well-formed trace, the events of a session $\sid$
are $k$-tagged with the expected tag, i.e. the tag defined in the preamble of the session
$\sid$. Moreover, the nonces used in a session $\sid$ are those that
are generated in a session that used the same tag as $\sid$ (or they
come from the intruder).

%%
%% The transformation
%%

\subsection{Our transformation of execution traces}
\label{subsec:1-step-transfo}

A valid execution trace is not necessarily well-formed. Our goal is to
show that we can however always transform an execution trace into a
well-formed execution trace. The main idea is to replace each subterm
that is not tagged in the expected way with a nonce known by the
attacker.  The difficulty will be to ensure that the resulting trace
is still a valid one (see Section~\ref{subsec:validity}) and still a
witness of the existence of an attack (see Section~\ref{subsec:sat}).

We first define our transformation on a term. For this we need to
introduce the notion of $\tagh$ 

\begin{defi}[$\tagh(\exec, t)$]\label{def:tagh} 
  Let $\Pi$ be a $k$-party protocol and $\exec$ be an execution trace (not necessarily valid) of $\widetilde{\Pi}$. We define the head tag of a term $t$ w.r.t. the trace $\exec$, denoted $\tagh(\exec, t)$.
  \[
  \tagh(\exec, t) = \left\{
    \begin{array}[c]{lr}
      \tau & \hspace{0.15cm}\text{if } t = \f(\pair{\tau}{u}, u_2, \dots, u_n)\in\encst(t) \\
      & \hfill\text{ and } \tau\in\tagst({\exec}) \\[2mm]
      \bot & \text{otherwise}
    \end{array}
  \right.
  \]
\end{defi}

Roughly, our transformation of a term proceeds as follows.  We replace
each cryptographic subterm which is not tagged properly with a nonce. We
also perform the same kind of replacement on nonces to ensure that
sessions that are tagged differently will not share any nonces.

\begin{defi}[$\overline{t}^{\exec,\sid}$]\label{def:abstract.term}
  Let $\Pi$ be a $k$-party protocol, $\exec$ be an execution trace
  (not necessarily valid) of $\widetilde{\Pi}$, $\sid$ be a session
  identifier and $\tau = \tagt(\exec, \sid)$.
  \begin{itemize}
    
  \item $\overline{n}^{\exec, \sid} = n^{\epsilon, S}_n$ if $n\in\mathcal{N}^\epsilon$ or if $\tau = \bot$, where $S = \sameTagas(\exec, \sid)$;

  \item $\overline{n_y^{\sid'}}^{\exec, \sid} = \left\{
      \begin{array}{l}
        n_y^{sid'} \mbox{ if }\sid' \in \sameTagas(\exec,\sid)\\
        \\
        n_{n_y^{sid'}}^{\epsilon,S} \mbox{ where } S =
        \sameTagas(\exec,\sid) \mbox{ otherwise ; }
      \end{array} 
      \right.$

  \item $\overline{a}^{\exec, \sid} = a$ if $a$ is the name of an agent;

  \item $\overline{\f(a_1, \dots, a_n)}^{\exec, \sid} = \f(a_1, \dots, a_n)$ for $\f \in \{\shk, \pub, \priv\}$

  \item $\overline{\langle u,v\rangle}^{\exec, \sid} = \langle \overline{u}^{\exec, \sid}, \overline{v}^{\exec, \sid} \rangle$;

  \item $\overline{\f(u_1, \dots, u_n)}^{\exec, \sid} = \left\{
      \begin{array}[c]{l}
        \f(\overline{u_1}^{\exec, \sid}, \dots, \overline{u_n}^{\exec, \sid}) \\
        \hfill \text{if } \tagh(\exec, \f(u_1,\dots, u_n)) = \tau \text{ and } \tau \neq \bot \\
        \\
        n^{\epsilon, S}_{\f(u_1, \dots, u_n)} \text{ where } S = \sameTagas(\exec, \sid) \text{ otherwise}
      \end{array}
    \right.$ 

    for any $\f \in \{\enc, \enca, \sign, \h\}$.
  \end{itemize}
\end{defi}

\noindent We extend our transformation on a trace in the expected way.

\begin{defi}[$\overline{\exec}$]\label{def:abstracted.trace} Let
  $\Pi$ be a protocol and $\exec = [\evt_1^{\sid_1}; \ldots;
  \evt_\ell^{\sid_\ell}]$ an execution trace (not necessarily valid)
  of $\widetilde{\Pi}$. We define $\overline{\exec} =
  \overline{\evt_1}^{\exec, \sid_1}; \ldots;
  \overline{\evt_\ell}^{\exec, \sid_\ell}$, where
%  $\overline{\evt}^{\exec, \sid}$ is defined as follows
  \[
  \overline{\evt}^{\exec, \sid} = \left\{
    \begin{array}[c]{ll}
      \Prd(\overline{u_1}^{\exec, \sid}, \dots, \overline{u_n}^{\exec, \sid}) &\;\; \text{if } \evt = \Prd(u_1, \dots, u_n) \\
      \snd(\overline{u}^{\exec, \sid}) &\;\; \text{if } \evt = \snd(u) \\
      \rcv(\overline{u}^{\exec, \sid}) & \;\;\text{if } \evt = \rcv(u)
    \end{array}
  \right.
  \]
\end{defi}

With this transformation, we still get a trace associated to the
  protocol under study. Moreover, the resulting execution trace is 
  well-formed. This is formally proved in Appendix~\ref{app:step-1} 
(Lemma~\ref{lem:concrete.trace} and Lemma~\ref{lem:trace-well-formed}).

\begin{prop}
\label{pro:both}
Let $\Pi$ be a $k$-party protocol, and $\exec$ be an execution trace
associated to~$\widetilde{\Pi}$ (not necessarily a valid one).  We
have that $\overline{\exec}$ is a well-formed execution trace (not
necessarily a valid one) associated to the protocol $\widetilde{\Pi}$.
\end{prop}

\subsection{Validity}
\label{subsec:validity}

Now, we show that the resulting execution trace, i.e. the one
obtained by applying our transformation $\overline{\; \cdot \;}$, is
still a valid one. In particular, we have to show that each term that
occurs in a receive event is deducible from the initial knowledge of
the attacker and the messages that have been sent so far.  For this,
we rely on the notion of simple proofs previously introduced
in~\cite{CD-fmsd09}.

\begin{defi}[simple proof]\label{def:simpleproof}
  Let $T_1 \subseteq T_2 \subseteq \dots \subseteq T_n$. 
 We say that a proof $\pi$ of $T_i \vdash u$ is {left-minimal} if,
 whenever there is a proof of $T_j \vdash u$ for some $j < i$, then
 $\pi$ is also a proof of $T_j \vdash u$. %In other words, the left-minimal proofs are those that can be performed in a minimal $T_j$.
Then, we say that a proof~$\pi$ is simple if
  \begin{enumerate}
  \item \label{def:simpleproof:subproof} any subproof of $\pi$ is {left-minimal},
  \item \label{def:simpleproof:composition} a composition rule of the form $\prooftree u_1 \quad u_2 \justifies u \endprooftree$ is never followed by a decomposition rule leading to $u_1$ or $u_2$, and
  \item \label{def:simpleproof:decomposition} any term of the form
    $\pair{u_1}{u_2}$ obtained by application of a decomposition rule 
or {labelling a leaf} is directly followed by a projection rule.
  \end{enumerate}
\end{defi}

\begin{exa}
  Let $T_1 = \{n_1\}$ and $T_2 = \{n_1, \, \enc(\langle n_1, n_2
  \rangle, k), \, k\}$. We have ${T_2 \vdash \langle n_1,n_2
    \rangle}$ with the proof tree $\pi$ described below.  However, $\pi$ is not a
  simple proof of $T_2 \vdash \langle n_1, n_2\rangle$. Indeed, the
  term $\langle n_1, n_2\rangle$ has been obtained by an application
  of a decomposition rule. Thus, by
  Condition~(\ref{def:simpleproof:decomposition}) of
  Definition~\ref{def:simpleproof} we have to decompose it. A simple
  proof of $T_2 \vdash \langle n_1, n_2\rangle$ is the proof tree
  $\pi'$ described below.
 \[
\begin{array}{ccc}
\pi = \left\{ \;
  \begin{prooftree}
    \enc(\langle n_1, n_2 \rangle, k) \;\;\; k
    \justifies
    \langle n_1, n_2\rangle
  \end{prooftree}
\right.
  &
\;\;\;\;\;\;\;\;\;
  &
\pi' = \left\{ \;
  \begin{prooftree}
    n_1 \;\;\;\;
    \begin{prooftree}
      \begin{prooftree}
        \enc(\langle n_1, n_2 \rangle, k) \;\;\; k
        \justifies
        \langle n_1, n_2\rangle
      \end{prooftree}
      \justifies
      n_2
    \end{prooftree}
    \justifies
    \langle n_1, n_2 \rangle 
  \end{prooftree}
\right.
\end{array}
  \]
  
\end{exa}

% pas de symbol pub et une notion de proof tree un peu differente
As it was done in~\cite{CD-fmsd09} in a slightly different setting, we
can show that it is always possible to consider such a proof tree,
i.e. if there is a proof of $T_i \vdash u$, then there is a simple
proof of it (w.r.t. a sequence $T_1 \subseteq T_2 \subseteq \dots
\subseteq T_n$).
%
%\begin{lem}\label{lem:simple.proofs}
 % Let $T_1 \subseteq T_2 \subseteq \dots \subseteq T_n$. If there is a proof of $T_i \vdash u$, then there is a simple proof of it.
%\end{lem}
%
Given a simple proof $\pi$ of $T_i \vdash u$, we can also show a locality
lemma (by structural induction on $\pi$) allowing us to
characterize the terms that occur in such a proof tree.

\begin{lem}[locality]
  \label{lem:locality}
  Let $T_1 \subseteq T_2 \subseteq \dots \subseteq T_n$ be a set of terms and $u$ be a term such that $T_i \vdash u$. Let $\pi$ be a simple proof of $T_i \vdash u$. We have that $\pi$ only involves terms in $\st(T_i \cup \{u\}) \cup \Kcal_\epsilon \cup \Nces_\epsilon \cup \mathcal{A} \cup \{\pub(a) ~|~ a \in \mathcal{A}\}$. Moreover, if $\pi$ ends with an instance of a decomposition rule (or is reduced to a leaf), we have that $\pi$ only involves terms in $\st(T_i) \cup \Kcal_\epsilon \cup \Nces_\epsilon \cup \mathcal{A} \cup \{\pub(a) ~|~ a \in \mathcal{A}\}$.
\end{lem}

%\begin{proof}
 % The proof can easily be done by structural induction on $\pi$.
%\end{proof}

Now, relying on this notion of simple proof, we can show that
deducibility is preserved by our transformation. This is the key lemma
to ensure the validity of the resulting trace.

\newsavebox{\lemdeducibilityofabstractedterms}
\sbox{\lemdeducibilityofabstractedterms}{\vbox{%
\begin{lem}\label{lem:deducibility.of.abstracted.terms}
  Let $\Pi$ be a $k$-party protocol and $\exec = [\evt_1^{\sid_1};
  \dots; \evt_\ell^{\sid_\ell}]$ be a valid execution trace of
  $\widetilde{\Pi}$, w.r.t.
some set $T_0$  of ground atoms. Let $i \in \{0, \ldots, \ell\}$ and~$t$ be a term such that  $\K(\exec_i) \cup T_0 \vdash t$. We have that  $\K(\overline{\exec}_i) \cup T_0 \vdash \overline{t}^{\exec, \sid}$ for any $\sid$.
\end{lem}
}}

\noindent\usebox{\lemdeducibilityofabstractedterms}

\begin{proof}(sketch)
Let $\tr = [\evt\evt_1^{\sid_1}; \dots; \evt\evt_\ell^{\sid_\ell}]$ be the symbolic trace associated to $\exec$ and~$\sigma$ be the substitution such that $\dom(\sigma) = \vars(\tr)$ and $\exec = \tr\sigma$.
Let $i \in \{0, \ldots, \ell\}$.  Let $\pi$ be a simple proof of
$\K(\exec_i) \cup T_0 \vdash t$. We prove that $\K(\overline{\exec}_i)
\cup T_0 \vdash \overline{t}^{\exec, \sid}$ by induction on $(i,
\pi)$. If $i = 0$ and $\pi$ is a simple proof reduced to a {leaf}
(possibly followed by some projection rules), then we have that $T_0
\vdash t$, and $\pi$ is necessarily reduced to a {leaf} since $T_0$
only contains atomic terms. Let $\sid$ be a session identifier, we
have that $\overline{t}^{\exec, \sid} \in \{t\} \cup
\mathcal{N}_\epsilon$ {since $t$ is an atomic term}. This allows us to conclude that $T_0 \vdash
\overline{t}^{\exec, \sid}$. Now, we distinguish two cases depending
on the last rule of~$\pi$.  \medskip{}

\begin{itemize}
\item   \noindent\emph{The proof $\pi$ ends with an instance of a composition rule, i.e. $t= \f(t_1, \dots, t_n)$ for some $\f\in\{\pair{}{}, \enc, \enca, \sign, \h\}$ and some terms $t_1,\ \dots,\ t_n$.}

  According to Definition~\ref{def:abstract.term}, we have that
  $\overline{t}^{\exec, \sid} \in \mathcal{N}^\epsilon \cup
  \{\f(\overline{t_1}^{\exec, \sid}, \dots, \overline{t_n}^{\exec,
    \sid})\}$. If $\overline{t}^{\exec, \sid} \in
  \mathcal{N}_\epsilon$, we easily conclude that
  $\K(\overline{\exec}_i) \cup T_0 \vdash \overline{t}^{\exec,
    \sid}$. Otherwise, since $\pi$ ends with a composition rule, we
  have that $\K(\exec_i) \cup T_0 \vdash t_1$, \dots, $\K(\exec_i)
  \cup T_0 \vdash t_n$. Moreover, the simple proofs witnessing these
  facts are strict subproofs of $\pi$ that are also simple. 
Hence, we can apply our induction hypothesis 
%in order to conclude that $\K(\overline{\exec}_i) \cup T_0 \vdash
%\overline{t_1}^{\exec, \sid}$, \dots, $\K(\overline{\exec}_i) \cup
%T_0 \vdash \overline{t_n}^{\exec, \sid}$. This allows us to 
and conclude that $\K(\overline{\exec}_i) \cup T_0 \vdash  \f(\overline{t_1}^{\exec, \sid}, \dots, \overline{t_n}^{\exec, \sid})$.
  
  % Decomposition
  \medskip{}
\item   \noindent\emph{The proof ends with the application of a
    decomposition rule (but not a projection) possibly followed by
    several applications of the projection rules until the resulting
    term is not a pair.} 

  We will here present the case of the symmetric decryption rule, but
  all the other decomposition rules (including the case {of a proof
  reduced to a leaf}) can be
  handled in a similar way. For some terms $t_1$ and $t_2$, the
  proof~$\pi$ is of the form
  \[
  \prooftree
  \prooftree
  \prooftree
  \prooftree
  \vdots
  \justifies
  \K(\exec_i) \cup T_0 \vdash \enc(t_1, t_2)
  \endprooftree
  \quad
  \prooftree
  \vdots
  \justifies
  \K(\exec_i) \cup T_0 \vdash t_2
  \endprooftree
  \justifies
  \K(\exec_i) \cup T_0 \vdash t_1
  \endprooftree
  \justifies
  \vdots
  \endprooftree
  \justifies
  \K(\exec_i) \cup T_0 \vdash t
  \endprooftree
  \]

  Let us first note that, by locality (Lemma~\ref{lem:locality}) of
  $\pi$ we know that $\enc(t_1, t_2) \in \st(\K(\exec_i)) \cup T_0
  \cup \Kcal_\epsilon \cup \Nces_\epsilon \cup \{\pub(a) ~|~ a \in
  \mathcal{A}\}$, and by atomicity of $T_0$, $\Nces_\epsilon$,
  $\Kcal_\epsilon$ and  $\{\pub(a) ~|~ a \in \mathcal{A}\}$, we know
  that $\enc(t_1, t_2) \in \st(\K({\exec_i}))$. (In case {of a proof
  reduced to a leaf}, and if there is no projection rule,  we may have that $t \in T_0$. In such a case, as in the base case, we have that $T_0 \vdash  \overline{t}^{\exec, \sid}$ and we easily conclude.)
  Hence,   there exists $k \le i$ such that  $\evt_k^{\sid_k} = \snd(u)$ and $\enc(t_1, t_2) \in \st(u)$. Let $k_0$ be the smallest such $k$ and $u_0$, $u_0'$ be such that $\evt^{\sid_{k_0}}_{k_0} = \snd(u_0)$ and $\evt\evt^{\sid_{k_0}}_{k_0} = \snd(u_0')$. Hence, we have that $u_0 = u_0'\sigma$.\\
  In order to prove the result, we first establish the
  following claim (proved in Appendix~\ref{sec:appvalidity}).

  \noindent{\bf Claim:} We have that $\overline{\enc(t_1, t_2)}^{\exec, \sid_{k_0}} = \enc(\overline{t_1}^{\exec, \sid_{k_0}}, \overline{t_2}^{\exec, \sid_{k_0}})$.

\noindent Now, relying on this claim and applying the induction
hypothesis, we have that:
  \begin{itemize}
  \item $\K(\overline{\exec}_i) \cup T_0 \vdash \enc(\overline{t_1}^{\exec, \sid_{k_0}}, \overline{t_2}^{\exec, \sid_{k_0}})$; and 

  \item $\K(\overline{\exec}_i) \cup T_0 \vdash \overline{t_2}^{\exec, \sid_{k_0}}$.
  \end{itemize}
  This allows us to deduce that $\K(\overline{\exec}_i) \cup T_0 \vdash \overline{t_1}^{\exec, \sid_{k_0}}$. 
%  \medskip{}
  In order to establish that $\K(\overline{\exec}_i) \cup T_0 \vdash
  \overline{t}^{\exec,\sid}$, we need to distinguish two
 cases: 

  \medskip{}
  \noindent \emph{Case 1. $t \in \A$, $t = \pub(a)$ or $t = \f(a_1,
    \dots, a_n)$ for some $\f\in\{\shk, \priv\}$.}  { In such a case,
    we have that $\overline{t}^{\exec, \sid} = \overline{t}^{\exec,
      \sid_{k_0}} = t$. Hence, we have that $\K(\overline{\exec}_i)
    \cup T_0 \vdash \overline{t}^{\exec, \sid}$ by applying some
    projection rules on the proof of $\K(\overline{\exec}_i) \cup T_0
    \vdash \overline{t_1}^{\exec, \sid_{k_0}}$.  }
%   \item[Case $t = \f(a_1, \dots, a_n)$ for some $\f\in\{\shk,
%     \priv\}$] According to Lemma~\ref{lem:plaintext}, we have that
%     $t\in\plaintext(\tr) \cup T_0\cup \mathcal{K}_\epsilon$. 

% \marginpar{Myrto ?}
% \stef{euh, on n'a pas cette hypothese dans le lemme. ceci dit la
%   transformation ne modifiera pas le terme et on obtient que le terme
%   souhaitee est deductible en utilisant la preuve de $\overline{t_1}$
%     et en appliquant des projections.}
% \stefcolor{By hypothesis, we have no long-term keys appearing in plaintext in $\widetilde{\Pi}$}, and thus by construction of $\tr$ (see Definition~\ref{def:scenario}) we have no long-term keys appearing in plaintext in $\tr$. Hence, $t\in\mathcal{K}_\epsilon \cup T_0$. Moreover, by Definition we have $\overline{t}^{\exec, \sid} = t \in \mathcal{K}_\epsilon \cup T_0$. Hence, we have that $\K(\overline{\exec}_i) \cup T_0 \vdash \overline{t}^{\exec,\sid}$. 

\smallskip{}

%\stef{This is the case that is difficult to sum up.}      
\noindent \emph{Case 2. $t \in \mathcal{N}$ or $t = \f(t'_1, \dots, t'_m)$ for some
  $\f\in\{\enc, \enca, \h, \sign\}$.}
{First, if $\overline{t}^{\exec,\sid}$ can be obtained by application
  of some projection rules on the proof of $\K(\overline{\exec}_i)
  \cup T_0 \vdash \overline{t_1}^{\exec, \sid_{k_0}}$, then we easily
  conclude. Otherwise, it means that the term $t$ is not abstracted in
  the same way in both cases. In such a case, we have that either
  $\overline{t}^{\exec,\sid} \in \Nces_\epsilon$ or
  $\overline{t}^{\exec,\sid_{k_0}} \in \Nces_\epsilon$. In the first
  case, we easily conclude. In the second case, i.e.
  $\overline{t}^{\exec,\sid_{k_0}} \in \Nces_\epsilon$ but
  $\overline{t}^{\exec,\sid} \not\in \Nces_\epsilon$, we can show that
  $t$ is a subterm of $u_0$ that either occurs as a component of $u_0$
  or in the term~$x\sigma$ for some $x \in \vars(u'_0)$. Actually, the
  first case is not possible since we have assumed that
  $\overline{t}^{\exec,\sid_{k_0}} \in \Nces_\epsilon$. Thus, only the
  second case remains. Thanks to the origination property, we know
  that~$t$ will occur in a previous receive event and we will be able
  to show that $t$ was deducible using a smaller prefix of the trace
  allowing us to conclude by applying our induction hypothesis.  }\qedhere
\end{itemize}

\end{proof}

\noindent Since our transformation preserves the  deducibility relation, we can
now prove the validity of the resulting trace by induction on the
length of the original trace.

\begin{prop}\label{prop:bar.valid}
  Let $\Pi$ be a $k$-party protocol and $\exec = [\evt^{\sid_1}_1;
  \ldots; \evt^{\sid_\ell}_\ell]$ be a valid execution trace
  associated to $\widetilde{\Pi}$, w.r.t. some initial intruder
  knowledge $T_0$. We have that $\overline{\exec}$ is a  {well-formed}
  and valid execution trace associated to $\widetilde{\Pi}$ w.r.t $T_0$.
\end{prop}

\begin{proof}
  First, according to Proposition~\ref{pro:both}, we know that
  $\overline{\exec}$ is an execution trace associated
  to~$\widetilde{\Pi}$ which is well-formed.
 It remains to establish its validity w.r.t.~$T_0$.  We show by
 induction on $i$ that for all $i\in\{1, \ldots, \ell\}$, $(\overline
 \exec)_i$ is a valid execution trace.  The base case, i.e. the empty
 trace $(\overline \exec)_i = []$, is trivially valid.  For the
 inductive step, we assume that $(\overline{\exec})_{\ell-1}$ is valid
 and we have to establish the validity of $\overline{\exec} =
 \overline{\exec}_\ell$.  
We distinguish~{2 cases} according to the nature of the last event in the trace.
  \medskip{}

  \noindent \emph{Case $\evt^{\sid_\ell}_\ell = \Prd(t_1, \dots, t_n)$
    {or  $\evt^{\sid_\ell}_\ell = \snd(t)$}.} By induction hypothesis, 
we know that $(\overline{\exec})_{\ell-1}$ is a valid execution trace, and  this is enough to conclude to the validity of $\overline{\exec}$.
  \medskip{}

%  \noindent\emph{Case $\evt^{\sid_\ell}_\ell = \snd(t)$.} By induction hypothesis, we know that  $(\overline{\exec})_{\ell-1}$ is a valid execution trace, and %this is enough to conclude to the validity of $\overline{\exec}$.
 % \medskip{}

  \noindent\emph{Case $\evt^{\sid_\ell}_\ell = \rcv(t)$.} By induction hypothesis, we know that $(\overline{\exec})_{\ell-1}$ is a valid execution trace. To conclude to the validity of~$\overline{\exec}$, we only need to establish that $\K((\overline{\exec})_{\ell-1}) \cup T_0 \vdash \overline{t}^{\exec, \sid_\ell}$. Since we know that  $\exec$ is a valid execution trace, we have that    $\K(\exec_{\ell-1}) \cup T_0 \vdash t$. Applying Lemma~\ref{lem:deducibility.of.abstracted.terms}, we conclude that $\K((\overline{\exec})_{\ell-1}) \cup T_0 \vdash \overline{t}^{\exec,\sid_\ell}$. This allows us to deduce that $\overline{\exec}$ is valid.
\end{proof}

\subsection{Satisfiability}
\label{subsec:sat}

The goal of this section is to show that the trace $\overline{\exec}$
resulting of the application of our transformation will still satisfy
the attack formula $\exists x_1. \ldots. \exists x_n. \phi$ under
study. To show the validity of such a formula on the trace
$\overline{\exec}$, we have to exhibit a substitution $\sigma'$ for
which $\langle \overline{\exec}, T_0 \rangle \models \phi\sigma'$. By
hypothesis, we know that $\langle \exec, T_0 \rangle \models
\phi\sigma$ for some $\sigma$.  Thus, the idea is to consider the
substitution $\sigma' = \{x_1 \mapsto
\overline{x_1{\sigma}}^{\exec,\sid_1}, \ldots, x_n \mapsto
\overline{x_n{\sigma}}^{\exec,\sid_n}\}$ where $\sid_1,
\ldots, \sid_n$ correspond to the sessions from which the terms
$x_1\sigma, \ldots, x_n\sigma$ come from.

\newsavebox{\probarattack}
\sbox{\probarattack}{\vbox{%
\begin{prop}
  \label{prop:bar.attack}
  Let $\Pi$ be a protocol, $\exec$ be an execution trace 
  % (not necessarily
  % valid) 
  of~$\widetilde{\Pi}$ w.r.t. some initial intruder
  knowledge $T_0$, and $\phi$ be an attack formula. We have that
  \begin{center}
    $\langle
    \exec, T_0\rangle \models \phi\ \Rightarrow\ \langle
    \overline{\exec}, T_0\rangle \models \phi$.
  \end{center}
\end{prop}
}}

\noindent\usebox{\probarattack}

\noindent The proof is done by structural induction on the formula and its details can be found in Appendix~\ref{app:sat}. The
technically difficult part is to formally link each variable
existentially quantified in $\phi$ with the term it has been
substituted with in order to satisfy the formula.

\section{Second step: reducing the number of sessions}
\label{sec:2-step}

Now, our goal is to reduce the number of sessions that are involved in
an execution trace witnessing the existence of an attack in order to
match the bound announced in Theorem~\ref{theo:transfo1}: the attack
trace has to involved at most $\|\phi\|$ sessions of each role. The
idea will be to identify a set of sessions $S$ and to remove all the
events that do not originate from a session in $S$ according to the
formal definition stated below.

\begin{defi}[restriction of $\tr$ to $S$] 
  \label{def:exec.restriction}
  Let $\Pi$ be a protocol, $\exec = [\evt^{\sid_1}_1; \dots;
  \evt^{\sid_\ell}_\ell]$ be an execution of $\Pi$, w.r.t. some set $T_0$
  of ground atoms, and $S$ be a set of session identifiers. The restriction of
  $\exec$ to $S$ is defined as the trace $\exec|_S = [\evt^{\sid_{i_1}}_{i_1};
  \dots; \evt^{\sid_{i_h}}_{i_h}]$ satisfying the following: $i_1 < \ldots <
  i_h$ and for all $j\in \{1, \ldots, \ell\}$, there exists $k\in \{1, \ldots,
  h\}$ such that $j = i_k$ if and only if $\sid_j\in S$.
\end{defi}

Given a valid and well-formed execution $\exec$ and a set of
sessions~$S$, the goal of this section is to show that the restriction
$\exec|_S$ is a valid and well-formed execution. Since messages coming
from one session can be used to build a message for another session,
to prove such a result, it is important to require some conditions
on~$S$. Basically, we will consider a set $S$ that satisfies the
following requirement:

\begin{center}
  for all $\sid_1$ and $\sid_2$ such that 
$\sameTagas(\exec,
  \sid_1) = \sameTagas(\exec, \sid_2)$, \\
we have that $\sid_1\in S$ if and only if $\sid_2\in
  S$.
\end{center}

\noindent This means that sessions using the same tag should have the
same status w.r.t. the set~$S$.

In the following of this section we will first show that
\begin{enumerate}[label=(\roman*)]
\item such a restricted execution is still a \emph{valid} execution, and
\item that the restriction preserves \emph{satisfiability} of attack formulas.
\end{enumerate}

%%%
%%%
%%% Validity of the restriction
%%%
%%%

\subsection{Validity of the restriction}
\label{subsec:validity-restriction}

First, we show that in a well-formed and valid execution trace, terms
that occur in sessions that are tagged differently do not share any
name.
 
% S.D.: OK relu

\newsavebox{\lemdisjointness}
\sbox{\lemdisjointness}{\vbox{%
\begin{lem}
\label{lem:disjointness}
  Let $\Pi$ be a $k$-party protocol, and $\exec = [\evt^{\sid_1}_1; \dots;
  \evt^{\sid_\ell}_\ell]$ be a well-formed valid execution of $\widetilde \Pi$
  w.r.t.  some set~$T_0$ of ground atoms. Let $\sess_1$ and
  $\sess_2$ be two session identifiers. We have that:
\[
\begin{array}{c}
\sameTagas(\exec, \sess_1) \not=
  \sameTagas(\exec, \sess_2) \\
\mbox{ implies } \\
\names(\exec, \sess_1) \cap \names(\exec,
  \sess_2) = \emptyset
\end{array}
\]
{where $\names(\exec,\sess) = \{u ~|~ u \in
\names(\evt^{\sid_j}_j)\mbox{ for some $1 \leq j \leq \ell$ such that $\sid_j = \sess$}\}$.}
\end{lem}
}}

\noindent\usebox{\lemdisjointness}

The goal of the next lemma is to show that deducibility is preserved when we
consider the trace $\exec|_S$. Note that the previous lemma allows us to ensure
that the terms we removed from the trace are ``sufficiently disjoint'' from the
ones we keep. This is important to ensure that deducibility is preserved in the
trace $\exec|_S$.

% S.D.: OK relu

\newsavebox{\lemprojecteddeducibility}
\sbox{\lemprojecteddeducibility}{\vbox{%
\begin{lem}
\label{lem:projected.deducibility}
Let $\Pi$ be a $k$-party protocol, and $\exec = [\evt^{\sid_1}_1;
\dots; \evt^{\sid_\ell}_\ell]$ a well-formed valid execution of
$\widetilde \Pi$ w.r.t. some set~$T_0$ of ground atoms, and such
that {$T_0 \cup \K(\exec) \not\vdash k$ for any $k \in \lgAtom
  \smallsetminus (\Kcal_\epsilon \cup T_0)$} ($\exec$ does not reveal
any long term keys).  Let $S$ be a set of sessions such that:
\begin{quote}
  for all session identifiers $sess_1$ and $sess_2$ such that
  $\sameTagas(\exec, sess_1) = \sameTagas(\exec, sess_2)$, we have
  that $sess_1\in S$ if and only if $sess_2\in S$.
\end{quote}
For all term $t\in\st(\exec|_S)$ such that $T_0\cup \K(\exec)\vdash t$, we have that $T_0 \cup
\K(\exec|_S)\vdash t$.
\end{lem}
}}

\noindent\usebox{\lemprojecteddeducibility}

Now, relying on Lemma~\ref{lem:projected.deducibility}, we are able to show that the trace $\exec|_S$ is valid.

%% SD: ok relu

\newsavebox{\provalidityrestriction}
\sbox{\provalidityrestriction}{\vbox{%
\begin{prop}
  \label{pro:validity-restriction}
  Let $\Pi$ be a $k$-party protocol, and $\exec = [\evt^{\sid_1}_1;
  \dots; \evt^{\sid_\ell}_\ell]$ a well-formed valid execution of
  $\widetilde \Pi$ w.r.t. some set~$T_0$ of ground atoms, and
  such that {$T_0 \cup \K(\exec) \not\vdash k$ for any $k \in \lgAtom
    \smallsetminus (\Kcal_\epsilon \cup T_0)$} ($\exec$ does not
  reveal any long term keys). Let $S$ be a set of sessions such that:
  \begin{quote}
    for all session identifiers $\sess_1$ and $\sess_2$ such that $\sameTagas(\exec, \sess_1) = \sameTagas(\exec, \sess_2)$, we have that $\sess_1\in S$ if and only if $\sess_2\in S$.
  \end{quote}
  We have that  $\exec|_S$ is also a well-formed and valid execution of $\widetilde \Pi$ w.r.t.  $T_0.$
\end{prop}
}}

\noindent\usebox{\provalidityrestriction}

%% SD: ok relu

\begin{proof}
  Let $1 \le i_1 < \dots < i_n \le \ell$ such that $\exec|_S =
  [\evt^{\sid_{i_1}}_{i_1}; \dots; \evt^{\sid_{i_n}}_{i_n}]$. We prove
  by induction on the length $n$ of $\exec|_S$, that $\exec|_S$ is a
  valid execution of $\widetilde \Pi$, w.r.t $T_0$.

  \begin{description}
  \item[Base case] If $n=0$ we have that $\exec|_S =[]$, and thus
    $\exec|_S$ is a valid execution of $\widetilde \Pi$ w.r.t.
    $T_0$.

  \item[Inductive case] By induction hypothesis, we know that
    \[[\evt^{\sid_{i_1}}_{i_1}; \dots;
    \evt^{\sid_{i_{n-1}}}_{i_{n-1}}]\] is a valid execution of
    $\widetilde \Pi$ w.r.t. $T_0$. If $\evt^{\sid_n}_n$ is a
    send or a status event, then \[\exec|_S = [\evt^{\sid_{i_1}}_{i_1};
    \dots; \evt^{\sid_{i_{n-1}}}_{i_{n-1}}; \evt^{\sid_{i_n}}_{i_n}]\]
    is a valid execution of $\widetilde \Pi$ w.r.t.  $T_0$ (see
    Definition~\ref{def:validexec}). On the other hand, if
    $\evt^{\sid_n}_n$ is a receive event, i.e.
    $\evt^{\sid_n}_n = \rcv(t)$, we need to show \[T_0 \cup
    \K([\evt^{\sid_{i_1}}_{i_1}; \dots;
    \evt^{\sid_{i_{n-1}}}_{i_{n-1}}]) \vdash t\] knowing that \[T_0 \cup
    \K([\evt^{\sid_{1}}_{1}; \dots; \evt^{\sid_{(i_n)-1}}_{(i_n)-1}])
    \vdash t\] which because $\evt_{i_n}$ is a reception event implies
    that $T_0 \cup \K(\exec) \vdash t$.  But then, according to
    Lemma~\ref{lem:projected.deducibility} we know that $T_0 \cup
    \K(\exec|_S) \vdash t$.  It suffices now to notice that by
    definition of $\K()$, because~$\evt_{i_n}$ is a reception event,
    we have that:
    \[\K([\evt^{\sid_{i_1}}_{i_1}; \dots;
    \evt^{\sid_{i_{n-1}}}_{i_{n-1}}]) = \K([\evt^{\sid_{i_1}}_{i_1};
    \dots; \evt^{\sid_{i_n}}_{i_n}]) = \K(\exec|_S).
    \] 
  \end{description}
  This concludes the proof that $\exec|_S$ is a valid execution of
  $\widetilde \Pi$ w.r.t. the initial intruder knowledge
  $T_0$. Finally, it is obvious that $\exec|_S$ satisfies the 3
  conditions of well-formedness (Definition~\ref{def:wfexec}), from
  the hypothesis that $\exec$ does.\qedhere
\end{proof}

\subsection{Satisfiability of the formula}

The way the set $S$ of sessions is chosen depends on the sessions that
are needed to satisfy the attack formula under study. We therefore
introduce the notion of \emph{witness sessions} which for a given
formula $\phi$ can be used to witness that $\phi$ holds.

\begin{defi}[witness sessions, $\Witness$]
\label{def:witness}
  Let $\Pi$ be a protocol, $\phi$ a closed quantifier-free formula of $\Logic$, and $T_0$ be a set of ground atoms.
  Let $\exec = [\evt_1^{\sid_1}; \dots; \evt_\ell^{\sid_\ell}]$ be a
  valid execution of $\Pi$ (w.r.t.~$T_0$) 
satisfying $\phi$, i.e. $\langle \exec, T_0\rangle \models
\phi$. We define the 
set of sessions $\Witness(\exec, \phi)$ witnessing that $\langle
\exec, T_0\rangle \models \phi$ by structural induction on $\phi$ as follows:
\begin{itemize}
\item   $\W(\exec, \neg\phi)  =  \W^-(\exec, \phi)$;
\item $\W(\exec,\true) = \W(\exec,\learn(t)) 
= \W(\exec,\Comp(t)) = \emptyset$;
\item $\W(\Q(t_1,\ldots,t_n)) = \{\sid_\ell\}$;
\item $\W(\exec,\Diamond \phi) = \W(\exec_i,\phi)$ where $i$ is such
  that $\langle \exec_i,T_0\rangle \models \phi$;
\item   $\W(\exec, \phi_1 \vee \phi_2)= \W(\exec,\phi_1)$ if $\langle
  \exec,T_0 \rangle \models \phi_1$ and $\W(\exec,\phi_2)$ otherwise;
\end{itemize}
where
\begin{itemize}
\item $\W^-(\exec, \neg\phi) = \W(\exec, \phi)$;
\item $\W^-(\exec,\true) = \W^-(\exec,\learn(t)) =
  \W^-(\exec,\Comp(t)) = \W^-(\exec,\Diamond\phi) = \emptyset$;
\item $\W^-(\Q(t_1,\ldots,t_n)) = \{\sid_\ell\}$  when $\length(\exec)
  > 0$ and $\emptyset$ otherwise;
\item $\W^-(\exec,\phi_1 \vee \phi_2) =  \W^-(\exec, \phi_1) \cup \W^-(\exec, \phi_2)$. 
\end{itemize}
\end{defi}

Intuitively, we keep in the trace the sessions that are needed to satisfy the
formula under study. Essentially, we have to keep those that are used to satisfy
the status events occurring in the formula.

\newsavebox{\lemprojsatisfiabilitybis}
\sbox{\lemprojsatisfiabilitybis}{\vbox{%
\begin{lem}\label{lem:projsatisfiabilitybis}
  Let $\Pi$ be a $k$-party protocol, and $\exec = [\evt_1^{\sid_1}; \dots;
  \evt_\ell^{\sid_\ell}]$ be a valid and
  well-formed execution of $\widetilde \Pi$ w.r.t. some set $T_0$ of
  ground atoms such that $T_0 \cup \K(\exec)
  \not\vdash k$ for any $k \in \lgAtom \smallsetminus (\Kcal_\epsilon
  \cup T_0)$.
Let $\phi = \exists
  x_1. \ldots. \exists x_n. \psi$ be an attack formula
of   $\Logic$, and $\sigma$
be a ground substitution such that $\langle \exec, T_0
\rangle \models \psi\sigma$.
Let $S$ be a set of session identifiers such
  that:
  \begin{enumerate}
%  \item for all $\learn(t)$ that occurs positively in $\phi$ such that $t\not\in \A \cup
 %   \lgAtom$, there exists $\sid \in S$ such that $t\in \st(\exec, \sid)$,
%  \item 
\item $\Witness(\exec, \psi\sigma) \subseteq S$, and
 % \item 
\item $\forall \sess_1, \sess_2$ with
    $\tagt(\exec, \sess_1) = \tagt(\exec, \sess_2)$, we have that 
\begin{center}
$\sess_1\in S$ if and only if
    $\sess_2\in S$.
\end{center}
\end{enumerate}
  We have that $\exec|_S$ is an execution of $\widetilde \Pi$ that
  satisfies $\phi$, i.e.
  $\langle \exec|_S, T_0\rangle \models \phi$.
\end{lem}
}}

\noindent\usebox{\lemprojsatisfiabilitybis}

\begin{proof} (sketch) The idea is to show that $\langle \exec|_S, T_0 \rangle
  \models \psi\sigma$. However, this result is wrong in general since the
  substitution $\sigma$ witnessing the fact that the attack formula $\phi$ is
  satisfiable can use some terms that only occur in events coming from sessions
  that are not in $S$.  Thus, the first step of the proof consists in showing
  that we can consider a substitution $\sigma$ that only involves subterms that
  occur in $\st(\exec|_S)$.  For instance, consider the formula $\exists
  x. \learn(x)$. Since, the variable $x$ does not occur in any status event, we
  cannot ensure that $x$ will be bound to a term coming from a session in~$S$.
  However, intuitively, we can replace such a term $x\sigma$ by a nonce in
  $\Nces_\epsilon$ still preserving the satisfiability of the attack formula.
  Now, we can assume w.l.o.g. that that for all $j\in\{1, \dots, n\}$,
  $\sigma(x_j) \in\st(\exec, S)\cup\A\cup\lgAtom \cup \Nces_\epsilon$.  Then, we
  proceed by induction on the length of the execution trace and the size of the
  formula, and we show that $\langle \exec|_S, T_0 \rangle \models
  \psi\sigma$. In other words, the attack formula is satisfiable and $\sigma$ is
  a witness of this fact.
\end{proof}

\newsavebox{\proattacktildepibis}
\sbox{\proattacktildepibis}{\vbox{%
\begin{prop}
  \label{pro:attacktildepibis}
  Let~$\Pi$ be a $k$-party protocol and $T_0$ be a finite set of
  ground atoms such that 
$\lgAtom(\Pi) \cap \plaintext(\Pi) \subseteq T_0 \cup
\Kcal_\epsilon$. 
Let $\exec$ be a valid and well-formed execution 
of $\widetilde{\Pi}$ w.r.t. $T_0$, and $\phi = \exists x_1. \ldots.
\exists x_n. \psi$ be an
 attack formula such that $\pair{\exec}{T_0} \models \psi\sigma$ for
 some ground substitution $\sigma$. 
We have that $\pair{\exec|_S}{T_0} \models \phi$ where $S = \{\sid\mid \exists \sid'\in\Witness(\exec, \psi\sigma) \text{ and } \sid\in\sameTagas(\exec, \sid')\}$.
\end{prop}
}}

\noindent\usebox{\proattacktildepibis}

\begin{proof}
 Let $\exec$ be a valid and well-formed execution of $\widetilde\Pi$
 w.r.t. $T_0$ such that $\langle \exec, T_0\rangle \models
 \phi$.

\medskip{}

\noindent \emph{Claim:  $T_0 \cup \K(\exec)
  \not\vdash k$ for any $k \in \lgAtom \smallsetminus (\Kcal_\epsilon
  \cup T_0)$.} Assume that there exists $k \in
  \lgAtom$ such that $T_0 \cup \K(\exec)
  \vdash k$. Using Lemma~\ref{lem:plaintext.deducibility}, we
  obtain that $k\in \plaintext(\exec) \cup T_0 \cup \Kcal_\epsilon$,
  and relying on Lemma~\ref{lem:plaintext}, we conclude 
that $k\in \plaintext(\tr) \cup T_0 \cup \mathcal{K}_\epsilon$ where
$\tr$ is the symbolic trace underlying $\exec$. 
Now, by construction of $\tr$, if $k\in \plaintext(\tr)$, then there
exists $k'\in \plaintext(\Pi)$ such that $k = k'\sigma$ for some
$\sigma: \Xvars \rightarrow \A$. Hence, we have that $k' \in
\lgAtom(\Pi) \cup \plaintext(\Pi)$. Thanks to our hypothesis, we
conclude that $k'\in T_0\cup
\Kcal_\epsilon$, and thus $k'=k\in T_0\cup\Kcal_\epsilon$, which
concludes the proof of the claim.

\medskip{}

By hypothesis, we have that  $\langle \exec, T_0\rangle \models
\psi\sigma$ for some ground substitution $\sigma$. Moreover, by
 hypothesis, we have that:

  \begin{enumerate}
%  \item for all $\learn(t)$ that occurs positively in $\phi$ such that $t\not\in \A \cup
 %   \lgAtom$, there exists $\sid \in S$ such that $t\in \st(\exec, \sid)$,
%  \item 
\item $\Witness(\exec, \psi\sigma) \subseteq S$, and
 % \item 
\item $\forall \sess_1, \sess_2$ with
    $\tagt(\exec, \sess_1) = \tagt(\exec, \sess_2)$, we have that 
\begin{center}
$\sess_1\in S$ if and only if
    $\sess_2\in S$.
\end{center}
\end{enumerate}

Hence, we can apply Lemma~\ref{lem:projsatisfiabilitybis} to conclude
that $ \langle \exec|_S, T_0\rangle \models \phi$.
\end{proof}

\section{Main results}
\label{sec:3-step}

In this section, we put the pieces together and prove Theorem~\ref{theo:transfo1}, the main result that was stated in Section~\ref{subsec:result}. We also prove Corollary~\ref{cor:secrecy} which allows us to obtain slightly stronger results for particular security properties.

To prove our main result, we first need to bound the number of sessions that are needed to witness the satisfiability of the attack formula under study.  This is the purpose of the following lemma that can be proved by induction on the structure of $\phi$.

\begin{lem}
\label{lem:boundingWitness}
Let $\Pi$ be a protocol, $\phi$ a closed quantifier-free formula of $\Logic$, and $T_0$ be set of ground atoms.
  Let $\exec = [\evt_1^{\sid_1}; \dots; \evt_\ell^{\sid_\ell}]$ be a valid execution of $\Pi$ (w.r.t.~$T_0$) satisfying $\phi$, i.e. $\langle \exec, T_0\rangle \models \phi$. We have that
$|\Witness(\exec, \phi)|\ \le\ \|\phi\|$.
\end{lem}

%\begin{proof}
 % This lemma is obvious by inspection of the Definitions~\ref{def:sizePhi} and~\ref{def:witness}. Its proof can be donne by induction on the %structure of $\phi$.
%\end{proof}

\subsection{Proof of Theorem~\ref{theo:transfo1}}

Now, we prove our main theorem.

\noindent\usebox{\theotransfoun}

\begin{proof}
  Let $\exec$ be a valid execution of $\widetilde\Pi$ w.r.t. $T_0$ such that $\langle \exec, T_0\rangle \models \phi$.  By Proposition~\ref{prop:bar.valid} we have that $\overline{\exec}$ is a valid well-formed execution of~$\widetilde\Pi$ w.r.t.~$T_0$, and according to Proposition~\ref{prop:bar.attack}, we have that $\langle \overline{\exec}, T_0\rangle \models \phi$.  By definition on an attack formula, we have that $\phi = \exists x_1. \ldots. \exists x_n. \psi$ and we deduce that there exists $\sigma$ such that $\langle \overline{\exec}, T_0 \rangle \models \psi\sigma$.

  Let $S = \{\sid\mid \exists \sid'\in\Witness(\overline{\exec}, \psi\sigma) \text{ and } \sid\in\sameTagas(\overline{\exec}, \sid')\}$.  Now, by Proposition~\ref{pro:validity-restriction}, we have that $\overline{\exec}|_S$ is also a well-formed and valid execution of $\widetilde \Pi$ w.r.t. $T_0$; and according to Proposition~\ref{pro:attacktildepibis}, we know that $\pair{\overline{\exec}|_S}{T_0} \models \phi$.  Finally, Lemma~\ref{lem:boundingWitness} tells us that $|\Witness(\overline{\exec}, \psi\sigma)| \le \|\phi\sigma\| = \|\phi\|$.
%\myrto{Voulez-vous que je detaille l'argument qui suit plus??} 
  But because by construction of $\widetilde{\Pi}$ (and hence of all of its symbolic traces), in every execution of $\widetilde{\Pi}$ all sessions of the same role are tagged differently (each session introduces its own nonce making them different), $S$ must contain at most $\|\phi\|$ sessions of each role. This allows us to conclude that $\overline{\exec}|_S$ is an attack that involves at most $\|\phi\|$ sessions of each role.  \end{proof}

\subsection{Secrecy, aliveness and weak agreement.}
For several classical security properties we are actually able to obtain a slightly stronger result and only consider one \emph{honest} session of each role. As we will see below this is a direct corollary from the proof of the main theorem.

\noindent\usebox{\corsecrecy}

\begin{proof}
  We only detail the proof in the case of secrecy. The case of aliveness and weak agreement {are} treated similarly.  Let $\overline{\phi_\Sec} = \exists x_1. \ldots. \exists x_n. \exists y. \overline{\psi_\Sec}$.  Following the proof of Theorem~\ref{theo:transfo1}, we can show that $\langle \overline{\exec}, T_0 \rangle \models \overline{\psi_\Sec}\sigma$ for some substitution $\sigma$.

  Let $S = \{\sid\mid \exists \sid'\in\Witness(\overline{\exec}, \overline{\psi_\Sec}\sigma) \text{ and } \sid\in\sameTagas(\overline{\exec}, \sid')\}$.  We have that $\Witness(\overline{\exec},\overline{\psi_\Sec}\sigma) = 1$, and thus the set $S$ contains at most \emph{one} session of each role.  To conclude, we have to show that $S$ only contains \emph{honest} sessions.  By definition of $\Witness$, we know that $\Witness(\overline{\exec},\overline{\psi_\Sec}\sigma) = \{\sid_0\}$ for some $\sid_0$ such that the status event $\Secret(x_1\sigma,\ldots,x_k\sigma,y\sigma)$ is issued from the session $\sid_0$ and we have that $\langle \overline{\exec},T_0 \rangle \models \Nocomp(x_1\sigma) \wedge \ldots \wedge \Nocomp(x_k\sigma)$.  Hence, we have that $\sid_0$ is an honest session.

We have that  $S = \{\sid\mid \exists \sid'\in\Witness(\overline{\exec}, \overline{\psi_\Sec}\sigma) \text{ and } \sid\in\sameTagas(\overline{\exec}, \sid')\}$ which means that 
  $S = \{\sid\mid \sid\in\sameTagas(\overline{\exec}, \sid_0)\}$. Since the names of the agents that are involved in a session occur in the tag, we know that all the sessions in $S$ are honest. This allows us to conclude.
\end{proof}

%\newpage

\section{Conclusion}
\label{sec:conclu}

In this paper we present a transformation which guarantees that
attacks on transformed protocols only require a number of sessions
which is a function of the security property under
study. {We prove this result for a class of security
  properties that includes secrecy and several flavors of
  authentication. Our logic for specifying security properties does
  not allow one to express injective authentication properties
  (e.g. injective agreement, matching conversations, etc.) but we
  believe that both the logic and our reduction result could be
  extended to this setting.}

A challenging topic for future research is to
obtain more fine-grained characterizations of decidable classes of protocols for
an unbounded number of sessions. The new insights gained by our work seem to be
a good starting point to extract the conditions needed to reduce the security
for an unbounded number of sessions to a finite number of sessions.

\section*{Acknowledgments}
This work has been partially supported by the projects JCJC VIP
ANR-11-JS02-006
and ERC grant agreement no 258865, project ProSecure.

%\paragraph{Acknowledgments.} 
%We would like to thank Yassine Lakhnech for discussions that initiated
%this work as well as Hubert Comon-Lundh, Véronique Cortier, Joshua
%Guttman and Ralf K\"usters for their helpful comments.

\bibliographystyle{abbrv}
\bibliography{reference}

\appendix

\section{Proofs of Section~\ref{subsec:1-step-transfo}}
\label{app:step-1}

%%
%% Notre transformation preserve le fait que c'est une trace associe
%% au protocole (pas necessairement well-formed, ni valide)
%%

In this section, we show that our transformation maps an execution trace $\exec$
to a well-formed execution trace $\overline{\exec}$. The resulting execution
trace $\overline{\exec}$ is still a trace associated to the protocol
$\widetilde{\Pi}$ under study.

\newsavebox{\lemconcretetrace}
\sbox{\lemconcretetrace}{\vbox{%
\begin{lem}\label{lem:concrete.trace} 
  Let $\Pi$ be a $k$-party protocol, and $\exec$ be an execution trace
  associated to $\widetilde{\Pi}$ (not necessarily a valid one).  We have that
  $\overline{\exec}$ is an execution trace (not necessarily a valid one)
  associated to the protocol $\widetilde{\Pi}$
\end{lem}
}}

\noindent\usebox{\lemconcretetrace}

\begin{proof}(sketch) 
  Let $\tr = [\evt_1^{\sid_1}; \ldots; \evt_\ell^{\sid_\ell}]$ be the symbolic
  trace of $\widetilde \Pi$, and $\sigma$ be the ground substitution such that
  $\dom(\sigma) = \vars(\tr)$ and $\exec = \tr\sigma$. Let $\overline\sigma$ be
  such that:
  \begin{itemize} 
  \item $\dom(\overline \sigma) = \dom(\sigma)$, and
  \item %for each variable $x\in \dom(\overline \sigma)$, we have that
    $\overline\sigma(x) = \overline{x\sigma}^{\exec,\sid}$, where $x\in\vars(\tr,
    \sid)$.
  \end{itemize}
  Clearly, we have that $\overline\sigma$ is a ground substitution. It remains
  to establish that $\overline{\exec} = \tr\overline{\sigma}$ so that the
  execution $\exec$ will rely on the same scenario than $\exec$.

  By definition $\overline{\exec} = [\overline{\evt^{\sid_1}_1\sigma}^{\exec,
    \sid_1}; \ldots; \overline{\evt^{\sid_\ell}_\ell\sigma}^{\exec,
    \sid_\ell}]$. Let $i\in \{1, \ldots, \ell\}$, then we have that
  $\evt_i^{\sid_i} =\rcv(u)$ (or $\evt_i^{\sid_i} = \snd(u)$, or
  $\evt_i^{\sid_i} = \Q(u_1,\ldots,u_n)$). Since the three cases can be handled
  in a similar way, we consider here the case where $\evt_i^{\sid_i}
  =\rcv(u)$. By definition, we have that
  $\overline{\evt^{\sid_i}_i\sigma}^{\exec, \sid_i} =
  \rcv(\overline{u\sigma}^{\exec, \sid_i})$, and we prove by structural
  induction on $u'\in\st(u)$ that $\overline{u'\sigma}^{\exec, \sid_i} =
  u'\overline\sigma$. 
  Finally, from this we conclude that $\overline{u\sigma}^{\exec, \sid_i} =
  u\overline\sigma$, and thus that $\overline{\evt^{\sid_i}_i\sigma}^{\exec,
    \sid_i} = \rcv(\overline{u\sigma}^{\exec, \sid_i}) = \rcv(u \overline\sigma)
  = \evt_i^{\sid_i}\overline\sigma$. By definition, this brings us to
  $\overline\exec = \tr\overline\sigma$.  \qedhere
\end{proof}

%%
%% Notre transformation donne une trace well-formed.
%%

\newsavebox{\lemtracewellformed}
\sbox{\lemtracewellformed}{\vbox{%
\begin{lem}
\label{lem:trace-well-formed}
  Let $\Pi$ be a $k$-party protocol, and $\exec$ be an execution trace associated to $\widetilde{\Pi}$ (not necessarily a valid one). We have that $\overline{\exec}$ is well-formed.
\end{lem}
}}

\noindent\usebox{\lemtracewellformed}

\begin{proof}(sketch)
  Let $\exec = \evt_1^{\sid_1}, \ldots, \evt_\ell^{\sid_\ell}$. Let $i\in\{1, \dots, \ell\}$, we show that:
  \begin{enumerate}
  \item $\overline{\evt^{\sid_i}_i}^{\exec, \sid_i}$ is $k$-tagged;
  \item $\Tags(\overline{\evt^{\sid_i}_i}^{\exec, \sid_i}) \subseteq \{\tagt(\overline{\exec},\sid_i)\}$;
  \item $\names(\overline{\evt^{\sid_i}_i}^{\exec, \sid_i}) \subseteq
    \{n_t^{\epsilon,S} ~|~t \in T\} \cup \{n^{\sid}_{{y}} ~|~
    \sid \in S \mbox{ {and $y \in \Yvars$}}\}$, \\where $S = \sameTagas(\exec,\sid_i)$.
  \end{enumerate}

\noindent Let $\evt^{\sid_i}_i = \rcv(u)$ for some term $u$. The cases where $\evt^{\sid_i}_i = \snd(u)$ or $\evt^{\sid_i}_i = \Q(u_1, \ldots,u_n)$ can be done in a similar way. We have that $\overline{\evt^{\sid_i}_i}^{\exec, \sid_i} = \rcv(\overline{u}^{\exec, \sid_i})$ and we prove by structural induction on $u'\in\st(u)$ that:
  \begin{enumerate}
  \item $\overline{u'}^{\exec, \sid_i}$ is $k$-tagged;
  \item $\Tags(\overline{u'}^{\exec, \sid_i}) \subseteq \{\tagt(\overline{\exec},\sid_i)\}$;
  \item $\names(\overline{u'}^{\exec, \sid_i}) \subseteq \{n_t^{\epsilon,S} ~|~t \in T\} \cup \{n^{\sid}_{{y}} ~|~ \sid \in S \mbox{ {and $y \in \Yvars$}}\}$, \\where $S = \sameTagas(\exec,\sid_i)$.
  \end{enumerate}
  And from this we derive that $\overline{u}^{\exec, \sid_i}$ satisfies the three conditions of well-formedness, and thus so is $\overline{\evt^{\sid_i}_i}^{\exec, \sid_i}$ for all $i\in \{1, \ldots, \ell\}$, which in turn implies by definition that $\overline{\exec}$ satisfies the three conditions of well-formedness and is thus well-formed.
\end{proof}

\section{Technical proofs about alien subterms}
\label{sec:app-alien}

We introduce the notion of alien subterms and we show that they satisfy some
good properties. Later on, we will see that those alien subterms correspond to
the subterms that are abstracted by our transformation $\overline{\; \cdot \;}$
and we will use the properties established on them to prove the validity of the
trace obtained after transformation.

\begin{defi}[$\stal(\exec, \tau, t)$]\label{dfn:stalien}
  Let $\Pi$ be a $k$-party protocol and $\exec = [\evt_1^{\sid_1}; \dots;
  \evt_\ell^{\sid_\ell}]$ be an execution trace (not necessarily valid) of
  $\widetilde{\Pi}$. We define the alien subterms of a term~$t$ w.r.t.
  the execution~$\exec$ and the active tag~$\tau$, denoted $\stal(\exec, \tau,
  t)$, as follows:
  \begin{itemize}
    
  \item $\stal(\exec, \tau, n) = \{n\}$ if $n\in\Nces_\epsilon$

  \item $\stal(\exec, \tau, n^\sid_y) = \left\{
      \begin{array}[c]{ll}
        \emptyset & \text{if } \tagt(\exec, \sid) = \tau \text{ and } \tau \neq \bot \\
        \{n^\sid_y\} & \;\; \text{otherwise}
      \end{array}
    \right.$

  \item $\stal(\exec, \tau, a) = \emptyset$ if $a$ is an agent name

  \item $\stal(\exec, \tau, \f(a_1, \dots, a_n)) = \emptyset$ if $\f \in \{\shk, \pub, \priv\}$

  \item $\stal(\exec, \tau, \pair{u}{v}) = \stal(\exec, \tau, u) \cup \stal(\exec, \tau, v)$

  \item $\stal(\exec, \tau, \f(u_1, \dots, u_n)) = \left\{
      \begin{array}[c]{l}
        \stal(\exec, \tau, u_1) \cup \dots \cup \stal(\exec, \tau, u_n) \\
        \text{if } \tagh(\exec, \f(u_1, \dots, u_n)) = \tau \text{ and }\tau \neq\bot \\
        \\
        \{\f(u_1, \dots, u_n)\} \cup \underset{i\in\{1, \dots, n\}}{\bigcup} \stal(\exec,\tau' , u_i) \\
        \text{otherwise where $\tau' = \tagh(\exec, \f(u_1, \dots, u_n))$} \\
      \end{array}
    \right.$ if $\f\in\{\enc, \enca, \sign, \h\}$.
  \end{itemize}
  
\end{defi}
We define $\stal(\exec, t) = \stal(\exec, \bot, t)$, and extend this notion to
sets of terms in the obvious way, i.e. $\stal(\exec, T) = \underset{t\in
  T}{\bigcup}\stal(\exec, t)$.

\begin{defi}[$\vars(\exec,\sid)$, $\st(\exec,\sid)$, $\names(\exec,\sid)$]
  Let $\Pi$ be a $k$-party protocol and $\exec$ be an execution trace of
  $\widetilde{\Pi}$. Let $\sid$ be a session identifier, and $ [\evt^{\sid}_{1};
  \dots; \evt^{\sid}_{h}] \stackrel{\mathsf{def}}{=} \exec|_{\{\sid\}}$. We
  define the variables, subterms, and names in $\exec$ of a session~$\sid$ as
  follows:
  \[
  \begin{array}{rcl}
    \vars(\exec,\sid)  & = & \{x~|~x\in\vars(\evt^\sid_j) \mbox{ for some } j \in \{1,\ldots,h\}\} \\
    \st(\exec,\sid)    & = & \{u~|~u\in\st(\evt^\sid_j) \mbox{ for some } j \in \{1,\ldots,h\}\} \\
    \names(\exec,\sid) & = & \{u~|~u\in\names(\evt^\sid_j) \mbox{ for some } j \in \{1,\ldots,h\}\}.
  \end{array}
  \]
\end{defi}

{Since we do not tag the pairing function symbol, this function symbol has a
  special status. We denote by $\comp(t)$ the components of a term $t$. This
  notion is formally defined as follows:}

\begin{defi}[$\comp(t)$]\label{dfn:comp}
  Let $t$ be a term, the set of components of $t$ is:
  \[
  \comp(t) = \left\{
    \begin{array}[c]{lr}
      \comp(u) \cup \comp(v) & \hspace{1cm}\text{if } t = \pair{u}{v} \\
      \{t\}                  & \text{otherwise}.
    \end{array}
  \right.
  \]
\end{defi}

\newsavebox{\lemstaliengeneral}
\sbox{\lemstaliengeneral}{\vbox{%
    \begin{lem}
      \label{lem:staliengeneral}
      Let $\Pi$ be a $k$-party protocol, $\exec$ be an execution trace of
      $\widetilde{\Pi}$, and $t$ be a term. For all $k$-{tags}~$\tau$, we have
      that:
      \begin{enumerate}
      \item $\stal(\exec, \tau, t)\ =\ \underset{t'\in\comp(t)}{\bigcup}
        \stal(\exec, \tau, t')$; \label{lemitem:stalien.comp}
      \item $\stal(\exec, \tau, t)\ \subseteq\ \stal(\exec,
        t)$; \label{lemitem:tau.in.bot}
      \item $\stal(\exec, t)\ \subseteq\ \comp(t) \cup \stal(\exec, \tau,
        t)$. \label{lemitem:stal.0}
      \end{enumerate}
    \end{lem}}}

\noindent\usebox{\lemstaliengeneral}

\begin{proof}
  \noindent We prove each statement separately by induction on the depth of~$t$.
  \qedhere
\end{proof}

% S.D.: relu

\newsavebox{\lemstal}
\sbox{\lemstal}{\vbox{%
\begin{lem}\label{lem:stal}
  Let $\Pi$ be a $k$-party protocol, $\exec$ be an execution trace of $\widetilde{\Pi}$, and $u$ be a term. For any $v  \in \st(u)$, we have that $\stal(\exec, v) \subseteq \stal(\exec, u) \cup \comp(v)$.
\end{lem}
}}

\noindent\usebox{\lemstal}

\begin{proof} We first need to establish the following result:
  \[
  \forall \tau'\ \exists\ \tau\ \stal(\exec, \tau, v) \subseteq \stal(\exec, \tau', u).
  \]

  If $v = u$ then we can choose $\tau=\tau'$ to prove what we want. Otherwise, we have that $v \neq u$, and we prove the result by induction on the depth of $u$, and for this we distinguish three cases:
     
  \smallskip{}
  % Case 1
  \noindent\emph{Case $u = \f(u_1,\ldots,u_n)$ for some $\f \in \{\pub,\priv,\shk\}$.} Then we have that ${\stal(v) = \emptyset}$ for any $v \in \st(u)$. This allows us to easily conclude.

  \smallskip{}
  % Case 1
  \noindent\emph{Case $u = \pair{u_1}{u_2}.$} In that case $v\in\st(u_1)$ or $v\in\st(u_2)$. Suppose $v\in\st(u_1)$ and let $\tau'$ be a tag. By induction hypothesis, we have that there exists $\tau$ such that $\stal(\exec, \tau, v) \subseteq \stal(\exec, \tau', u_1)$. By Definition~\ref{dfn:stalien}, we have that $\stal(\exec, \tau', u_1) \subseteq \stal(\exec, \tau', u)$. Hence, we easily conclude. The case where $v\in\st(u_2)$ can be handled in a similar way. 
  
  \smallskip{}
  % Case 2
  \noindent\emph{Case $u = \f(u_1, \dots, u_n)$ for some $\f\in\{\enc, \enca, \sign, \h\}$.} In that case, we have that $v\in\st(u_{i_0})$ for some $i_0\in\{1, \dots, n\}$. Let $\tau'$ be a $k$-tag. According to Definition~\ref{dfn:stalien}, we have that $\underset{i\in\{1, \dots, n\}}{\bigcup} \stal(\exec, \tau'', u_i) \subseteq \stal(\exec, \tau', u)$ where $\tau'' = \tagh(\exec, u)$. Moreover, by induction hypothesis, we know  that there exists $\tau$ such that $\stal(\exec, \tau, v) \subseteq \stal(\exec, \tau'', u_{i_0})$. Hence, we deduce that $\stal(\exec, \tau, v) \subseteq \stal(\exec, \tau', u)$.
     
  \medskip{}

  This allows us to conclude that $\forall \tau'\ \exists \tau,\ \stal(\exec, \tau, v) \subseteq \stal(\exec, \tau', u)$. We have shown that $\forall \tau'\ \stal(\exec, \tau', u) \subseteq \stal(\exec, u)$ (see Lemma~\ref{lem:staliengeneral} - Item~\ref{lemitem:tau.in.bot}). Hence, we can infer that there exists $\tau$ such that $\stal(\exec, \tau, v) \subseteq \stal(\exec, u)$. Hence, we have that:

  $
  \begin{array}[b]{rcll}
    \stal(\exec, v) &\subseteq& \stal(\exec, \tau, v) \cup \comp(v)& \;\; \text{(Lemma~\ref{lem:staliengeneral} - Item~\ref{lemitem:stal.0})}\\
    &\subseteq& \stal(\exec, u) \cup \comp(v) &
  \end{array}
  $\qedhere
\end{proof}

%%
%% SOUS-TERMES ALIEN et DEDUCTION
%%

% S.D.: relu

\newsavebox{\lemstalT}
\sbox{\lemstalT}{\vbox{%
\begin{lem}\label{lem:stalT}
  Let $\Pi$ be a $k$-party protocol and $\exec$ be an execution trace of $\widetilde{\Pi}$. Let $T$ be a set of terms such that  $T \vdash v$ for any $v \in \stal(\exec, T)$, and $t$ be a term such that $T \vdash t$. We have $T \vdash u$ for any $u \in \stal(\exec, t)$.
\end{lem}
}}

\noindent\usebox{\lemstalT}

\begin{proof}
  Let $u\in\stal(\exec, t)$. We prove that $T\vdash u$ by induction on $\pi$, a prooftree witnessing the fact that $T \vdash t$.  If $\pi$ is reduced to a leaf then we have that $t\in T \cup \A \cup \Kcal_\epsilon \cup \Nces_\epsilon\cup \{\pub(a) ~|~ a \in \A\}$. Actually, if $t \in \A \cup \Kcal_\epsilon \cup \Nces_\epsilon \cup \{\pub(a) ~|~ a \in \A\}$, then $\stal(\exec,t) = \emptyset$, leading to a contradiction. Hence, we have that $t \in T$, and thus $u \in \stal(\exec,T)$.  We can thus conclude by hypothesis that $T \vdash u$.
  \smallskip{}
  
  Otherwise, we proceed by case analysis on the last rule used in the proof~$\pi$.
  \smallskip{}
  
  \noindent\emph{Case 1: the last rule is a composition rule.} Then $t = \f(t_1, \dots, t_n)$ for some terms $t_1, \dots, t_n$ and some $\f\in\{\pair{}{}, \enc, \enca, \sign, \h\}$. Let $\pi_1, \ldots, \pi_n$ be the direct subproofs of $\pi$. We have that $\pi_i$ is a proof of $T \vdash t_i$ for $i\in \{1,\ldots, n\}$. According to Definition~\ref{dfn:stalien} of alien subterms, $\stal(\exec, t) \subseteq \{t\} \cup \underset{i\in\{1, \dots, n\}}{\bigcup} \stal(\exec, \tau, t_i)$ for some $\tau$, and by Lemma~\ref{lem:staliengeneral} (Item~\ref{lemitem:tau.in.bot}) we can thus infer that 
  \[
  \stal(\exec, t) \subseteq \{t\} \cup \underset{i\in\{1, \dots, n\}}{\bigcup} \stal(\exec, t_i).
  \]
  If $u = t$, then by hypothesis we know that $T\vdash u$. On the other hand , if $u \in \stal(\exec, t_i)$ for some $i\in\{1, \dots, n\}$, then we  conclude by applying our induction hypothesis on $\pi_i$. In both cases, we have that $T \vdash u$.

  % Case 2
  \smallskip{}
  \noindent\emph{Case 2: the last rule is a projection rule.} Then $t = t_{i_0}$ for some terms $t_1$, $t_2$, and some $i_0\in\{1, 2\}$. Let $\pi'$ be the direct subproof of $\pi$. We have that $\pi$ is a proof of $T \vdash \pair{t_1}{t_2}$. According to Definition~\ref{dfn:stalien} of alien subterms, $\stal(\exec, t) \subseteq \stal(\exec, \pair{t_1}{t_2})$, i.e.  $u\in \stal(\exec, \pair{t_1}{t_2})$. We can thus conclude by applying our induction hypothesis on $\pi'$ that $T \vdash u$. 
  \smallskip{}

  \noindent\emph{Case 3: the last rule is another decomposition rule.} In such a case, there exists $t'$ such that one of the direct subproofs of $\pi$ is labeled with $\f(t,t')$. Let $\pi'$ be such a proof. Thanks to Lemma~\ref{lem:stal} we know that either $u \in \stal(\exec, \f(t, t'))$ or $u \in \comp(t)$. In the first case, we can conclude by applying our induction hypothesis on $\pi'$  that $T \vdash u$. In the second case, we know that by application of the projection rules one can derive $u$ from $t$, hence $T \vdash u$. 
\end{proof}

%%
%% SOUS-TERME ALIEN
%%

% S.D: relu - OK

\newsavebox{\lemstalun}
\sbox{\lemstalun}{\vbox{%
\begin{lem}\label{lem:stal.1}
  Let $\Pi$ be a $k$-party protocol and $\exec$ be an execution trace of $\widetilde{\Pi}$ associated to the symbolic trace $\tr = [\evt_1^{\sid_1}; \dots; \evt_\ell^{\sid_\ell}]$. Let $\sigma$ be the substitution such that $\dom(\sigma) = \vars(\tr)$ and $\exec = \tr\sigma$.
  \[
  \forall t \in \st(\tr, \sid)\ \ \ \stal(\exec, t\sigma)\ \subseteq\ \comp(t\sigma)\ \cup\ \underset{x\in \vars(t)}{\bigcup} \stal(\exec, \tau, x\sigma)
  \]
  where $\tau = \tagt(\exec, \sid)$.
\end{lem}
}}

\noindent\usebox{\lemstalun}

\begin{proof}
  Let $t \in\st(\tr, \sid)$ and  $\tau = \tagt(\exec, \sid)$. We show by structural induction on~$t$ that 
  \[
  \stal(\exec, t\sigma)\ \subseteq\ \comp(t\sigma)\ \cup\ \underset{x\in \vars(t)}{\bigcup} \stal(\exec, \tau, x\sigma)
  \]

  \noindent We distinguish several cases.
  \smallskip{}
  
  \noindent\emph{Case $t \in \mathcal{Y}$.} In such a case, we can easily conclude thanks to Lemma~\ref{lem:staliengeneral} (Item~\ref{lemitem:stal.0}). Indeed, we have that:
  \[
  \stal(\exec, t\sigma)  \subseteq  \comp(t\sigma) \cup \stal(\exec, \tau, t\sigma)  =   \comp(t\sigma) \cup \underset{x\in \vars(t)}{\bigcup} \stal(\exec, \tau, x\sigma)
  \]
  
  \noindent\emph{Case $t \in \mathcal{N}$.} Then  $t\sigma = t$, $\stal(\exec, t\sigma) = \{t\sigma\}$, and $\comp(t\sigma) = \{t\sigma\}$. Thus, we have that:
  \[ 
  \stal(\exec, t\sigma) = \{t\sigma\} \subseteq \comp(t\sigma)\ \cup\ \underset{x\in \vars(t)}{\bigcup} \stal(\exec, \tau, x\sigma)
  \]
  
  \noindent\emph{Case $t\in\A$ or $t = \f(a_1, \dots, a_n)$ for some $\f\in\{\shk, \pub, \priv\}$.}  In such a case, $\vars(t)=\emptyset$ and thus $\stal(\exec, t\sigma) = \emptyset$. This allows us to conclude.
  \smallskip{}
  
  \noindent\emph{Case $t = \pair{t_1}{t_2}$.} Then we have that $\stal(\exec, t\sigma) = \stal(\exec, t_1\sigma) \cup \stal(\exec, t_2\sigma)$. Applying our induction hypothesis, we deduce that  \\[2mm]
  \null \hfill $\stal(\exec, t_i\sigma)\ \subseteq\ \comp(t_i\sigma)\ \cup\ \underset{x\in \vars(t_i)}{\bigcup} \stal(\exec, \tau, x\sigma)$ for $i \in \{1,2\}$. \hfill\null
  
  Hence, we conclude that $\stal(\exec, t\sigma)\ \subseteq\ \comp(t\sigma)\ \cup\ \underset{x\in \vars(t)}{\bigcup} \stal(\exec, \tau, x\sigma)$ 
  \smallskip{}
  
  \noindent\emph{Case $t = \f(t_1, \dots, t_n)$ for some $\f\in\{\enc, \enca, \sign, \h\}$.} Let $\tau' = \tagh(\exec, t\sigma)$. We have that $\stal(\exec, t\sigma) = \{t\sigma\} \cup \underset{i\in\{1, \dots, n\}}{\bigcup} \stal(\exec, \tau', t_i\sigma)$. By construction of $\tr$, for all subterms $u\in\encst(\tr, \sid)$, $\tagh(\exec, u\sigma) = \tau$, thus ${\tau' = \tau}$. Thanks to Lemma~\ref{lem:staliengeneral} (Item~\ref{lemitem:tau.in.bot}), we have that $\stal(\exec, t\sigma) \subseteq \{t\sigma\} \cup \underset{i\in\{1, \dots, n\}}{\bigcup} \stal(\exec, t_i\sigma)$. We have that $\comp(t\sigma) = \{t\sigma\}$ and  thanks to our  induction hypothesis we have for each $i\in\{1, \dots, n\}$ the following inclusion
  \[
  \begin{array}[c]{l}
    \stal(\exec, t_i\sigma)\ \subseteq\ \comp(t_i\sigma)\ \cup\ \underset{x\in \vars(t_i)}{\bigcup} \stal(\exec, \tau, x\sigma)
  \end{array}
  \]
  Thus, $\stal(\exec, t\sigma) \subseteq \comp(t\sigma) \ \cup\ \underset{i\in\{1, \dots, n\}}{\bigcup} \comp(t_i\sigma)\ \cup\ \underset{x\in \vars(t)}{\bigcup} \stal(\exec, \tau, x\sigma)$.
  
  Now, in order to conclude, it remains to show that for all $u\in\stal(\exec, t\sigma)$, if $ u\in \underset{i\in\{1, \dots, n\}}{\bigcup} \comp(t_i\sigma)$ then there exists $x\in \vars(t)$ such that $u\in \stal(\exec, \tau, x\sigma)$. First, we notice the following:
  \[
  \begin{array}[c]{rcll}
    \stal(\exec, t\sigma) & {=} & \{t\sigma\} \cup \underset{i\in\{1, \dots, n\}}{\bigcup} \stal(\exec, \tau', t_i\sigma)& \text{(Definition~\ref{dfn:stalien})} \\
    
    & {=} & \{t\sigma\} \cup \underset{i\in\{1, \dots, n\}}{\bigcup} \underset{w\in\comp(t_i\sigma)}{\bigcup} \stal(\exec, \tau', w) & \text{(Lemma~\ref{lem:staliengeneral})} \\ % - Item~\ref{lemitem:stalien.comp})} \\
    
    & {=} & \{t\sigma\} \cup \underset{i\in\{1, \dots, n\}}{\bigcup} \underset{v\in\comp(t_i)}{\bigcup} \underset{w\in\comp(v\sigma)}{\bigcup} \stal(\exec, \tau', w) & \text{(Definition~\ref{dfn:comp})} \\
    
    & {=} & \{t\sigma\} \cup \underset{i\in\{1, \dots, n\}}{\bigcup} \underset{v\in\comp(t_i)}{\bigcup} \stal(\exec, \tau', v\sigma) & \text{(Lemma~\ref{lem:staliengeneral})}
  \end{array}
  \]
  Let $i\in\{1, \dots, n\}$ be such that $u \in \stal(\exec, t\sigma)$ and $u \in \comp(t_i\sigma)$. In that case, according to the equation stated above, there exists $j \in \{1,\ldots,n\}$ such that $v\in \comp(t_j)$ and $u\in\stal(\exec, \tau', v\sigma)$. We now proceed by case analysis on~$v$:
  
  \begin{itemize}
  \item \emph{Case $v \in \A$ or $v = \f(a_1,\ldots,a_n)$ for some $f \in \{\pub,\priv,\shk\}$.} In such a case, we have that $\stal(\exec,\tau',v\sigma) = \emptyset$. Thus, this case in not possible.
    
  \item \emph{Case $v \in \mathcal{N}$.} In such a case, we have that $u = v$ and by construction of $\tr$ we have that $v    = n^\sid_y$ for some variable~$y$. Since, $\tau =\tau'$, we have that $\stal(\exec,\tau',v\sigma) = \emptyset$. Thus, this case is not possible.
    
  \item \emph{Case $v = \mathsf{g}(v_1, \dots, v_m)$ for some $\mathsf{g}\in\{\enc, \enca, \sign, \h\}$.}  In such a case, we have that $u = v\sigma$ and by construction of $\tr$ we know that $\tagh(\exec, v\sigma) = \tagh(\exec, t\sigma) = \tau (=\tau')$.  Hence, we deduce that $v\sigma \not\in \stal(\exec,\tau',v\sigma)$, and thus $u \not\in \stal(\exec,\tau',v\sigma)$ leading again to a contradiction.
    
  \item \emph{Case $v$ is a variable.} In such a case, we have that $v \in \vars(t_j) \subseteq \vars(t)$. Hence, we have the expected conclusion.
  \end{itemize}
  
  \smallskip{}
  Altogether, this allows us to conclude that

  $
  \hspace{2cm}\begin{array}[b]{lcl}
    \stal(\exec, t\sigma) & \subseteq & \comp(t\sigma)\ \cup\ \underset{x\in \vars(t)}{\bigcup} \stal(\exec, \tau, x\sigma).
  \end{array}
  $\qedhere
  % \end{description}
\end{proof}

%%
%% SOUS-TERME ALIEN ORIGINATION
%%

%
% S.D: relu - OK

\newsavebox{\lemstalienorigination}
\sbox{\lemstalienorigination}{\vbox{%
\begin{lem}\label{lem:stalien.origination}
  Let $\Pi$ be a $k$-party protocol and $\exec = [\evt_1^{\sid_1}; \dots;
  \evt_\ell^{\sid_\ell}]$ be an execution trace of $\widetilde{\Pi}$,
  w.r.t. some set $T_0$ of ground atoms, associated to the symbolic trace
  $\tr = [\evt\evt_1^{\sid_1}; \dots; \evt\evt_\ell^{\sid_\ell}]$. Let $\sigma$
  be the substitution such that $\dom(\sigma) =\vars(\tr)$ and $\exec =
  \tr\sigma$. Let $\sid$ be a session identifier, $x$ be a variable in
  $\vars(\tr, \sid)$, $\tau =\tagt(\exec, \sid)$, and $u \in \st(\tr)$ such that
  $x \in \vars(u)$. We have that $\stal(\exec, \tau, x\sigma) \subseteq
  \stal(\exec, \tau, u\sigma)$.
\end{lem}
}}

\noindent\usebox{\lemstalienorigination}

\begin{proof}
  Let $u$ be a subterm of $\tr$ such that $x\in\vars(u)$. We prove the result by structural induction on~$u$. First, note that by construction of $\tr$, we have that $\vars(\tr, \sid')\cap \vars(\tr, \sid'') = \emptyset$ when $\sid' \neq \sid''$. Hence,  for all $i \in \{1, \dots, \ell\}$ if $x\in\vars(\evt\evt_i^{\sid_i})$, then $\sid_i = \sid$; and thus, for all $i\in\{1, \dots, \ell\}$ such that $u \in \st(\evt\evt^{\sid_i}_i)$, we know that $\sid_i = \sid$. Now, since $x\in\vars(u)$, we know that $u$ is not  ground, and we only need to consider the three following cases: 
  \smallskip{}

  % Case 1
  \noindent\emph{Case $u \in \mathcal{Y}$.} In this case $u = x$, and the result trivially holds.
  \smallskip{}

  % Case 2
  \noindent\emph{Case $u = \pair{u_1}{u_2}$ for some terms $u_1$ and $u_2$.} In that case, $x\in\vars(u_1)$ or $x\in\vars(u_2)$. Assume that $x \in \vars(u_1)$. The other case can be handled in a similar way. By induction hypothesis, we know that $\stal(\exec, \tau, x\sigma) \subseteq \stal(\exec, \tau, u_1\sigma)$ and we have that $\stal(\exec, \tau, u_i\sigma) \subseteq \stal(\exec, \tau, u\sigma)$. Combining these two we easily conclude.
  \smallskip{}

  % Case 3
  \noindent\emph{Case $u = \f(u_1, \dots, u_n)$ for some $\f\in\{\enc, \enca, \sign, \h\}$ and some terms $u_1, \dots, u_n$.} In that case $x\in \vars(u_i)$ for some $i\in\{1, \dots, n\}$. Let $j\in \{1, \dots, n\}$ such that $x\in \vars(u_j)$. Now, by construction of $\tr$, we know that $\tagh(\exec, u\sigma) =  \tagt(\exec, \sid)$, hence we have that $\stal(\exec, \tau, u_j\sigma) \subseteq \stal(\exec, \tau, u\sigma)$. Applying our induction hypothesis on $u_j$, we deduce that $\stal(\exec, \tau, x\sigma) \subseteq \stal(\exec, \tau, u_j\sigma)$. This allows us to conclude.
\end{proof}

%%
%% LEMME TRACE ST ALIEN
%%

Now, we can show that the alien subterms that occur in a valid trace are deducible.

\newsavebox{\lemtracestalien}
\sbox{\lemtracestalien}{\vbox{%
\begin{lem}\label{lem:trace.stalien}
  Let $\Pi$ be a $k$-party protocol and $\exec = [\evt_1^{\sid_1}; \dots; \evt_\ell^{\sid_\ell}]$ be  an execution trace of $\widetilde{\Pi}$ that is valid w.r.t. some  set $T_0$ of ground atoms. Let $i \in \{0, \ldots, \ell\}$. We have that $\K(\exec_i) \cup T_0 \vdash u$ for any $u \in \stal(\exec, \K(\exec_i) \cup T_0)$.
\end{lem}
}}

\noindent\usebox{\lemtracestalien}

\begin{proof}
  Let $\tr = [\evt\evt_1^{\sid_1}; \dots; \evt\evt_\ell^{\sid_\ell}]$ be the symbolic trace associated to~$\exec$. Let $\sigma$ be the substitution such that $\dom(\sigma) = \vars(\tr)$ and $\exec = \tr\sigma$. We prove the result by induction on~$i$. The base case, where $i =0$ is obvious since $\K(\exec_i) = \emptyset$ and $\stal(\exec, T_0) \subseteq T_0$. Now, to deal with the inductive case, we distinguish three cases depending on the nature of the last event in $\exec_i$.
  \smallskip{}

  % Case 1
  \noindent\emph{Case $\evt^{\sid_i}_i = \Prd(t_1, \dots, t_n)$.} Then, we have that $\K(\exec_i) = \K(\exec_{i-1})$ and thus that $\stal(\exec,\K(\exec_i) \cup T_0) = \stal(\exec,\K(\exec_{i-1}) \cup T_0)$. Thanks to our induction hypothesis, we know that $\K(\exec_{i-1}) \cup T_0 \vdash \stal(\exec, \K(\exec_{i-1})\cup T_0))$, thus we easily conclude.
  \smallskip{}

  % Case 2
  \noindent\emph{Case $\evt^{\sid_i}_i = \rcv(t)$.} This case is similar to the previous one.
  \smallskip{}

  % Case 3
  \noindent\emph{Case $\evt^{\sid_i}_i = \snd(t)$.} In such a case, we have that $\evt\evt^{\sid_i}_i = \snd(t')$ for some term~$t'$ such that $t = t'\sigma$. Let $u\in\stal(\exec, \K(\exec_i) \cup T_0))$. The only case for which we can not easily conclude by applying our induction hypothesis is when  $u \in \stal(\exec, t)$. So, assume that $u \in \stal(\exec,t)$. According to Lemma~\ref{lem:stal.1}, $u \in \comp(t'\sigma) \cup \underset{x\in\vars(t')}{\bigcup} \stal(\exec, \tau, x\sigma)$ where $\tau = \tagt(\exec, \sid_i)$. We distinguish two cases:
  \begin{enumerate}
  \item  $u \in \comp(t'\sigma)$. We have that $t'\sigma = t \in \K(\exec_{i})$ and thus $\K(\exec_i) \cup T_0 \vdash u$.
    
  \item $u \in \stal(\exec, \tau, x\sigma)$ for some $x \in \vars(t')$ and $u\not\in \comp(t'\sigma)$. By the origination property we know that there exists $j<i$ such that $\sid_{j} = \sid_i$, $\evt\evt^{\sid_j}_j = \rcv(v')$ with $x\in\vars(v')$, and thus that $x\sigma\in\st(v'\sigma)$. By Lemma~\ref{lem:stalien.origination}, we deduce that $u \in \stal(\exec, \tau, v'\sigma)$, and thanks to Lemma~\ref{lem:staliengeneral} (Item~\ref{lemitem:tau.in.bot}), we have that $u \in \stal(\exec, v'\sigma)$. We can then apply our induction hypothesis in order to deduce that $\K(\exec_{j-1}) \cup T_0 \vdash w$ for any $w \in \stal(\exec, \K(\exec_{j-1})\cup T_0)$, and because $\exec$ is a valid trace, we have also that $\K(\exec_{j-1})\cup T_0 \vdash v'\sigma$. Thus, according Lemma~\ref{lem:stalT}, we deduce that $\K(\exec_{j-1}) \cup T_0 \vdash w$ for any $w \in \stal(\exec, v'\sigma)$. In particular, we conclude that $\K(\exec_{j-1}) \cup T_0 \vdash u$.\qedhere
  \end{enumerate}
\end{proof}

\section{Proofs of section~\ref{subsec:validity}}
\label{sec:appvalidity}

In order to show the validity of the resulting trace, we first characterize the subterms that are abstracted by our transformation.
Actually, we can show that those subterms are alien subterms, 
and thus they enjoy the properties established in Appendix~\ref{sec:app-alien}.

%%
%% ABSTACTED => ALIEN
%%

\newsavebox{\lemabstractedimpliesalien}
\sbox{\lemabstractedimpliesalien}{\vbox{%
\begin{lem}\label{lem:abstracted.implies.alien}
  Let $\Pi$ be a $k$-party protocol, $\exec$ be an execution trace of $\widetilde{\Pi}$, $t$ be a term and $p$ be a position. If there exists $\sid$ such that $(\overline{t}^{\exec, \sid})|_p\in\mathcal{N}_\epsilon$, then  we have that $t|_p \in \stal(\exec, t)$.
\end{lem}
}}
\noindent\usebox{\lemabstractedimpliesalien}

\begin{proof}
  We will prove by induction on~$p$ that $t|_p \in \stal(\exec, t)$.
  
  \smallskip{}
  \noindent\emph{Base case $p = \epsilon$.} In that case, according to Definition~\ref{def:abstract.term}, either $t\in\Nces$, or $t = \f(t_1, \dots, t_n)$ for some $\f\in\{\enc, \enca, \sign, \h\}$. If $t \in \mathcal{N}$, we have that $\stal(\exec, t) = \{t\}$, and thus $t|_p = t|_\epsilon \in \stal(\exec, t)$. Otherwise, i.e. $t = \f(t_1, \dots, t_n)$ for some $\f\in\{\enc, \enca, \sign, \h\}$, then we have that
  \[
  \stal(\exec, t) = \{t\} \cup \underset{i\in\{1, \dots, n\}}{\bigcup} \stal(\exec, \tau, t_i)
  \]
  where $\tau = \tagh(\exec, t)$. Hence, we have that $t|_p = t|_\epsilon \in \stal(\exec, t)$. 

  \smallskip{}
  \noindent\emph{Inductive case $p = i_0\cdot q$.} First, note that~$t$ cannot
  be a long-term key, i.e. $t$ is not a term of the form $\pub(t')$,
  $\priv(t')$ or $\shk(t_1,t_2)$. Indeed, in such a case, $(\overline{t}^{\exec,
    \sid})|_p\not\in\mathcal{N}_\epsilon$ for any $\sid$. This would contradict
  one of our hypothesis. Thus, two cases remain:

  % Case 1
  \smallskip{}
  
  \noindent {\bf Case $t = \pair{t_1}{t_2}$.} By Definition~\ref{def:abstract.term}, $\overline{t}^{\exec, \sid} = \pair{\overline{t_1}^{\exec, \sid}}{\overline{t_2}^{\exec, sid}}$. Suppose $i_0 = 1$. Then $(\overline{t_1}^{\exec, \sid})|_q \in \mathcal{N}_\epsilon$, and thanks to our induction hypothesis we can derive that $t_1|_q \in \stal(\exec, t_1)$. We have that $\stal(\exec, t) = \stal(\exec, t_1) \cup \stal(\exec, t_2)$, thus $t_1|_q \in \stal(\exec, t)$. Finally, since $t_1|_q = t|_{1\cdot q} = t|_p$ we can conclude that $t|_p \in \stal(\exec, t)$. The case where $i_0 = 2$ can be done in a  similar way.

  % Case 2
  \smallskip{}
  \noindent {\bf Case $t = \f(t_1, \dots, t_n)$ for some $\f\in\{\enc, \enca, \sign, \h\}$.} Since     $ (\overline{t}^{\exec, \sid})|_{i_{0}.q}\in\Nces_\epsilon$, we know that  $(\overline{t}^{\exec, \sid})|_\epsilon\not\in\Nces_\epsilon$. Hence, by Definition~\ref{def:abstract.term}, we have that
  \begin{itemize}
  \item $\overline{t}^{\exec, \sid} = f(\overline{t_1}^{\exec, \sid}, \dots, \overline{t_n}^{\exec, \sid})$;
  \item  $\tagh(\exec, t) = \tagt(\exec, \sid) \not= \bot$; and
  \item $ (\overline{t}^{\exec, \sid})|_{{i_0}.q} = (\overline{t_{i_0}}^{\exec, \sid})|_q$.
  \end{itemize}
  Hence, we have that  $(\overline{t_{i_0}}^{\exec, \sid})|_q \in \Nces_\epsilon$. Thanks to our induction hypothesis, we deduce that  $t_{i_0}|_q \in \stal(\exec, t_{i_0})$. Applying Lemma~\ref{lem:stal}, we conclude that  $t_{i_0}|_q \in \comp(t_{i_0}) \cup \stal(\exec, t)$.
 
  \smallskip{}

  In order to conclude, it is sufficient to show that if $t_{i_0}|_q \in \comp(t_{i_0})$, then we also have that $t_{i_0}|_q \in \stal(\exec, t)$. Assume that $t_{i_0}|_q \in \comp(t_{i_0})$. First, let $\tau = \tagh(\exec, t)$, thanks to Lemma~\ref{lem:staliengeneral} (item~\ref{lemitem:stalien.comp}), we have that:
  \[
  \stal(\exec, t) = \{t\} \cup \underset{j\in\{1, \dots, n\}}{\bigcup} \stal(\exec, \tau, t_j) = \{t\} \cup \bigcup_{i=1}^{n} \underset{t'_j\in\comp(t_j)}{\bigcup} \stal(\exec, \tau, t'_j)
  \]
  Hence, we have that $\stal(\exec, \tau, (t_{i_0})|_q) \subseteq \stal(\exec, t)$. To conclude, it is hence enough to show that $t_{i_0}|_q \in \stal(\exec, \tau, t_{i_0}|_q)$. Since $(\overline{t}^{\exec, \sid})|_p = (\overline{t}^{\exec, \sid})|_{{i_0}.q} = (\overline{t_{i_0}}^{\exec, \sid})|_q \in \Nces_\epsilon$, we need to distinguish three cases:
      
  \begin{itemize}
  \item \emph{Case $(t_{i_0})|_q \in \Nces_\epsilon$.} In such a case, we have that $\stal(\exec, \tau, (t_{i_0})|_q) = \{(t_{i_0})|_q)\}$.

  \item \emph{Case $(t_{i_0})|_q \in \Nces\smallsetminus\Nces_\epsilon$.} In such a case, we have that $(t_{i_0})|_q = n^{\sid'}_y$ for some $sid'\not\in\sameTagas(\exec, \sid)$, thus  $\stal(\exec, \tau, (t_{i_0})|_q) = \{(t_{i_0})|_q)\}$.

  \item \emph{Case $(t_{i_0})|_q = \mathsf{g}(u_1, \dots, u_m)$ for some $\mathsf{g}\in\{\enc, \enca, \sign, \h\}$.} In such a case, we have that $\tagh(t_{i_0}|_q)  \neq \tau$, and thus 

    $
    \stal(\exec, \tau, (t_{i_0})|_q) = \{(t_{i_0})|_q)\} \cup \underset{j\in\{1, \dots, m\}}{\bigcup} \stal(\exec, \tagh(t_{i_0}|_q), u_j).
    $
\qedhere
  \end{itemize}
\end{proof}
%%
%% LA transformation preserve les inegalites
%%

We show that our transformation preserves disequalities even if terms
are  not abstracted using the same session identifier. 
This result can be proved by structural induction on $\overline{m}^{\exec, \sid}$. 

%S.D: OK relu
\newsavebox{\lemmu}
\sbox{\lemmu}{\vbox{%
\begin{lem}\label{lem:mu}
  Let $\Pi$ be a $k$-party protocol and $\exec = [\evt^{\sid_1}_1; \ldots; \evt^{\sid_\ell}_\ell]$ be a valid execution trace of $\widetilde{\Pi}$, w.r.t. some initial intruder knowledge $T_0$. Let $m$ and $m'$ be two terms such that $m \neq m'$, and  $\sid$, $\sid'$ be two session identifiers. We have that $\overline{m}^{\exec, \sid}\not= \overline{m'}^{\exec, \sid'}$.
\end{lem}
}}

\noindent\usebox{\lemmu}

\noindent\usebox{\lemdeducibilityofabstractedterms}

\begin{proof}
  Let $\tr = [\evt\evt_1^{\sid_1}; \dots; \evt\evt_\ell^{\sid_\ell}]$ be the symbolic trace associated to $\exec$ and~$\sigma$ be the substitution such that $\dom(\sigma) = \vars(\tr)$ and $\exec = \tr\sigma$.
  Let $i \in \{0, \ldots, \ell\}$.  Let $\pi$ be a simple proof of
  $\K(\exec_i) \cup T_0 \vdash t$. We prove that
  $\K(\overline{\exec}_i) \cup T_0 \vdash \overline{t}^{\exec, \sid}$
  by induction on $(i, \pi)$. If $i = 0$ and $\pi$ is a simple proof
  reduced to {a leaf} (possibly followed by some projection
  rules), then we have that $T_0 \vdash t$, and $\pi$ is necessarily
  reduced to {a leaf} since $T_0$ only contains atomic terms. Let $\sid$ be a session identifier, we have that $\overline{t}^{\exec, \sid} \in \{t\} \cup \mathcal{N}_\epsilon$. This allows us to conclude that $T_0 \vdash \overline{t}^{\exec, \sid}$ Now, we distinguish several cases depending on the last rule of~$\pi$.
  \smallskip{}

  \noindent\emph{The proof $\pi$ ends with an instance of a composition rule, i.e. $t= \f(t_1, \dots, t_n)$ for some $\f\in\{\pair{}{}, \enc, \enca, \sign, \h\}$ and some terms $t_1,\ \dots,\ t_n$.}

  According to Definition~\ref{def:abstract.term}, we have that $\overline{t}^{\exec, \sid} \in \mathcal{N}_\epsilon \cup \{\f(\overline{t_1}^{\exec, \sid}, \dots, \overline{t_n}^{\exec, \sid})\}$. If $\overline{t}^{\exec, \sid} \in \mathcal{N}_\epsilon$, we easily conclude that $\K(\overline{\exec}_i) \cup T_0 \vdash \overline{t}^{\exec, \sid}$. Otherwise, since $\pi$ ends with a composition rule, we have that $\K(\exec_i) \cup T_0 \vdash t_1$, \dots, $\K(\exec_i) \cup T_0 \vdash t_n$. Moreover, the simple proofs witnessing these facts are strict subproofs of $\pi$ that are also simple. Hence, we can apply our induction hypothesis in order to conclude that $\K(\overline{\exec}_i) \cup T_0 \vdash \overline{t_1}^{\exec, \sid}$, \dots, $\K(\overline{\exec}_i) \cup T_0 \vdash \overline{t_n}^{\exec, \sid}$. This allows us to conclude that $\K(\overline{\exec}_i) \cup T_0 \vdash  \f(\overline{t_1}^{\exec, \sid}, \dots, \overline{t_n}^{\exec, \sid})$.
  
  % Decomposition
  \smallskip{}
  \noindent\emph{The proof ends with the application of a
    decomposition rule (but not a projection) possibly followed by
    several applications of the projection rules until the resulting
    term is not a pair.} We will here present the case of the
  symmetric decryption rule, but all the other decomposition rules
  (including the case where the proof is reduced to a leaf) can be handled in a similar way. For some terms $t_1$ and $t_2$, the proof~$\pi$ is of the form
  \[
  \prooftree
  \prooftree
  \prooftree
  \prooftree
  \vdots
  \justifies
  \K(\exec_i) \cup T_0 \vdash \enc(t_1, t_2)
  \endprooftree
  \quad
  \prooftree
  \vdots
  \justifies
  \K(\exec_i) \cup T_0 \vdash t_2
  \endprooftree
  \justifies
  \K(\exec_i) \cup T_0 \vdash t_1
  \endprooftree
  \justifies
  \vdots
  \endprooftree
  \justifies
  \K(\exec_i) \cup T_0 \vdash t
  \endprooftree
  \]

  Let us first note that, by locality (Lemma~\ref{lem:locality}) and
  by simplicity of $\pi$ we know that $\enc(t_1, t_2) \in
  \st(\K(\exec_i)) \cup T_0 \cup \Kcal_\epsilon \cup \Nces_\epsilon
  \cup \{\pub(a) ~|~ a \in \mathcal{A}\}$, and by atomicity of $T_0$,
  $\Nces_\epsilon$, $\Kcal_\epsilon$ and  $\{\pub(a) ~|~ a \in
  \mathcal{A}\}$, we know that $\enc(t_1, t_2) \in
  \st(\K(\exec))$. (In case of a {proof reduced to a leaf}, and if there is no projection rule,  we may have that $t \in T_0$. In such a case, as in the base case, we have that $T_0 \vdash  \overline{t}^{\exec, \sid}$ and we easily conclude.)
  Hence,   there exists $k \le i$ such that  $\evt_k^{\sid_k} = \snd(u)$ and $\enc(t_1, t_2) \in \st(u)$. Let $k_0$ be the smallest such $k$ and $u_0$, $u_0'$ be such that $\evt^{\sid_{k_0}}_{k_0} = \snd(u_0)$ and $\evt\evt^{\sid_{k_0}}_{k_0} = \snd(u_0')$. Hence, we have that $u_0 = u_0'\sigma$.
  \smallskip{}

  \noindent In order to prove the result, we first establish the following claim.
  \smallskip{}

  \noindent{\bf Claim:} We have that $\overline{\enc(t_1, t_2)}^{\exec, \sid_{k_0}} = \enc(\overline{t_1}^{\exec, \sid_{k_0}}, \overline{t_2}^{\exec, \sid_{k_0}})$.
  \smallskip{}
 
  Assume by contradiction, that this equality does not hold.

  \smallskip{}
  %
%  \myrto{Est-ce que c'est mieux la?}
  %\noindent 
First, we have that $\enc(\overline{t_1}^{\exec, \sid_{k_0}}, \overline{t_2}^{\exec, \sid_{k_0}}) \not\in \st(\overline{u_0}^{\exec, \sid_{k_0}})$. Indeed,  for having $\enc(\overline{t_1}^{\exec, \sid_{k_0}}, \overline{t_2}^{\exec, \sid_{k_0}}) \in \st(\overline{u_0}^{\exec, \sid_{k_0}})$, there must exist $v\in \st(u_0)$ such that $\overline{v}^{\exec, \sid_{k_0}} = \enc(\overline{t_1}^{\exec, \sid_{k_0}}, \overline{t_2}^{\exec, \sid_{k_0}})$. But this would imply that $v = \enc(t'_1, t'_2)$ for some terms $t'_1, t'_2$ such that $\overline{t'_1}^{\exec, \sid_{k_0}} = \overline{t_1}^{\exec, \sid_{k_0}}$ and $\overline{t'_2}^{\exec, \sid_{k_0}} = \overline{t_2}^{\exec, \sid_{k_0}}$. However this would in turn imply according to Lemma~\ref{lem:mu} that $t'_1 = t_1$ and $t'_2 = t_2$. In other words we would have $v = \enc(t_1, t_2)\in \st(u_0)$ but with $\overline{v}^{\exec, \sid_{k_0}} = \enc(\overline{t_1}^{\exec, \sid_{k_0}}, \overline{t_2}^{\exec, \sid_{k_0}})$ which would contradict our hypothesis. Hence, necessarily $\enc(\overline{t_1}^{\exec, \sid_{k_0}}, \overline{t_2}^{\exec, \sid_{k_0}}) \not\in \st(\overline{u_0}^{\exec, \sid_{k_0}})$.

  Now since $\enc(\overline{t_1}^{\exec, \sid_{k_0}}, \overline{t_2}^{\exec, \sid_{k_0}}) \not\in \st(\overline{u_0}^{\exec, \sid_{k_0}})$, while $\enc(t_1, t_2)\in \st(u_0)$, there must exist a position $p$ (smaller or equal to the position where $\enc(t_1, t_2)$ occurs in $u_0$) such that $(\overline{u_0}^{\exec, \sid_{k_0}})|_p \in \Nces_\epsilon$ and $\enc(t_1, t_2) = u_0|_p$. Hence, Lemma~\ref{lem:abstracted.implies.alien} tells us that ${u_0}|_p \in \stal(\exec, u_0)$. Thanks to Lemma~\ref{lem:stal.1} and Lemma~\ref{lem:staliengeneral} (Item~\ref{lemitem:tau.in.bot}), we conclude that:
  \[
  {u_0}|_p \in \comp(u_0) \cup \underset{x\in\vars(u_0')}{\bigcup} \stal(\exec, x\sigma)
  \]

  \noindent We now distinguish two cases and show that each case leads us to a contradiction.

  \noindent\emph{Case 1:} ${u_0}|_p \in \comp(u_0) \smallsetminus \underset{x\in\vars(u_0')}{\bigcup} \stal(\exec, x\sigma)$. In such a case, there exists $u_0''\in\comp(u_0')$ such that $u_0|_p \in \comp(u_0''\sigma)$. But because we are considering the case where ${u_0}|_p \not\in \underset{x\in\vars(u_0')}{\bigcup} \stal(\exec, x\sigma)$, it must be that ${u_0}|_p = u''_0\sigma$. Now, by construction of $\tr$, it must be that $\tagt(\exec, \sid_{k_0}) = \tagh(\exec, u_0''\sigma)$, and thus $\overline{(u_0)|_p}^{\exec, \sid_{k_0}} \not\in \Nces_\epsilon$. Finally, because $(u_0)|_p\in\comp(u_0)$, we have that $ (\overline{u_0}^{\exec, \sid_{k_0}})|_p = \overline{(u_0)|_p}^{\exec, \sid_{k_0}}$. However, this equality is not possible since we have shown that $(\overline{u_0}^{\exec, \sid_{k_0}})|_p \in \Nces_\epsilon$ whereas  $\overline{(u_0)|_p}^{\exec, \sid_{k_0}} \not\in \Nces_\epsilon$. Hence, we obtain a contradiction.
  \smallskip{}

  \noindent\emph{Case 2:} ${u_0}|_p \in \underset{x\in\vars(u_0')}{\bigcup} \stal(\exec, x\sigma)$. In such a case, there exists $x \in \vars(u'_0)$ such that ${u_0}|_p\in\stal(\exec, x\sigma)$  and  $\enc(t_1, t_2)\in\st(x\sigma)$. Thanks to the origination property (see Definition~\ref{def:role} - Condition~\ref{def:role:or}), we know that there exists $j < k_0$ such that $\evt\evt_j^{\sid_j} = \rcv(v')$ and $x \in \vars(v')$. Hence, we have that $\enc(t_1,t_2) \in \st(v'\sigma)$. Since $\exec$ is a valid trace, we have that $\K(\exec_{j-1}) \cup T_0 \vdash v'\sigma$.
  
  Let $\pi'$ be a simple proof of $\K(\exec_{j-1}) \cup T_0 \vdash v'\sigma$, and $\pi''$ be a minimal subproof of $\pi'$ whose root is labeled with a term  $t'$  such that $\enc(t_1, t_2) \in \st(t')$. By locality of $\pi'$ (Lemma~\ref{lem:locality}), and because $\enc(t_1, t_2) \not\in \st(\K(\exec_{j-1}))$ (remember here that we choose $k_0$ such that for all $j<k_0$, we have that $\enc(t_1, t_2) \not\in \st(\K(\exec_{j-1}))$), we know that $\pi''$ ends with a composition rule. Unless $t' = \enc(t_1,t_2)$, this contradicts the minimality of $\pi''$. Hence, we have that $t' =\enc(t_1,t_2)$ and $\pi''$ is a simple proof of $\enc(t_1,t_2)$ whose last rule is a composition. Actually, since $\enc(t_1, t_2) \not\in  \st(\K(\exec_{j-1}) \cup T_0)$, any simple proof of $\enc(t_1,t_2)$ ends with a composition. This will contradict the fact that $\pi$ is a simple proof of $t$.
  \smallskip{}

  This allows us to conclude the proof of the claim. 
  \smallskip{}

  \noindent Now, by relying on our claim and by applying our induction hypothesis, we have that:
  \begin{itemize}
  \item $\K(\overline{\exec}_i) \cup T_0 \vdash \enc(\overline{t_1}^{\exec, \sid_{k_0}}, \overline{t_2}^{\exec, \sid_{k_0}})$; and 

  \item $\K(\overline{\exec}_i) \cup T_0 \vdash \overline{t_2}^{\exec, \sid_{k_0}}$.
  \end{itemize}
  This allows us to deduce that $\K(\overline{\exec}_i) \cup T_0 \vdash \overline{t_1}^{\exec, \sid_{k_0}}$. 
  \smallskip{}

  In order to establish that $\K(\overline{\exec}_i) \cup T_0 \vdash \overline{t}^{\exec,\sid}$, we need to distinguish several cases: 

  \begin{description}
% S.D.: 17/10/2012 J'ai legerement modifie la preuve. En fait, je ne comprends pas pourquoi il faudrait faire appel a la plaintext origination property pour les clef privees. En plus, cette propriete n'est pas en hypothese du lemme que l'on est en train de prouver.

  % \item[Case $t \in \A$ or $t = \pub(a)$] In such a case, we have that  $\overline{t}^{\exec, \sid} = t$. Hence, we have that $\K(\overline{\exec}_i) \cup T_0 \vdash \overline{t}^{\exec, \sid}$. 
    
  % \item[Case $t = \f(a_1, \dots, a_n)$ for some $\f\in\{\shk, \priv\}$] According to Lemma~\ref{lem:plaintext},we have that $t\in\plaintext(\tr) \cup T_0\cup \mathcal{K}_\epsilon$. By hypothesis, we have no long-term keys appearing in plaintext in $\widetilde{\Pi}$, and thus by construction of $\tr$ (see Definition~\ref{def:scenario}) we have no long-term keys appearing in plaintext in $\tr$. Hence, $t\in\mathcal{K}_\epsilon \cup T_0$. Moreover, by Definition we have $\overline{t}^{\exec, \sid} = t \in \mathcal{K}_\epsilon \cup T_0$. Hence, we have that $\K(\overline{\exec}_i) \cup T_0 \vdash \overline{t}^{\exec,\sid}$. 
      
\item[Case $t \in \A$, $t = \pub(a)$ or  $t = \f(a_1, \dots, a_n)$ for some $\f\in\{\shk, \priv\}$] \hfill

\noindent In such a case, we have that  $\overline{t}^{\exec, \sid} =
\overline{t}^{\exec, \sid_{k_0}} = t$. Hence, we have that
$\K(\overline{\exec}_i) \cup T_0 \vdash \overline{t}^{\exec, \sid}$ by
applying some projection rules on the proof of $\K(\overline{\exec}_i) \cup T_0 \vdash \overline{t_1}^{\exec, \sid_{k_0}}$. 
    
  \item[Case $t \in \mathcal{N}$ or $t = \f(t'_1, \dots, t'_m)$ for some $\f\in\{\enc, \enca, \h, \sign\}$] \hfill

\noindent If $\overline{t}^{\exec, \sid}\in \comp(\overline{t_1}^{\exec, \sid_{k_0}})$, then we easily conclude that $\K(\overline{\exec}_i)\cup T_0 \vdash \overline{t}^{\exec, \sid}$ since we have established that $\K(\overline{\exec}_i)\cup T_0 \vdash \overline{t_1}^{\exec, \sid_{k_0}}$. Otherwise, we have that $\overline{t}^{\exec, \sid}\not\in \comp(\overline{t_1}^{\exec, \sid_{k_0}})$. In that case, and according to Definition~\ref{def:abstract.term} and Lemma~\ref{lem:staliengeneral} (item~\ref{lemitem:stalien.comp}), either $\overline{t}^{\exec, \sid} \in \mathcal{N}_\epsilon$ or $\overline{t}^{\exec, \sid_{k_0}} \in \mathcal{N}_\epsilon$. In the first case, we trivially conclude. In the second case, i.e. $\overline{t}^{\exec, \sid} \not\in \mathcal{N}_\epsilon$ but $\overline{t}^{\exec, \sid_{k_0}} \in \mathcal{N}_\epsilon$, we have that $t\in\stal(\exec, \enc(t_1, t_2))$ (thanks to Lemma~\ref{lem:abstracted.implies.alien}. Since $\enc(t_1,t_2) \in \st(u_0)$, we deduce that $t\in\stal(\exec,u_0) \cup \comp(\enc(t_1, t_2))$ by applying Lemma~\ref{lem:stal}. Now, since $t \neq \enc(t_1,t_2)$, we deduce that $t \in \stal(\exec,u_0)$. Thus, applying Lemma~\ref{lem:stal.1}, we have that 
    \[
    t\in \stal(\exec, u_0) \subseteq \comp(u_0) \cup \underset{x\in\vars(u_0')}{\bigcup} \stal(\exec, \tau, x\sigma)
    \]
    \noindent where $\tau = \tagt(\exec, \sid_{k_0})$.

    Assume that $t \in \comp(u_0)$ and $t  \not\in \stal(\exec,\tau,x\sigma)$ for any $x \in \vars(u'_0)$. In such a case, we have that there exists $t' \in \comp(u'_0)$ such that $t \in \comp(t'\sigma)$ and we know that $t'
    \not\in \vars(u'_0)$. Hence $t$ is either a nonce and we have that $t=t'$. Moreover, we know that $t' = {n_y}^{\sid_{k_0}}$ for some $y$ (by construction of $\tr$. In such a case, $\overline{t}^{\exec,\sid_{k_0}} \not \in \Nces_\epsilon$. This leads us to a contradiction. Otherwise $t$ is an encrypted term and we have that $t = t'\sigma$ and again by construction of $\tr$, we have that $\overline{t'\sigma}^{\exec,\sid_{k_0}} \not\in \Nces_\epsilon$, leading us to a contradiction. Hence, we know that this case is not possible.
 
    Hence, we have that $t \in \stal(\exec,\tau,x\sigma)$ for some $x \in \vars(u'_0)$. Thanks to the origination property, we know that there exists $j < k_0$ such that $\sid_j = \sid_{k_0}$, $\evt\evt_j^{\sid_j} = \rcv(v')$ with $x \in \vars(v')$. Hence, we have that $x\sigma \in \st(v'\sigma)$. Then, applying Lemma~\ref{lem:stalien.origination}, we deduce that   $t \in \stal(\exec, \tau, v'\sigma)$, and thanks to Lemma~\ref{lem:staliengeneral} (item~\ref{lemitem:tau.in.bot}), we have that $t\in\stal(\exec,v'\sigma)$.
  
    Now, according to Lemma~\ref{lem:trace.stalien}, we know that $\K(\exec_{j-1}) \cup T_0 \vdash w$ for any $w \in \stal(\exec, \K(\exec_{j-1})\cup T_0)$. Since $\exec$ is a valid trace, we have  that $\K(\exec_{j-1})\cup T_0 \vdash v'\sigma$. Applying Lemma~\ref{lem:stalT}, we deduce that $\K(\exec_{j-1}) \cup T_0 \vdash w$ for any $w \in \stal(\exec, v)$. In particular, we have that $\K(\exec_{j-1}) \cup T_0 \vdash t$ and we conclude by relying on our induction hypothesis.\qedhere
  \end{description}

\end{proof}

%%
%% PROOFS on the formula
%%

\section{Proofs of Section~\ref{subsec:sat}}
\label{app:sat}

In order to prove Proposition~\ref{prop:bar.attack} we will annotate formulas.  For the sake of homogeneity, we chose to annotate each term that occurs in the formula even though it would have been sufficient to only annotate variables. Moreover, we state the definition for a general formula, but in our setting, terms that occur in a formula are either names or variables.

\begin{defi}\label{dfn:1} (annotated formula)
Given a formula $\phi$, we define its annotated version $\Annotate(\phi)$ as follows:
\[
\begin{array}{l}
  \Annotate(\true)             =\true \\
  \Annotate(\neg \phi)        = \neg\Annotate(\phi) \\ 
  \Annotate(\learn(t))           =  \learn(t^t) \\
  \Annotate(\Comp(u))            =  \Comp(u^u) 
\end{array}
\;\;\;
\begin{array}{l}
 \Annotate(\Q(t_1, \dots, t_n))  =  \Q(t_1^{t_1}, \dots, t_n^{t_n}) \\
  \Annotate(\phi_1 \vee \phi_2)  =  \Annotate(\phi_1) \vee \Annotate(\phi_2) \\
  \Annotate(\Diamond\phi)        =  \Diamond \Annotate(\phi) \\
  \Annotate(\exists x. \phi)     =  \exists x. \Annotate(\phi)
\end{array}
\]
\end{defi}

We emphasize that those annotations are syntactic decorations that do not interfere in the semantics of the formulas. 
We also suppose that these annotations are not affected by substitutions, i.e., when $x$ is a variable annotated with $a$, $(x^a)\sigma = (x\sigma)^a$.
Relying on this notion of annotated formulas, we are now able to link each variable that occurs in $\phi$ with the term it has been substituted with in order to satisfy the formula.  More precisely, we only need to know the session identifiers from which those terms are issued. The idea is that these sessions are important to satisfy the attack formula whereas the other ones could be discarded from the execution trace.

\begin{defi}\label{dfn:mu}
Let $\phi$ be an attack formula  and $\psi$ its annotated version, i.e. $\psi = \Annotate(\phi)$.
Let $\Pi$ be a protocol, and 
$\exec = [\evt_1^{\sid_1}; \ldots; \evt_\ell^{\sid_\ell}]$ be an execution trace (not necessarily valid) 
of $\Pi$ w.r.t. some initial intruder knowledge $T_0$ and such that $\langle \exec, T_0\rangle \models \psi$. 
Let $\pi$ be a proof tree witnessing the fact that $\langle \exec, T_0\rangle\models \psi$. 
We define $\mu(\pi)$ as described in Figure~\ref{fig:mu}.
\end{defi}

%maps variables occuring positively in a status event in $\psi$ to session identifiers.

%\myrto{Est-ce qu'on veut un lemme qui dit que $\mu$ est une fonction si $\phi$ est une attack formula??}

\begin{figure}[h]
\[
\begin{array}{lcl}
%TRUE
  \mu\left(
    \prooftree \justifies \langle \exec_i, T_0\rangle \models \true\endprooftree
  \right) \; = \; \emptyset & \hspace{1cm}&
% LEARN  
  \mu\left(
    \prooftree\justifies \langle \exec_i, T_0\rangle \models \learn(t^u) \endprooftree
  \right)\; = \; \emptyset \\
  \\
 % COMP 
  \mu\left(
    \prooftree \justifies \langle \exec_i, T_0\rangle \models \Comp(u^v) \endprooftree
  \right)\; = \;  \emptyset 
&&
 % non COMP 
  \mu\left(
    \prooftree \justifies \langle \exec_i, T_0\rangle \models \neg\Comp(u^v)\endprooftree
  \right)
\; = \; \emptyset \\ \\
\multicolumn{3}{l}{
 %non Q(t1,...,tn) 
  \mu\left(
    \prooftree \justifies \langle \exec_i, T_0\rangle \models \neg\Q(t_1^{u_1}, \dots, t_n^{u_n})\endprooftree
  \right) \; = \; \emptyset}
% &&
%  % NON ONCE simplifie parce qu'attack formula
%   \mu\left(
%     \prooftree \pi' \justifies \langle \exec_i, T_0\rangle \models \neg(\Diamond \psi)\endprooftree
%   \right) \; = \; \stefcolor{\emptyset}
\\ \\
%Q(t1,..,tn)
\multicolumn{3}{l}{  \mu\left(
    \prooftree \justifies \langle \exec_i, T_0\rangle \models \Q(t_1^{u_1}, \dots, t_n^{u_n}) \endprooftree
  \right) \; = \; \{(u'_{1}, \sid_i); ..., (u'_{m}; \sid_i)\}} \\[-3mm]
\multicolumn{3}{r}{\hspace{2cm}\text{ where } \vars(\{u_1, \dots, u_n\}) = \{u'_{1}, \dots, u'_{m}\}} \\[2mm]
% OU
 \multicolumn{3}{l}{ \mu\left(
    \prooftree \prooftree \pi'\justifies \langle \exec_i, T_0\rangle \models \psi_j\endprooftree\justifies \langle \exec_i, T_0\rangle \models \psi_1 \vee \psi_2 \endprooftree
  \right) \;  = \; 
  \mu\left(
    \prooftree \pi'\justifies \langle \exec_i, T_0\rangle \models \psi_j\endprooftree
  \right) \text{ with } j\in\{1, 2\}}\\
  \\
 % NON OU 
 \multicolumn{3}{l}{ \mu\left(
    \prooftree \prooftree \pi_1\justifies \langle \exec_i, T_0\rangle \models \neg\psi_1 \endprooftree\prooftree \pi_2\justifies \langle \exec_i, T_0\rangle \models \neg\psi_2 \endprooftree\justifies \langle \exec_i, T_0\rangle \models \neg(\psi_1 \vee \psi_2) \endprooftree
  \right)\; = \;  \bigcup_{j \in \{1,2\}}
  \mu\left(
    \prooftree \pi_j\justifies \langle \exec_i, T_0\rangle \models \neg\psi_j \endprooftree\right) }
 \\
  \\
% ONCE
\multicolumn{3}{l}{
  \mu\left(
    \prooftree \prooftree \pi'\justifies \langle \exec_j, T_0\rangle \models \psi\endprooftree\justifies \langle \exec_i, T_0\rangle \models \Diamond \psi\endprooftree
  \right) \; = \; 
  \mu\left(
\prooftree {\pi'} \justifies 
    \langle \exec_j, T_0\rangle \models \psi \endprooftree
  \right) \text{ where $j \le i$}}\\
  \\
 % NON ONCE 
\multicolumn{3}{l}{
  \mu\left(
    \prooftree \prooftree \pi_1\justifies \langle \exec_1, T_0\rangle \models \neg\psi\endprooftree  \dots \prooftree \pi_i\justifies \langle \exec_i, T_0\rangle \models \neg\psi\endprooftree\justifies \langle \exec_i, T_0\rangle \models \neg(\Diamond \psi)\endprooftree
  \right) \; = \; \underset{i\in\{1, \dots, i\}}{\bigcup} 
  \mu\left(
    \prooftree \pi_i\justifies \langle \exec_i, T_0\rangle \models \neg\psi\endprooftree
  \right)
} \\
  \\
% EXISTENTIEL
\multicolumn{3}{l}{
  \mu\left(
    \prooftree \prooftree \pi' \justifies \langle \exec_i, T_0\rangle \models \psi\{t/x\}\endprooftree\justifies \langle \exec_i, T_0\rangle \models \exists x. \psi\endprooftree
  \right) \;  = \; 
  \mu\left(
    \prooftree \pi' \justifies \langle \exec_i, T_0\rangle \models \psi\{t/x\}\endprooftree
  \right) }\\
  \\
\multicolumn{3}{l}{
  \mu\left(
    \prooftree \prooftree \pi'\justifies \langle \exec_i, T_0\rangle \models \psi\endprooftree\justifies \langle \exec_i, T_0\rangle \models \neg \neg\psi\endprooftree
  \right) \; = \; 
  \mu\left(
    \prooftree \pi'\justifies \langle \exec_i, T_0\rangle \models \psi\endprooftree
  \right)}
\end{array}
\]
  
\caption{Definition of the function $\mu$}
\label{fig:mu}
\end{figure}

Intuitively, $\mu(\pi)$ maps variables occurring positively in a status event in the attack formula $\phi$ to session identifiers.  Note also that since by definition of an attack formula each variable occurs at most once in a positive status event and by Condition~\ref{def:attack-cond4} of Definition~\ref{def:attack}, we have that $\mu(\pi)$ is actually a function.

\noindent\usebox{\probarattack}

\begin{proof}
Let $\exec = [\evt^{\sid_1}_1, \ldots, \evt^{\sid_\ell}_\ell]$ for some $\ell$, and some session identifiers $\sid_1, \ldots, \sid_\ell$.
 By definition of an attack formula, $\phi$ is of the form
  \[
  \phi = \exists x_1. \dots. \exists x_n. \psi
  \]
  for some quantifier-free formula $\psi$. Now, according to the semantics of $\Logic$, 
$\langle \exec, T_0\rangle \models \phi$ implies that there exists $n$ ground terms $m_1$, \dots, $m_n$ 
such that there exists a proof $\pi$ of $\langle \exec, T_0\rangle \models \phi$ of the form:
  \[
  \pi\ =\ 
  \prooftree 
  \prooftree
  \dots
  \justifies
  \langle \exec, T_0\rangle \models \psi^a\sigma
  \endprooftree
  \justifies 
  \langle \exec, T_0\rangle \models \phi^a
  \endprooftree
  \]
  where $\sigma = \{x_1 \mapsto m_1, \dots, x_n \mapsto m_n\}$ and $\phi^a = \exists x_1. \dots. \exists x_n. \psi^a = \Annotate(\phi)$.
Let $\overline{\sigma}=  \{x_1 \mapsto \overline{m_1}^{\exec, \sid'_1}, \dots, x_n \mapsto \overline{m_n}^{\exec,\sid'_n}\}$ where
$\sid'_j =  \mu(\pi)(x_j)$ when $x_j \in \dom(\mu(\pi))$ and $0$ otherwise.

Note that {all except the last two nodes of} $\pi$ are labeled with $\langle \exec_i, T_0 \rangle \models \psi'\sigma$ where $i \leq \length(\exec)$ and~$\psi'$ is smaller than $\psi$. Thus, the proof tree is finite.
Moreover, by definition of $\mu$, we have that any leaf of $\pi$ of the form $\langle \exec_i, T_0 \rangle \models \Q(u_1,\ldots, u_k)\sigma$ is such that $\mu(\pi)(x) =\sid_i$ for any $x \in \vars(\{u_1,\ldots,u_k\})$.
We prove that the proof tree obtained from $\pi$ by replacing each node labeled with 
 $\langle \exec_i, T_0 \rangle \models \psi'\sigma$ by  $\langle \exec_i,T_0\rangle \models \psi'\overline{\sigma}$ is a (valid) proof tree witnessing the fact that $\langle \overline{\exec},T_0 \rangle \models \psi^a\overline{\sigma}$.

 \smallskip{}

\noindent  \emph{Base cases: the leaves of the proof tree $\pi$.} 
In such a case, we 
have $\langle \exec_i, T_0 \rangle \models \psi_0\sigma$ for a formula $\psi_0$ of the form $\true$, $\Comp(x)$, $\neg \Comp(x)$, $\learn(u_0)$, $\Q(u_1,\ldots,u_k)$, or $\neg Q(u_1,\ldots, u_k)$.

\begin{itemize}
\item $\psi_0 = \true$: in such a case, we easily conclude.
\item $\psi_0 = \Comp(x)$ (resp. $\neg \Comp(x)$): in such a case, we have that $\Comp(x\sigma) = \Comp(x\overline{\sigma})$ since $\overline{a}^{\exec,\sid} = a$ for any agent name $a$ and any $\sid$, and since the semantics of $\Comp$ does not rely on the execution trace, we can also easily conclude in this case.

\item $\psi_0 = \learn(u_0)$: in such a case, by definition of an attack formula, we know that $u_0$ is either an agent name (in such a case, we easily conclude) 
or a variable in $\{x_1,\ldots, x_n\}$. Let $j$ be such that $u_0 = x_j$.
By hypothesis, we have that $T_0 \cup \K(\exec_i) \vdash u_0\sigma$.
% and we have to show that $T_0 \cup \K(\overline{\exec_i}) \vdash u_0\overline{\sigma}$.
According to Lemma~\ref{lem:deducibility.of.abstracted.terms}, we know that $T_0 \cup \K(\overline{\exec}_i) \vdash \overline{u_0\sigma}^{\exec,\sid}$ for any $\sid$, and thus in particular for $\sid'_j$.
Actually, we have that $\overline{x_j\sigma}^{\exec,\sid'_j} = x_j\overline{\sigma} (= \overline{m_j}^{\exec,\sid'_j})$, and this allows us to conclude that 
$\langle \overline{\exec}_i, T_0 \rangle \models \learn(u_0)\overline{\sigma}$.

\item $\psi_0 = \Q(u_1,\ldots, u_k)$: in such a case, we know that each $u_j$ is either an agent name or a variable, and we have that
$u_j\sigma = t_j$ for any $j \in \{1,\ldots,k\}$ where
$\evt^{\sid_i}_i = Q(t_1,\ldots,t_k)$.  By definition of $\mu$, we have that either $u_j$ is an agent name or $u_j$ is a variable and $\mu(\pi)(u_j) = \sid_i$. In order to conclude that $\langle \overline{\exec}_i, T_0 \rangle \models \Q(u_1, \ldots, u_k){\overline\sigma}$, we have to show that $\overline{t_j}^{\exec,\sid_i} = u_j\overline{\sigma}$.  Let $j \in \{1,\ldots,k\}$. By hypothesis, we have that $u_j\sigma = t_j$, and thus $\overline{u_j\sigma}^{\exec,\sid_i} = \overline{t_j}^{\exec,\sid_i}$.  We distinguish two cases. Either $u_j$ is an agent name, and we have that $\overline{u_j\sigma}^{\exec,\sid_i} = u_j = u_j\overline{\sigma}$. Otherwise, $u_j$ is a variable, and we also have that $\overline{u_j\sigma}^{\exec,\sid_i} = u_j\overline{\sigma}$ since by definition of $\mu$, we have that $\mu(\pi)(u_j) = \sid_i$.

\item $\psi_0 = \neg \Q(u_1,\ldots,u_k)$: in such a case, we know that each $u_j$ is either an agent name or a variable, and we have that either $\exec_i = []$ or $\Q(u_1,\ldots,u_n)\sigma \neq \evt^{\sid_i}_i$. In the first case, we have that $\overline{\exec}_i = []$ and we easily conclude. From now on, assume that $\Q(u_1,\ldots,u_n)\sigma \neq \evt^{\sid_i}_i$. If $\evt^{\sid_i}_i \neq \Q(t_1,\ldots,t_k)$ for any terms $t_1, \ldots, t_k$, then it is easy to see that
 $\overline{\evt^{\sid_i}_i}^{\exec,\sid_i} \neq \Q(u_1,\ldots,u_k)\overline{\sigma}$ and this allows us to conclude. 
Now, assume that $\evt^{\sid_i}_i = \Q(t_1,\ldots,t_k)$ for some terms $t_1,\ldots,t_k$. In such a case, there exists $j \in \{1, \ldots, k\}$ such that $u_j\sigma \neq t_j$.
Using Lemma~\ref{lem:mu}, we deduce that $\overline{u_j\sigma}^{\exec,\mu(u_j)} \neq \overline{t_j}^{\exec,\sid_i}$, and by definition of $\mu$  we have that
$u_j\overline{\sigma} = \overline{u_j\sigma}^{\exec,\mu(u_j)}$. This allows us to conclude that 
$\Q(u_1,\ldots, u_k)\overline{\sigma} \neq \overline{\Q(t_1,\ldots,t_k)}^{\exec,\sid_i}$, and thus $\langle \overline{\exec}_i,T_0 \rangle \models \psi_0\overline{\sigma}$.
\end{itemize}

\noindent \emph{Inductive cases.}
In such a case, we have that $\langle \exec_i, T_0 \rangle \models \psi_0\sigma$ for a formula $\psi_0$ of the form $\neg \neg \psi'_0$, $\psi_1 \vee \psi_2$, $\neg (\psi_1 \vee \psi_2)$, $\Diamond \psi'_0$, or $\neg \Diamond \psi'_0$.

\begin{itemize}
\item $\psi_0 = \neg \neg \psi'_0$: in such a case, we have that $\langle \exec_i,T_0 \rangle \models \psi'_0\sigma$, and using our induction hypothesis we conclude that 
$\langle \overline{\exec}_i,T_0 \rangle \models \psi'_0\overline{\sigma}$, and thus $\langle \overline{\exec}_i,T_0 \rangle \models \neg \neg \psi'_0\overline{\sigma} = \psi_0\overline{\sigma}$.

\item $\psi_0 = \psi_1 \vee \psi_2$: in such a case, we have that $\langle \exec_i,T_0 \rangle \models \psi_j\sigma$ for some $j \in \{1,2\}$, and using our induction hypothesis we conclude that $\langle \overline{\exec}_i,T_0 \rangle \models \psi_j\overline{\sigma}$, and thus $\langle \overline{\exec}_i,T_0 \rangle \models (\psi_1 \vee \psi_2)\overline{\sigma} = \psi_0\overline{\sigma}$.

\item $\psi_0 = \neg (\psi_1 \vee \psi_2)$:  in such a case, we have that $\langle \exec_i,T_0 \rangle \models \neg \psi'_j\sigma$ with $j \in \{1,2\}$, and using our induction hypothesis we conclude that $\langle \overline{\exec}_i,T_0 \rangle \models \neg \psi_j\overline{\sigma}$ with $j \in \{1,2\}$, and thus $\langle \overline{\exec}_i,T_0 \rangle \models \neg (\psi_1 \vee \psi_2)\overline{\sigma} = \psi_0\overline{\sigma}$.

\item $\psi_0 = \Diamond \psi'_0$: in such a case, we have that $\langle \exec_j,T_0 \rangle \models \psi'_0\sigma$ for some $j \le i$, and using our induction hypothesis, we conclude that $\langle \overline{\exec_j},T_0 \rangle \models \psi'_0\overline{\sigma}$, and thus   $\langle \overline{\exec}_i,T_0 \rangle \models \Diamond \psi'_0\overline{\sigma} = \psi_0\overline{\sigma}$.

\item $\psi_0 = \neg \Diamond \psi'_0$: in such a case, we have that  $\langle \exec_j,T_0 \rangle \models \neg \psi'_0\sigma$ for any $j \in \{1, \ldots, j\}$, and using our induction hypothesis, we conclude that $\langle \overline{\exec_j},T_0 \rangle \models \neg \psi'_0\overline{\sigma}$,  and thus   $\langle \overline{\exec}_i,T_0 \rangle \models \neg \Diamond \psi'_0\overline{\sigma} = \psi_0\overline{\sigma}$.\qedhere
\end{itemize}
\end{proof}

\section{Proofs of Section~\ref{sec:2-step}}

This appendix contains the proofs of Section~\ref{sec:2-step}.  Actually,
Section~\ref{subsec:app-validity-restriction} contains the proofs related to the
validity of the resulting trace $\exec|_S$ whereas
Section~\ref{subsec:app-sat-formula} contains those related to the
satisfiability of the attack formula.

\subsection{Validity of the resulting trace}
\label{subsec:app-validity-restriction}

In order to preserve the validity of the resulting trace, it is important to
show that sessions that are not tagged in the same way cannot share any
name. This is the purpose of the following lemma.

% S.D.: OK relu
\noindent\usebox{\lemdisjointness}

\begin{proof}
  Let $\sess_1$ and $\sess_2$ be two sessions and $n$ be a name such that:
  \begin{itemize}
  \item $\sameTagas(\exec, \sess_1) \not= \sameTagas(\exec, \sess_2)$; and
  \item $n\in \names(\exec, sess_1) \cap \names(\exec, sess_2)$.
  \end{itemize}
  Let $S = \sameTagas(\exec, \sess_1)$.  According to
  Condition~\ref{def:wfexec3} of well-formedness (Definition~\ref{def:wfexec}),
  $n\in\names(\exec, sess_1)$ implies that either $n$ is of the form
  $n^{\epsilon, S}_t$ or of the form $n^{\sid}_t$ for some term $t$ and session
  identifier $\sid\in S$. We treat these two cases separately:
  \begin{description}
  \item[Case $n = n^{\epsilon,S}_t$] According to Condition~\ref{def:wfexec3} of
    well-formedness (Definition~\ref{def:wfexec}), we obtain $n^{\epsilon,S}_t
    \in\names(\exec, \sess_2)$ implies that $S = \sameTagas(\exec,
    \sess_2)$. But this contradicts the hypothesis $\sameTagas(\exec, \sess_1)
    \not= \sameTagas(\exec, \sess_2)$.

  \item [Case $n = n^{\sid}_t$] In that case, $\sid\in S$ and $\sameTagas(\exec,
    \sid) = \sameTagas(\exec, \sess_1)$. Now, according to
    Condition~\ref{def:wfexec3} of well-formedness
    (Definition~\ref{def:wfexec}), we have that $n^{\sid}_t \in\names(\exec,
    \sess_2)$ implies that $\sid \in \sameTagas(\exec, \sess_2)$. However, this
    means that $\sameTagas(\exec, \sid) = \sameTagas(\exec, \sess_2)$ which
    contradicts our hypothesis.
  \end{description}
  By contradiction we conclude that $\names(\exec, \sess_1) \cap \names(\exec,
  \sess_2) = \emptyset$.\qedhere
\end{proof}

Now, provided that $S$ and $t$ satisfy some conditions, we show that a term $t$
that was deducible from $\exec$ will still be deducible from $\exec|_S$.

% S.D.: OK relu
\noindent\usebox{\lemprojecteddeducibility}

\begin{proof}
  Let $\sid \in S$, $t\in \st(\exec,\sid)$, and $\pi$ be a simple proof of
  $T_0\cup \K(\exec)\vdash t$.
  % S.D.: une notion de minimalite en nombre de noeud devrait suffire mais pour
  % appeler le lemme 3, il faut mieux utiliser la notion de preuve simple.
  We prove this result by structural induction on~$\pi$. But, we first need to
  establish the following preliminary result {(still under the
  hypotheses stated in Lemma~\ref{lem:projected.deducibility})}.

  \smallskip{}

  \noindent{\bf Claim.}
  If $\names(t) \subseteq \Nces_\epsilon$ then $T_0\vdash t$.  \smallskip{}

  \noindent \emph{Proof of the claim.} Let us suppose that there exists $u \in
  \encst(t)$. Because $\exec$ is well-formed, we know by
  Conditions~\ref{def:wfexec1} and~\ref{def:wfexec2} of well-formedness
  (Definition~\ref{def:wfexec}) that $t$ is $k$-tagged and thus that $u =
  \f(\langle \tau, u_1 \rangle, \dots, u_n)$ with $\tau = \tagt(\exec, \sid)
  \not=\bot$. Now, according to the definition of a symbolic trace
  (Definition~\ref{def:scenario}) and of our protocol transformation
  (Definition~\ref{def:transfo}), we know that there exists
  $n_v^{\sid}\in\names(\tau)\subseteq\names(u)\subseteq \names(t)$, which
  contradicts the hypothesis that $\names(t) \subseteq \Nces_\epsilon$. Thus it
  must be that $\encst(t) = \emptyset$, and hence, $t$ must be a tuple of atoms,
  i.e. a tuple of terms in $\A \cup T_0 \cup \Nces_\epsilon \cup \Kcal_\epsilon
  \cup \{\pub(a) \mid a\in \A\} \cup \{\priv(a), \shk(a, b) \mid a, b\in
  \A\}$. Now, because we only consider executions that do not reveal any
  long-term decryption keys, we necessarily have that the atomic subterms of $t$
  are in $\A \cup T_0 \cup \Nces_\epsilon \cup \Kcal_\epsilon \cup \{\pub(a)
  \mid a\in \A\}$. This implies according to Definition~\ref{def:deducibility},
  that any atomic subterm of $t$ is deducible from $T_0$. Finally, since $t$ is
  a tuple of deducible terms, $t$ can be deduced by application of the pairing
  rule, and thus $T_0 \vdash t$.

  \smallskip{}

  We now proceed with our induction
  \begin{description}
  \item[Base case: $\pi$ is reduced to a leaf] In that case, $t\in \A \cup T_0
    \cup \Nces_\epsilon \cup \Kcal_\epsilon \cup \{\pub(a) \mid a\in \A\} \cup
    \K(\exec)$. If $\names(t) \subseteq \Nces_\epsilon$, then by the above claim
    we have that $T_0\vdash t$, and thus $T_0 \cup \K(\exec|_S)\vdash t$. Let us
    now suppose that there exists $n^{\sid'}_v \not \in \Nces_\epsilon$. In that
    case, $n^{\sid'}_v \in \names(\exec, \sid)$ and $t\in\K(\exec)$, i.e.  there
    exists $i \in \{1, \dots, \ell\}$, such that $\evt_i^{\sid_i} =
    \snd(t)$. Thus, $n^{\sid'}_v \in \names(\exec, \sid_i)$ and hence
    $\names(\exec, \sid) \cap \names(\exec, \sid_i) \not=\emptyset$. This,
    according to Lemma~\ref{lem:disjointness}, implies that
    \[
    \sameTagas(\exec, sid_i)\ =\ \sameTagas(\exec, sid)
    \]
    By hypothesis on $S$, we have $sid_i\in S$, and by definition we have
    $\evt^{\sid_i}_i\in \exec|_S$, which implies that $t\in \K(\exec|_S)$. We
    can thus conclude that $T_0 \cup \K(\exec|_S) \vdash t$. \\

  \item[Inductive case] In that case we need to distinguish two cases according
    to the last rule applied in the proof $\pi$.

    \begin{description}
      % composition
    \item[Case 1 -- the last rule is a composition rule] We have that the term
      $t$ is of the form $\f(t_1, \dots, t_n)$, and the derivation $T_0 \cup
      \K(\exec) \vdash t$ is of the form
      \[
      \prooftree 
      T_0\cup \K(\exec) \vdash t_1 \quad \dots \quad 
      T_0\cup \K(\exec) \vdash t_n \justifies T_0\cup \K(\exec) \vdash \f(t_1, \dots, t_n)
      \endprooftree
      \]
      For all $i\in \{1, \dots, n\}$, $t_i\in\st(t)\subseteq \st(\exec, \sid)$,
      and by induction hypothesis $T_0 \cup \K(\exec|_S) \vdash t_i$. We can
      thus conclude that by application of the corresponding composition
      rule. We have that:
      \[
      \prooftree 
      T_0\cup \K(\exec|_S) \vdash t_1 \quad \dots \quad 
      T_0\cup \K(\exec|_S) \vdash t_n \justifies T_0\cup \K(\exec|_S) \vdash f(t_1, \dots, t_n)
      \endprooftree
      \]
      % decomposition
    \item[Case 2 -- the last rule is a decomposition rule] We have that the
      proof tree witnessing $T_0 \cup \K(\exec) \vdash t$ is of the form
      \[
      \prooftree 
      T_0\cup \K(\exec) \vdash t_1 \quad \dots \quad 
      T_0\cup \K(\exec) \vdash t_n \justifies T_0\cup \K(\exec) \vdash t
      \endprooftree
      \]
      If $\names(t) \subseteq \Nces_\epsilon$, then we have seen that $T_0\vdash
      t$, and thus we conclude.  Now, assume that there exists $n^{sid'}_v \in
      (\names(t)\smallsetminus \Nces_\epsilon) \subseteq \names(\exec,
      \sid)$. By Definition of a symbolic trace and of an execution trace (see
      Definition~\ref{def:scenario}), $n^{\sid'}_v \in \names(\exec,
      \sid')$. Thus $\names(\exec, \sid) \cap \names(\exec, \sid') \not=
      \emptyset$, and thanks to Lemma~\ref{lem:disjointness}, we have that:
      $\sameTagas(\exec, \sid)\ =\ \sameTagas(\exec, \sid')$.

      \smallskip{}

      We need to prove that for all $i\in\{1, \dots, n\}$, $T_0 \cup
      \K(\exec|_S) \vdash t_i$.  Since $\pi$ is minimal, we know by locality
      (Lemma~\ref{lem:locality}) that $t_i \in \st(T_0 \cup \Nces_\epsilon \cup
      \Kcal_\epsilon \cup \K(\exec)) \cup \A \cup \{\pub(a)~|~a \in \A\}$. We
      consider two cases:

      If $\names(t_i) \subseteq \Nces_\epsilon$, then we have already
      established that $T_0\vdash t_i$, and thus $T_0 \cup \K(\exec|_S)\vdash
      t_i$.

      Otherwise, there exists $n^{sid''}_w \in (\names(t_i)\smallsetminus
      \Nces_\epsilon)$. In that case, $t_i \in\st(\K(\exec))$, i.e. there exists
      $k\in\{1, \dots, \ell\}$ such that $t_i\in \st(\evt_k^{\sid_k}) \subseteq
      \st(\exec, \sid_k)$; and thus $n^{\sid''}_w \in \names(\exec,
      \sid_k)$. Moreover, by Definition of a symbolic trace and of an execution
      trace (see Definition~\ref{def:scenario}), $n^{\sid''}_w \in \names(\exec,
      \sid'')$. Hence, we have that $\names(\exec, \sid'') \cap \names(\exec,
      \sid_k) \not= \emptyset$, which according to Lemma~\ref{lem:disjointness}
      implies that
      \[
      \sameTagas(\exec, \sid'')\ =\ \sameTagas(\exec, \sid_k).
      \]
      By inspection of the decomposition rules, we note that there must exist
      $j\in\{1, \dots, n\}$, such that for all $i\in\{1, \dots, n\}$, $\names(t)
      \cup \names(t_i) \subseteq \names(t_j)$, and therefore $n^{\sid'}_v,
      n^{\sid''}_w \in (\names(t_j) \smallsetminus \Nces_\epsilon)$. Moreover,
      we have that $t_j\in \st(\K(\exec))$, i.e. there exists $h\in \{1, \dots,
      \ell\}$ such that $t_j \in \st(\evt_j^{\sid_h}) \subseteq \st(\exec,
      \sid_h)$. Hence, $n^{\sid'}_v, n^{\sid''}_w \in \names(\exec, \sid_h)$,
      which according to Lemma~\ref{lem:disjointness} implies
      \[
      \begin{array}{c}
        \sameTagas(\exec, \sid')\ =\ \sameTagas(\exec, \sid_h) \\
        \sameTagas(\exec, \sid'')\ =\ \sameTagas(\exec, \sid_h)
      \end{array}
      \]
      We therefore can infer that
      \[
      \sameTagas(\exec, \sid)\ =\ \sameTagas(\exec, \sid_k).
      \]
      and by hypothesis on $S$ that $\sid_k\in S$. We have thus demonstrated
      that $t_i\in \st(\exec, sid_k)$ with $\sid_k \in S$, which according to
      our induction hypothesis implies $T_0 \cup \K(\exec|_S) \vdash t_i$.

      Since for all $i \in \{1, \dots, n\}$, $T_0 \cup \K(\exec|_S) \vdash t_i$,
      we can conclude by application of the corresponding decomposition rule
      that:

      $
      \hspace{2cm}
      \begin{array}[b]{c}
        \prooftree 
        T_0\cup \K(\exec|_S) \vdash t_1 \quad \dots \quad 
        T_0\cup \K(\exec|_S) \vdash t_n
        \justifies 
        T_0\cup \K(\exec|_S) \vdash t
        \endprooftree
      \end{array}
      $\qedhere
    \end{description}
  \end{description}
\end{proof}

%%
%% SATISFIABILITY OF THE FORMULA
%%

\subsection{Satisfiability of the formula}
\label{subsec:app-sat-formula}

\newsavebox{\lemprojsatisfiability} 
\sbox{\lemprojsatisfiability}{\vbox{%
    \begin{lem}\label{lem:projsatisfiability}
      Let $\Pi$ be a $k$-party protocol, $\phi$ a closed quantifier-free formula
      in $\Logic$, and $\exec = [\evt_1^{\sid_1}; \dots; \evt_\ell^{\sid_\ell}]$
      be a well-formed valid execution of $\widetilde \Pi$ that satisfies
      $\phi$, w.r.t. some set $T_0$ of ground atoms. Moreover, we assume
      that $T_0 \cup \K(\exec) \not\vdash k$ for any $k \in \lgAtom
      \smallsetminus (\Kcal_\epsilon \cup T_0)$ ($\exec$ does not reveal any
      long term keys).  Let $S$ be a set of session identifiers such that:
      \begin{enumerate}
      \item for all $\learn(t)$ that occurs positively in $\phi$ such that
        $t\not\in \A \cup \lgAtom$, there exists $\sid \in S$ such that $t\in
        \st(\exec, \sid)$,
      \item $\W(\exec, \phi) \subseteq S$, and
      \item $\forall \sess_1, \sess_2$ with $\tagt(\exec, \sess_1) =
        \tagt(\exec, \sess_2)$, we have that
        \begin{center}
          $\sess_1\in S$ if and only if $\sess_2\in S$.
        \end{center}
  \end{enumerate}
  We have that $\exec|_S$ is an execution of $\widetilde \Pi$ that satisfies
  $\phi$, i.e.  $\langle \exec|_S, T_0\rangle \models \phi$.
\end{lem}
}}

\noindent\usebox{\lemprojsatisfiability}

\begin{proof}
  We prove this by induction on {$(\ell,\size(\phi))$} using the lexicographic ordering. Here,
  {$\ell$} denotes the length (i.e. number of events) of the trace $\exec$ and {$\size(\phi)$} is
  the size of $\phi$ (i.e. number of symbols that occur in $\phi$ without
  counting the symbol $\neg$ and after {elimination} of double negation, i.e.,
  $\neg \neg \psi$ is rewritten in $\psi$).

  \smallskip{}
  
  % \begin{description}
  % \item[Base cases]
  \noindent We need to distinguish several base cases.
  \begin{description}
    % taille de l'execution a zero
  \item[Case $|\exec| = 0$] In that case $\exec|_S = \exec$, and thus by
    hypothesis if $\langle \exec, T_0 \rangle \models \phi$, then also $\langle
    \exec|_S, T_0 \rangle \models \phi$.
    % formule true / false
  \item[Case $\phi = \true$ (resp. $\phi = \neg \true$)] In such a case, we have
    that $\langle \exec|_S, T_0 \rangle \models \phi$. The case where $\phi =
    \neg \true$ is impossible.
    % predicat positif
  \item[Case $\phi = \Q(t_1, \dots, t_n)$] If $\langle \exec, T_0 \rangle
    \models \phi$, then $\evt^{\sid_\ell}_\ell = \Q(t_1, \dots, t_n)$, and
    $\W(\exec, \phi) = \{\sid_\ell\} \subseteq S$. By
    Definition~\ref{def:exec.restriction}, $\exec|_S$ ends with the event
    $\Q(t_1,\ldots, t_n)$. We can thus conclude that $\langle \exec|_S, T_0
    \rangle \models \phi$.
    % predicat negatif
  \item[Case $\phi = \neg\Q(t_1,\ldots,t_n)$] If $\langle \exec, T_0 \rangle
    \models \neg \Q(t_1,\ldots,t_n)$, we have that $\evt^{\sid_\ell}_\ell
    \not= \Q(t_1, \dots, t_n)$, and $\W(\exec, \phi) = \{\sid_\ell\} \subseteq
    S$ (note that we have already considered the case where $\exec = []$, and
    thus now we assume that $\exec \neq []$).  We have that $\exec|_S$ does not
    end with $\Q(t_1, \dots, t_n)$. We can thus conclude that $\langle \exec|_S,
    T_0 \rangle \models \neg\Q(t_1,\ldots,t_n)$, i.e.  $\langle \exec|_S, T_0
    \rangle \models \phi$.
    % learn positif
  \item[Case $\phi = \learn(t)$] If $\langle \exec, T_0 \rangle \models \phi$,
    then $T_0 \cup \K(\exec) \vdash t$. If $t\in \A \cup \lgAtom$, since $\exec$
    doesn't reveal any long-term decryption key, $T_0 \vdash t$, and thus $T_0
    \cup \K(\exec|_S) \vdash t$. If $t\not\in \A \cup \lgAtom$, then by
    hypothesis we know there exists $\sid\in S$ such that $t\in \st(\exec,
    \sid)$. According to Lemma~\ref{lem:projected.deducibility}, since by
    hypothesis $\sid\in S$, $T_0 \cup \K(\exec|_S) \vdash t$. Hence, we can
    conclude that $\langle \exec|_S, T_0 \rangle \models \phi$.
    % learn negatif
  \item[Case $\phi = \neg \learn(t)$] If $\langle \exec, T_0 \rangle \models
    \neg \learn(t)$, then $T_0 \cup \K(\exec) \not\vdash t$. But since $T_0 \cup
    \K(\exec|_S) \subseteq T_0 \cup \K(\exec)$, it is also the case that $T_0
    \cup \K(\exec|_S) \not\vdash t$, and thus that $\langle \exec|_S, T_0
    \rangle \models \neg\learn(t)$.
    % Compromised positif
  \item[Case $\phi = \Comp(u)$] If $\langle \exec, T_0 \rangle
      \models \Comp(u)$, then we have that $T_0 \vdash \priv(u)$ or $T_0 \vdash
      \shk(u,v)$ for some $v \neq \epsilon$. Hence, we also have that $\langle
      \exec|_S, T_0 \rangle \models \Comp(t)$.

    %   Compromised negatif
  \item[Case $\phi = \neg\Comp(u)$] If $\langle \exec, T_0 \rangle
      \models \neg\Comp(u)$, then we have that $T_0 \not \vdash \priv(u)$ and
      $T_0 \not \vdash \shk(u,v)$ for all $v \neq \epsilon$. Hence, we also have
      that $\langle \exec|_S, T_0 \rangle \models \neg\Comp(t)$.
      % But since $T_0 \cup \K(\exec|_S) \subseteq T_0
      % \cup \K(\exec)$, it is also the case that $T_0 \cup \K(\exec|_S)
      % \not\vdash k$ for all $k \in \{\priv(t), \shk(t, t') \mid t'\in \A\ \wedge\
      % t'\not= \epsilon\}$, and thus that $\langle \exec|_S, T_0 \rangle
      % \models \neg\Comp(t)$.
  \end{description}

  \smallskip{}

  % \item[Inductive case ($|\exec| > 1$ and $\size(\phi)>1$)]
  \noindent We distinguish several inductive cases ($|\exec| > 1$ and
  {$\size(\phi)>1$}).

  \begin{description}
    % ou positif
  \item[Case $\phi = \phi_1 \vee \phi_2$] If $\langle \exec, T_0 \rangle \models
    \phi$ then $\langle \exec, T_0 \rangle \models \phi_1$ or else $\langle \exec,
    T_0 \rangle \models \phi_2$. Assume that $\langle \exec, T_0 \rangle \models
    \phi_1$ (the other case can be done in a similar way).  It is easy to see
    that the three conditions needed to apply our inductive hypothesis are
    fulfilled.  We can thus apply our inductive hypothesis on $\phi_1$ to
    conclude that $\langle \exec|_S, T_0 \rangle \models \phi_1$, and thus that
    $\langle \exec|_S, T_0 \rangle \models \phi_1 \vee \phi_2$.
    % ou negatif
  \item[Case $\phi = \neg(\phi_1 \vee \phi_2)$] If $\langle \exec, T_0 \rangle
    \models \neg(\phi_1 \vee \phi_2)$, then $\langle \exec, T_0 \rangle \models
    \neg \phi_1$ and $\langle \exec, T_0 \rangle \models \neg\phi_2$. Again, the
    three conditions needed to apply our inductive hypothesis are full-filled.
    We can thus apply our inductive hypothesis to conclude that $\langle
    \exec|_S, T_0 \rangle \models \neg \phi_1$ and $\langle \exec|_S, T_0
    \rangle \models \neg \phi_2$, and thus $\langle \exec|_S, T_0 \rangle
    \models \neg(\phi_1 \vee \phi_2)$.
    % Diamond positif
  \item[Case $\phi = \Diamond \psi$] If $\langle \exec, T_0 \rangle \models
    \phi$, then we know that there exists $i\in \{1, \dots, \ell\}$ such that
    ${\langle \exec_i, T_0 \rangle \models \psi}$ and $\W(\exec, \phi) =
    \W(\exec_i, \psi)$.
    \begin{itemize}
    \item Let $\learn(t)$ be a subformula that occurs positively in $\psi$ such
      that $t\not\in \A\cup \lgAtom$. Then, by definition, $\learn(t)$ also
      occurs positively in $\phi$, and thus by hypothesis, there exists $\sid
      \in S$ such that $t\in\st(\exec, \sid)$.
    \item We have that $\W(\exec, \phi) = \W(\exec_i, \psi)$, and by hypothesis
      $\W(\exec, \phi) \subseteq S$. Thus $\W(\exec_i, \psi) \subseteq S$.

    \item By hypothesis, $S$ satisfies: for all $sess_1$ and
      $sess_2$ with $\tagt(\exec, sess_1) = \tagt(\exec, sess_2)$,
      $sess_1\in S$ if and only if $sess_2\in S$.
    \end{itemize}
    The three conditions are fulfilled, we can thus apply our inductive
    hypothesis to conclude that $\exec_i|_S$ also satisfies $\psi$,
    i.e. $\langle \exec_i|_S, T_0 \rangle \models \psi$. But then there exists
    $j$ such that $\exec_i|_S = (\exec|_S)_j$, and thus such that $\langle
    (\exec|_S)_j, T_0 \rangle \models \psi$, which according to the semantics of
    $\Logic$ gives us $\exec|_S$ satisfies $\Diamond \psi$, i.e.  $\langle
    \exec|_S, T_0 \rangle \models \Diamond \psi$.
      % Diamond negatif
  \item[Case $\phi = \neg \Diamond \psi$] If $\langle \exec, T_0 \rangle \models
    \neg \Diamond \psi$, then according to the semantics of $\Logic$, we have
    that $\langle \exec_{\ell-1}, T_0 \rangle \models \neg\Diamond \psi$ and
    $\langle \exec, T_0 \rangle \models \neg\psi$.
    \begin{itemize}
    \item In the syntax of $\Logic$, {see Definition~\ref{def:syntax}}, $\learn(t)$ must not occur under a modality,
      so the first condition is trivially fulfilled.
    \item By definition, $\W(\exec, \phi) = \emptyset \subseteq S$.
    \item By hypothesis, $S$ satisfies: for all $sess_1$ and
      $sess_2$ with $\tagt(\exec, sess_1) = \tagt(\exec, sess_2)$,
      $sess_1\in S$ if and only if $sess_2\in S$.
    \end{itemize}
    We apply our inductive hypothesis and conclude that $\langle
    (\exec_{\ell-1})|_S, T_0 \rangle \models \neg\Diamond \psi$.  Now, we
    distinguish two cases: either $\sid_\ell \in S$ or $\sid_\ell \not\in S$. In
    the first case, we can also apply our inductive hypothesis on $\langle
    \exec, T_0 \rangle \models \neg\psi$ (note that $\W(\exec,\psi) \subseteq
    \{\sid_\ell\} \subseteq S$ {since $\psi$ is from the
      restricted syntax according to Definition~\ref{def:syntax}}) and conclude that $\langle \exec|_S, T_0 \rangle
    \models \neg\psi$. This allows us to conclude that $\langle \exec|_S, T_0
    \rangle \models \neg\Diamond\psi$. In the second case, we have that
    $\exec|_S = \exec_{\ell -1}|_S$, and thus conclude that $\langle
    \exec|_S, T_0 \rangle \models \neg \Diamond \psi$.\qedhere
  \end{description}
    % \end{description}
\end{proof}

%\newsavebox{\lemexectermsattack} 
%\sbox{\lemexectermsattack}{
%  \vbox{%
    \begin{lem}
      \label{lem:execterms.attack}
      Let $\Pi$ be a $k$-party protocol, $\exec = [\evt_1^{\sid_1}; \dots;
      \evt_\ell^{\sid_\ell}]$ be a valid execution of $\widetilde\Pi$
      w.r.t. some set $T_0$ of ground atoms, $\phi = \exists
      x_1. \ldots. \exists x_n. \psi$ be an attack formula of $\Logic$ (see
      Definition~\ref{def:attack}), $\sigma = \{x_1 \mapsto m_1, \dots,
      x_n\mapsto m_n\}$ be a ground substitution, $S$ be a set of session
      identifiers such that $\W(\exec, \psi\sigma) \subseteq S$, and
      $n^\epsilon_\epsilon\in \Nces_\epsilon$ be an intruder nonce not appearing
      in $\exec$.
      % If some status events occur positively in $\psi\sigma$, we also assume
      % that $\W(\exec_i, \psi\sigma) \subseteq S$.
      If $\langle \exec, T_0\rangle \models \psi\sigma$ then we have that
      $\langle \exec, T_0\rangle \models \psi\sigma'$ where for all $j\in\{1,
      \dots, n\}$
      \[
      \sigma'(x_j) = \left\{
        \begin{array}{lr}
          n^\epsilon_\epsilon & $if $ \sigma(x_j)\not\in\st(\exec, S)\cup\A\cup\lgAtom \cup \Nces_\epsilon\cup \Kcal_\epsilon \\
          \\
          \sigma(x_j)      & $otherwise$
        \end{array}
      \right.
      \]
    \end{lem}%}}

%\noindent\usebox{\lemexectermsattack}

\begin{proof}
  We prove this result by induction on {$(\ell,\size(\psi))$} using the lexicographic
  ordering where {$\ell$} denotes the length of the trace $\exec$, and $\size(\psi)$ the size
  of $\psi$ (i.e. number of symbols that occur in $\psi$ without counting the
  symbol $\neg$ and after elimination of double negation, i.e., $\neg \neg
  \psi$ is rewritten in $\psi$). Actually, we {strengthen} the induction hypothesis by
  only requiring the hypothesis $\W(\exec_p, \psi\sigma) \subseteq S$ when some
  status event occurs positively in $\psi\sigma$.\smallskip

%\item[Base case {($\size(\psi) = 1$)}]
\noindent Base case {($\size(\psi) = 1$)}: We distinguish several base cases.

%  \smallskip{}

%  \noindent We first distinguish several base cases.
  \begin{itemize}
    % formula true
  \item \emph{Case $\psi'= \true$.} In that case $\psi\sigma = \psi\sigma' =
    \true$ and we easily conclude.
    % formula neg true
  \item \emph{Case $\psi'= \neg \true$.} This case is impossible since such a
    formula is not satisfiable.
    % formula predicat
    \item \emph{Case $\psi = \Q(t_1, \dots, t_h)$.}  In that case, we have
      that 
      $\psi\sigma = \Q(t_1\sigma, \dots, t_h\sigma)$, and $\evt_p^{\sid_p} =
      \Q(t_1\sigma, \dots, t_h\sigma)$ with $\sid_p\in S$. But then,
      by Definition of $\sigma'$, we have that  $\evt_p^{\sid_p} =
      \Q(t_1\sigma, \dots, t_h\sigma) = \Q(t_1\sigma', \dots,
      t_h\sigma')$, and we conclude that $\langle \exec_p, T_o
      \rangle \models \psi\sigma'$.
      % formula neg predicat
    \item \emph{Case $\psi = \neg\Q(t_1, \dots, t_h)$.} In that case,
      $\psi\sigma = \Q(t_1\sigma, \dots, t_h\sigma)$, and either $\exec_p =[]$
      or $\evt_p^{\sid_p} \not= \Q(t_1\sigma, \dots, t_h\sigma)$.  In the first
      case, according to the semantics of our logic $\Logic$, we conclude that
      $\langle \exec_p, T_0\rangle \models \psi\sigma'\ (=\neg\Q(t_1\sigma',
      \dots, t_h\sigma'))$.  In the second case, i.e.  $\evt_p^{\sid_p} \not=
      \Q(t_1\sigma, \dots, t_h\sigma)$, by Definition of $\sigma'$, we have that
      $t_k\sigma' \in \{t_k\sigma, n^\epsilon_\epsilon\}$ for any $k \in \{1,
      \ldots, h\}$.  Hence, we have that $\evt_p^{\sid_p} \not= \Q(t_1\sigma',
      \dots, t_h\sigma')$, and thus $\langle \exec_p, T_0\rangle \models
      \psi\sigma'\ (=\neg\Q(t_1\sigma', \dots, t_h\sigma'))$.
      % formula learn
    \item \emph{Case $\psi = \learn(t)$.} In that case $t\sigma' \in \{t\sigma,
      n^\epsilon_\epsilon\}$ (thanks to Condition~\ref{def:attack-cond1} of
      Definition~\ref{def:attack}), then we know by hypothesis that
      $\K(\exec_p)\cup T_0 \vdash t\sigma'$ and thus, we conclude.
      % formula neg learn
    \item \emph{Case $\psi = \neg\learn(t)$.} This case cannot occur because
      $\psi$ satisfies the conditions of an attack formula
      ({see Definition~\ref{def:attack}}), and in particular no
      $\learn(u)$ appears negatively in $\psi$.
      % formula compromised
    \item \emph{Case $\psi = \Comp(t)$ or $\neg\Comp(t)$.} In that case, we have
      that $t\sigma\in\A$. By construction, we have that $t\sigma = t\sigma'$,
      and this allows us to conclude.
    \end{itemize}

    \smallskip{}

    \noindent We now distinguish several inductive cases.
    % \item[Inductive case ($|\psi|>1$)]
    \begin{itemize}
      % formula ou positif
    \item \emph{Case $\psi = \psi_1\vee \psi_2$.} Assume that $\langle \exec_p,
      T_0\rangle \models \psi_1\sigma$.  The case where $\langle \exec_p,
      T_0\rangle \not\models \psi_1\sigma$ but $\langle \exec_p, T_0\rangle
      \models \psi_2\sigma$ can be proved in a similar way.  By definition, we
      have that $\W(\exec_p, \psi_1\sigma) = \W(\exec_p, \psi\sigma) \subseteq
      S$. We can thus apply our inductive hypothesis to conclude that $\langle
      \exec_p, T_0\rangle \models \psi_1\sigma'$ and thus $\langle \exec_p,
      T_0\rangle \models \psi\sigma'\ (= \psi_1\sigma' \vee \psi_2\sigma')$.
      % formula ou negatif
    \item \emph{Case $\psi = \neg(\psi_1\vee \psi_2)$.} In that case, $\langle
      \exec_p, T_0\rangle \models \neg\psi_1\sigma$ and $\langle \exec_p,
      T_0\rangle \models \neg\psi_2\sigma$. By definition, we have that:
      \[\W(\exec_p,
      \neg\psi_1\sigma) \cup \W(\exec_p, \neg\psi_2\sigma)= \W(\exec_p,
      \psi\sigma) \subseteq S.
      \]
      By applying our inductive hypothesis, we obtain that $\langle \exec_p,
      T_0\rangle \models \neg \psi_j\sigma'$ for $j\in\{1, 2\}$. This allows us
      to conclude that $\langle \exec_p, T_0\rangle \models \psi\sigma'\ (=
      \neg(\psi_1\sigma' \vee \psi_2\sigma'))$.
      % formula diamond positif
    \item \emph{Case $\psi = \Diamond\psi'$.} In that case, according to the
      semantics of our logic $\Logic$, there exists $j\le i$ such that $\langle
      \exec_j, T_0\rangle \models \psi'\sigma$, and thus by inductive hypothesis
      we know that $\langle \exec_j, T_0\rangle \models \psi'\sigma'$. Hence, we
      have that $\langle \exec_p, T_0\rangle \models \psi\sigma'\ (=
      \Diamond\psi'\sigma')$.

      % formula diamond negatif
    \item \emph{Case $\psi = \neg\Diamond\psi'$.} In that case, according to the
      semantics of our logic, we have that $\langle \exec_{p-1}, T_0\rangle
      \models \psi\sigma$ and $\langle \exec_p, T_0\rangle \models
      \neg\psi'\sigma$. By inductive hypothesis we know that $\langle
      \exec_{p-1}, T_0\rangle \models \psi\sigma'$. Note that, by definition of
      an attack formula  ({see Definition~\ref{def:attack}}), there is no positive status event in
      $\psi\sigma$. Moreover, using our inductive hypothesis, we obtain that
      $\langle \exec_p, T_0\rangle \models \neg\psi'\sigma'$ (note that, again,
      by definition of an attack formula, we know that there is no positive
      status event in $\neg \psi'\sigma$). This allows us to conclude that
      $\langle \exec_p, T_0\rangle \models \psi\sigma'\
      =(\neg\Diamond\psi'\sigma')$.\qedhere
    \end{itemize}
    % \end{description}
\end{proof}

\noindent\usebox{\lemprojsatisfiabilitybis}

\begin{proof}
  First, we apply Lemma~\ref{lem:execterms.attack} to ensure that the
  substitution $\sigma$ witnessing the fact that the attack formula $\phi$ is
  satisfiable only uses atomic terms and subterms that occur in $\st(\exec, S)$.
  % This is not necessarily the case, in particular for the variables that do
  % not occur in a status event (\emph{e.g.} $\exists x. \learn(x)$).
  Hence, thanks to this lemma, we can assume w.l.o.g. that for all $j\in\{1,
  \dots, n\}$, $\sigma(x_j) \in\st(\exec, S)\cup\A\cup\lgAtom \cup
  \Nces_\epsilon \cup \Kcal_\epsilon$. Then, we apply
  Lemma~\ref{lem:projsatisfiability} in order to conclude.
\end{proof}

\end{document}